\documentclass[11pt, a4paper]{article}

\usepackage{booktabs}
\usepackage{graphicx}
\usepackage{amsmath}
\usepackage{authblk}
\usepackage{mathtools}
\usepackage{natbib}
\usepackage{amssymb}
\usepackage{booktabs}
\usepackage[margin=.8in]{geometry}
\usepackage{comment}
\usepackage{algorithm}
\usepackage{algorithmicx}
\usepackage{algpseudocode}
\usepackage{amsthm}
\usepackage{multirow}
\usepackage[toc,page]{appendix}
\usepackage{setspace}
\allowdisplaybreaks
\newtheorem{theorem}{Theorem}
\newtheorem{proposition}[theorem]{Proposition}%
\newtheorem{lemma}[theorem]{Lemma}
\newtheorem{example}{Example}%
\newtheorem{remark}{Remark}%
\newtheorem{definition}{Definition}

\newtheorem{model}{Model}

\newcommand{\ceil}[1]{\lceil {#1} \rceil}
\newcommand{\floor}[1]{\lfloor {#1} \rfloor}
\newcommand\norm[1]{\lVert#1\rVert}
\DeclarePairedDelimiter{\abs}{\lvert}{\rvert}
\usepackage{csquotes}
\usepackage{float}
\usepackage{longtable}
\usepackage{accents}
\usepackage{bbm}
\usepackage{hyperref}

\newcommand{\R}{\mathbb{R}}
\newcommand{\ncp}{K}
\newcommand{\E}[1]{\mathbb{E}\left[#1\right]}
\renewcommand{\hat}{\widehat}

\graphicspath{{Figures/}}

\usepackage{todonotes}
\usepackage{enumitem}
\usepackage[capitalise,noabbrev]{cleveref}
\usepackage{threeparttable}

\newenvironment{unlist}{%
  \begin{list}{}%
    {\setlength{\labelwidth}{0pt}%
     \setlength{\labelsep}{0pt}%
     \setlength{\topsep}{\medskipamount}%
     \setlength{\itemsep}{3pt}%
     \setlength{\leftmargin}{2em}%
     \setlength{\itemindent}{-2em}}}
{\end{list}}

\providecommand{\keywords}[1]
{
  \small	
  \textbf{\textbf{Keywords: }} #1
}

\begin{document}
\title{Multiscale Quantile Regression with Local Error Control}
\author{Zhi Liu \thanks{Email: zhi.liu@uni-goettingen.de}}
\author{Housen Li \thanks{Email: housen.li@mathematik.uni-goettingen.de}}
\date{27. March 2025}
\affil{\small{Institute for Mathematical Stochastics, Georg August University of G\"ottingen \protect\\
Goldschmidtstraße 7, 37077, G\"ottingen, Germany}}
\maketitle

\begin{abstract}
For robust and efficient detection of change points, we introduce a novel methodology MUSCLE (\underline{m}ultiscale q\underline{u}antile \underline{s}egmentation \underline{c}ontrolling \underline{l}ocal \underline{e}rror) that partitions serial data into multiple segments, each sharing a common quantile. It leverages multiple tests for quantile changes over different scales and locations, and variational estimation. Unlike the often adopted global error control, MUSCLE focuses on local errors defined on individual segments, significantly improving detection power in finding change points. Meanwhile, due to the built-in model complexity penalty, it enjoys the finite sample guarantee that its false discovery rate (or the expected proportion of falsely detected change points) is upper bounded by its unique tuning parameter. Further, we obtain the consistency and the localization error rates in estimating change points, under mild signal-to-noise-ratio conditions. Both match (up to log factors) the minimax optimality results in the Gaussian setup. All theories hold under the only distributional assumption of serial independence. Incorporating the wavelet tree data structure, we develop an efficient dynamic programming algorithm for computing MUSCLE. Extensive simulations as well as real data applications in electrophysiology and geophysics demonstrate its competitiveness and effectiveness. An implementation via R package \texttt{muscle} is available on GitHub.
\end{abstract}


\tableofcontents
\section{Introduction}\label{s:intro}
Detecting and localizing abrupt changes in time series data has been a vibrant area of research in statistics and related fields since its early study (e.g.\ \citealp{wald1945sequential,page1955test}). See \citet{NHZ16} and \citet{TOV20} for a comprehensive overview.  These methods find applications across a wide range of fields, such as bioinformatics \citep{Futschik14},  ecosystem ecology \citep{Mae20}, electrophysiology \citep{pein2018fully,IDC25}, epidemiology \citep{Pries20},  economics and finance \citep{russell2019breaks}, molecular biophysics \citep{fan2015identifying} and social sciences \citep{liu2013change}. The extensive applications have sparked an explosion of methods for detecting changes in the recent literature. Notable examples include the regularised maximum likelihood estimation (e.g.\ \citealp{boysen2009consistencies,HaLe10,killick2012optimal,du2016stepwise,maidstone2017optimal}), the multiscale approaches (e.g.\ \citealp{frick2014multiscale,li2016fdr,pein2017heterogeneous,mosum18}) and binary segmentation type methods (e.g.\ \citealp{vostrikova1981detecting,fryzlewicz2014wild,baranowski2019narrowest,fang2020segmentation,kovacs2023seeded}), to name only a few. 

There is, however, a relative lack of understanding on whether detected change points truly reflect genuine structural changes, arise from random fluctuations, or result from statistical modelling errors. The distinction poses a delicate challenge in practical applications, necessitating an in-depth grasp of the data as well as the employed change point detection methods. In pursuit of an automatic and versatile solution, we introduce here a \emph{robust} change point estimator, named MUSCLE (\underline{m}ultiscale q\underline{u}antile \underline{s}egmentation \underline{c}ontrolling \underline{l}ocal \underline{e}rror; see~\cref{S: 2}), that focuses on the recovery of genuine structural changes while maintaining high detection power as if a strong statistical model were faithfully available. Importantly, we adopt a quantile segmentation perspective and focus on local error control instead of often considered global error control, as explained in detail below. 

\subsection{Quantile segmentation model}

Following \citet{jula2022multiscale} and \citet{fryzlewicz2024robust}, we inspect the underlying structural changes in terms of changes in quantiles. 

\begin{model}[Quantile segmentation]\label{QSR}
Let $\beta \in (0,1)$ and ${f}:[0,1) \to \R$ be a piecewise constant function
\[
{f}  = \sum_{k=0}^{\ncp} \theta_k \mathbbm{1}_{[\tau_k,\tau_{k+1})}\quad\text{with } 0=\tau_0 <\cdots <\tau_{\ncp+1} = 1
\]
where the number $\ncp$ and the locations $\tau_k$ of change points, and the {segment values} $\theta_k$ are all \emph{unknown}. Assume that  $\theta_k\neq \theta_{k+1}$ and the observations $Z = (Z_1,\dots,Z_{n})$ are \emph{independent} random variables with their $\beta$-quantiles equal to ${f}(x_i)$, i.e.,
\begin{equation*}
\mathbb{P}\bigl(Z_i\leq {f}(x_i)\bigr) = \beta \quad\text{ or }\quad F_{i}\bigl({f}(x_i)\bigr) = \beta
\end{equation*}
where $x_i = (i-1)/n$, $i = 1,\ldots, n$, are equidistant sampling points, and $F_{i}$ the distribution function of $Z_i$.
\end{model}

Equivalently, Model~\ref{QSR} can be written into a typical form of additive noise model
\begin{equation}\label{e:noise_signal}
Z_i = {f}\left(x_i\right) +\varepsilon_i,\quad \quad i = 1, \dots, n,   
\end{equation}
where noise terms $\varepsilon_i$ are independent with their $\beta$-quantiles equal to zero. For instance,  the choice of $\beta = 0.5$ corresponds to the median segmentation regression. The subscript $\beta$ will be suppressed whenever there is no ambiguity. Since our aim is to provide a $\beta$-quantile segmentation, with an arbitrary yet fixed $\beta\in(0,1)$, we do not consider $\beta$ as a {tuning} parameter.  

We emphasize that no distributional assumptions are imposed on the observations~$Z_i$, or equivalently on the noise~$\varepsilon_i$, beyond statistical independence. 
This framework accommodates arbitrarily heavy-tailed, skewed, or heterogeneous noise distributions, as well as outliers or mixtures thereof (see~\cref{S: 5.1.1} for examples). 
Moreover, the noise~$\varepsilon_i$ may follow discrete distributions, provided that $\mathbb{P}(\varepsilon_i \le 0) = \beta$.  When multiple quantile levels~$\beta$ are considered in Model~\ref{QSR}, the proposed method can further detect distributional changes in the observations (see~\cref{ss:cp_box} and~\cref{f:MMUSCLE}). 
Finally, if we focus on a single quantile level, the assumption of independence among the observations~$Z_i$ can be relaxed to the independence of the truncated random variables~$\mathbbm{1}_{\{Z_i \le f(x_i)\}} \equiv \mathbbm{1}_{\{\varepsilon_i \le 0\}}$.

\subsection{Local error control}\label{S: local error control}
As a motivating example, we consider a specific instance of the additive noise model~\eqref{e:noise_signal} with sample size $n = 400$. The signal is defined as 
$$
{f} = 2\cdot\mathbbm{1}_{[0.05,0.1)} -2 \cdot \mathbbm{1}_{[0.6,0.725)}{+ 1\cdot \mathbbm{1}_{[0.725,1)}}
$$ 
and additive errors  $\varepsilon_i$ are centered (i.e.\ $\beta = 0.5$) with $\varepsilon_i \sim \mathcal{N}(0,1)$ if $1\leq i\leq 100$ and $\varepsilon_i \sim t_3/(2\sqrt{3})$ if $101\leq i\leq 400$. Here and below $\mathcal{N}(\mu, \sigma^2)$ denotes the normal distribution with mean $\mu$ and variance $\sigma^2$, and $t_m$ the Student's~$t$ distribution of $m$ degrees of freedom. To illustrate the impact of data windowing, we apply the proposed MUSCLE method and several recent robust or nonparametric segmentation approaches (\citealp{haynes2017computationally,baranowski2019narrowest,fearnhead2019changepoint,madrid2021optimal,jula2022multiscale,fryzlewicz2024robust}) to two subsets of the data: the first 80 observations and the full set of 400 observations, see~\cref{f: E1}. All methods, except MUSCLE, produce markedly different estimates (shown as solid and dashed lines) of the signal $f$ on the interval $[0,\, 80/n)$ when applied to these two data windows. This reflects the often unspoken fact that the choice of data window acts as a hyper-parameter in real world data applications (cf.\ \citealp{ESL23}). The reason behind is that many data analysis methods are globally tuned to the specific data window they are applied to.

\begin{figure}[H]
\centering
\includegraphics[width=\linewidth]{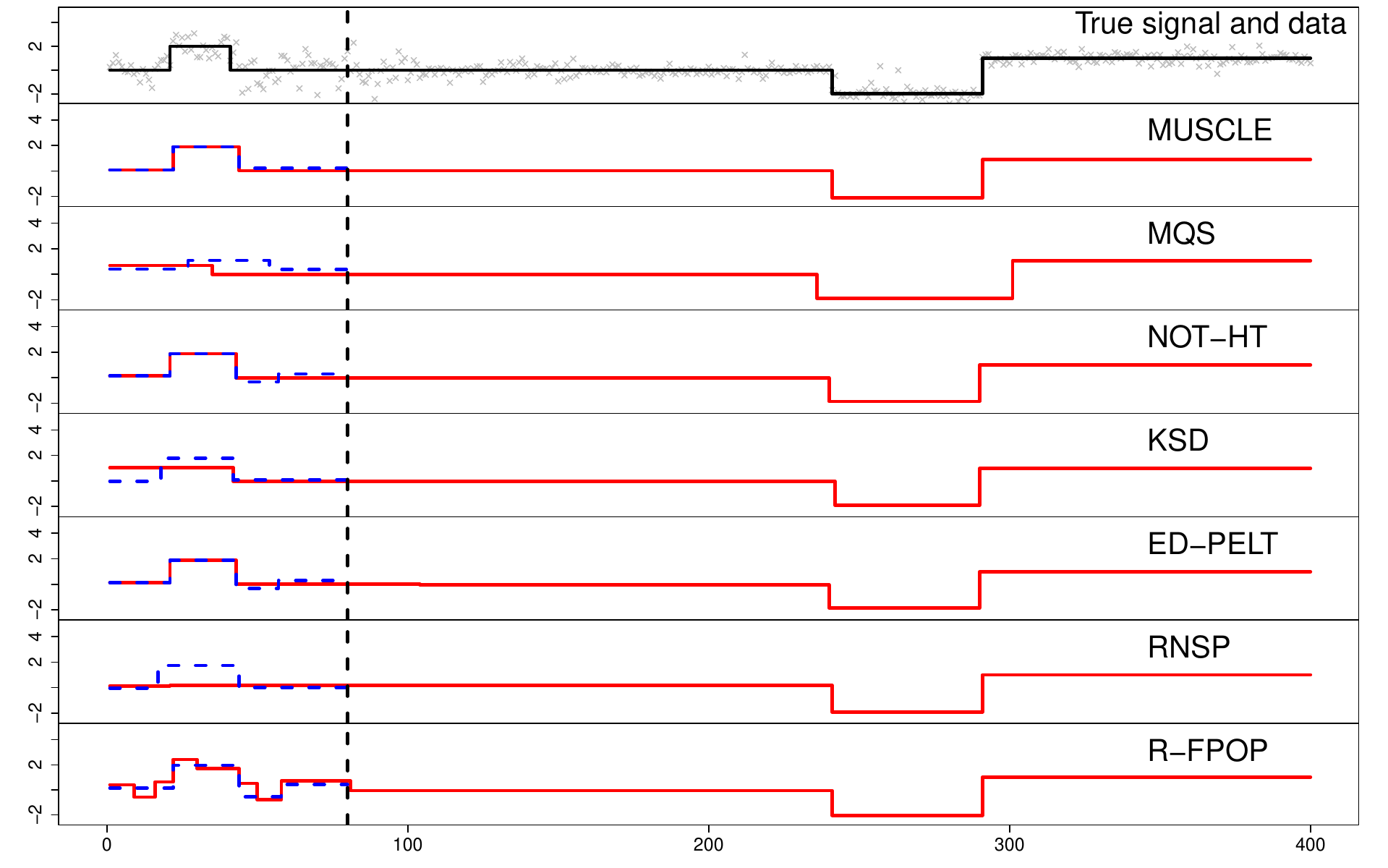}
\caption{Influence of data windowing on segmentation methods. The top panel shows the true signal (line) and data (crosses). The lower ones display the segmentation results using the first 80 observations (dashed line) and using the whole 400 observations (solid line) for the proposed MUSCLE, MQS~\citep{jula2022multiscale}, NOT-HT~\citep[the heavy tailed version]{baranowski2019narrowest}, KSD~\citep{madrid2021optimal}, ED-PELT~\citep{haynes2017computationally}, RNSP~\citep{fryzlewicz2024robust} and R-FPOP~\citep{fearnhead2019changepoint}. The vertical dashed line marks the 80th observation.}\label{f: E1}
\end{figure}

The distinctive feature of MUSCLE is its stability across different data windows, which primarily stems from  its control of local errors on individual segments. Specifically, it controls the probability of falsely detecting change points across all local intervals within every estimated segment. This stability not only simplifies real-data analysis in practical applications, but also has important implications for computational efficiency and statistical power. For instance, the data can be split into equally sized segments, MUSCLE can be applied to each segment independently, and the results can then be merged to obtain a final estimate. This splitting-merging approach yields a computation time that scales linearly with the sample size $n$, and enables a linear speedup through parallelization with multiple computing units (see~\cref{S: implementation} and, in particular, \cref{r: compleixity}). Besides, compared to global error control, local error control, being a weaker criterion, can enhance the power to detect change points, although it may lead to an increased risk of false positives. MUSCLE addresses this by incorporating a built-in penalty on model complexity, with false positive control governed by a single tuning parameter. Further, local error control is especially well suited for change point detection, as a change point is inherently a local feature occurring between two adjacent segments.

\subsection{Our contribution and related work}
Under the quantile segmentation model (i.e., Model~\ref{QSR}), the proposed MUSCLE scans over intervals of varying scales and locations in search for  changes in quantile, and estimates the underlying quantile function in the framework of variational multiscale estimation \citep{HaLiMu22}. It penalizes the model complexity in terms of the number of change points, under a convex polytope constraint that builds upon a meticulously crafted multiscale statistics towards local error control on each individual segment. Due to this construction, MUSCLE is robust and has a high detection power at the same time, as supported by numerous simulations and real data analyses (see~\cref{S: 4}). Moreover, we showcase two extensions of MUSCLE to the detection of changes in distribution (M-MUSCLE; \cref{ss:cp_box}) and to a dependent scenario (D-MUSCLE; \cref{S: 5.2.2}).

Let $\hat{f}$ be the MUSCLE estimator and $\hat{\ncp}$ its number of estimated change points, and assume Model~\ref{QSR}. We have obtained the finite sample guarantee (\cref{t: FDR}) on falsely detected change points
$$
\E{\tfrac{\max\{\hat{\ncp} - \ncp, 0\}}{\max\{\hat{\ncp},1\}}}\; \le\; \alpha
$$
where $\alpha \in (0,1)$ is the only tuning parameter of MUSCLE. This bound, together with the fact that a larger $\alpha$ often leads to a larger $\hat{\ncp}$, allows data analysts to tune MUSCLE according to their needs in real data applications. 

Besides, we have derived consistency (\cref{t: consistency 1}) and bounds on localization errors (\cref{t: Hausdorff}) for MUSCLE in estimating change points, under the detectability condition: 
\begin{equation}\label{e:ms_cond}
 \min_{i \in\{k,k+1\}}\{\tau_i - \tau_{i-1}\}\cdot \min_{i \in[\tau_{k-1}, \tau_{k+1})} \left|F_{i}\left(\tfrac{\theta_k + \theta_{k-1}}{2}\right) - \beta \right|^2
 \ge \frac{C\log n}{n},\qquad k = 1,\ldots, K,    
\end{equation}
for some large enough constant~$C$. These results match, up to log factors, the results of minimax optimality in the Gaussian setup \citep{verzelen2023optimal}. Further, our consistency result follows from the stronger results of exponential bounds on under- and overestimation errors (\cref{t: under,t: over}).  

Computationally, we have developed an efficient dynamic programming algorithm for MUSCLE, which has a worst-case computational complexity of $O(n^3 \log^2 n)$ and requires $O(n \log n)$ memory (see~\cref{S: implementation}). The algorithm's key components include the wavelet tree \citep{Nav14} data structure and the use of a sparse collection of intervals at different scales and locations \citep{Walther2022Calibrating}. In contrast, the commonly used linear data structure would increase the computational complexity by a factor of $O(n)$, and the double heap structure, as used in \citet{jula2022multiscale}, would result in a worst-case memory requirement of $O(n^2)$. The use of all intervals would also lead to an additional factor of $O(n)$ in the computational complexity.  

The proposed MUSCLE belongs to the family of multiscale approaches (initiated by \citealp{frick2014multiscale}), among which the most closely related methods are FDRSeg~\citep{li2016fdr} and MQS~\citep{jula2022multiscale}. MUSCLE extends FDRSeg, designed for the standard Gaussian noise, to handle general independent noise, and significantly enhances the detection power of MQS by switching to local error control. Further, the statistical results in this paper are sharper and more general compared to the aforementioned two papers. More precisely, we have improved the upper bound of false discovery rate in \citet{li2016fdr} from $2\alpha/(1 - \alpha)$ to $\alpha$, while also admitting of general distributions beyond Gaussian, by employing a completely different proof technique (see~\cref{S: 3.2}). In comparison with~\citet{jula2022multiscale}, we have refined the overestimation bound on the number of change points (\cref{t: over}), derived a quantitative estimate of the tail probability of multiscale statistics (Lemma~\ref{l: ub q_n} in Appendix~\ref{A1}), and reformulated the detection condition into a weaker form~\eqref{e:ms_cond}, which allows frequent large jumps and rare small jumps over long segments to occur at the same time. These improvements are mainly due to the use of strengthened technical inequalities (e.g., Talagrand's concentration inequality for convex distances) and a subtle adjustment of multiscale statistics (see footnote~\ref{f:tech} on page~\pageref{f:tech}). In addition, the optimization problem defining MUSCLE poses notable challenges compared to FDRSeg and MQS. However, by incorporating the wavelet tree data structure, we are able to achieve nearly (up to a log factor) the same worst-case computational complexity as FDRSeg. 

Further related works include robust segmentation methods. \citet{baranowski2019narrowest}, though primarily considered Gaussian setups, suggested a heavy tail extension (NOT-HT) that truncates the data into $\pm 1$ by a certain sign transform and then treats the truncated data as if they were Gaussian distributed. \citet{fryzlewicz2024robust} further developed this idea, and introduced Robust Narrowest Significance Pursuit (RNSP), which utilizes a sign-multiresolution norm, and focuses on changes in medians. From a penalized M-estimation perspective, \citet{fearnhead2019changepoint} proposed R-FPOP employing Tukey's biweight loss, and derived asymptotic properties for a restrictive scenario with fixed change points. Another line of related works are nonparametric change point methods, which assume that the observations are piecewise i.i.d.\ (a  stronger model than Model~\ref{QSR}; cf.\ \citealp{dumbgen91}). Recently, \citet{MaJa14} introduced a hierarchical clustering-type method based on the energy distance between probability measures, and proved consistency results (with no rates) under additional moment conditions. Utilizing a weighted Kullback--Leibler divergence, \citet{zou2014nonparametric} proposed Nonparametric Multiple Change-point Detection (NMCD), together with consistency and localization rates of change points. As a computational speedup variant of NMCD, \citet{haynes2017computationally} introduced ED-PELT by simplifying the involved cost function through discrete approximation. Besides, \cite{madrid2021optimal} developed a CUSUM procedure (KSD) based on the Kolmogorov--Smirnov distance, and established better statistical guarantees than~\citet{zou2014nonparametric}, but requiring prior knowledge on the minimum segment length. We include a large selection of the above methods in our comparison study on simulated and real world data (see \cref{S: 4}).

\subsection{Outline and notation}
\cref{S: 2} introduces MUSCLE, a novel methodology for multiple change point segmentation under Model~\ref{QSR}. \cref{S: 3} delves into the statistical theories of MUSCLE, including consistency, localization error rates of change points and finite sample guarantees on falsely detected change points. \cref{S: implementation} discusses the computation of MUSCLE. \cref{S: 4} presents simulation studies and real data applications. The paper concludes with a discussion in \cref{S: 6}. Additional material and proofs are given in the Appendix. The {R} package \texttt{muscle} implementing MUSCLE and its variants is available on GitHub via \url{https://github.com/liuzhi1993/muscle}. 

By convention, all intervals in this paper are of form $[a,b)$, with $a,b \in \R$ and $a < b$. For sequences $\{a_m\}$ and $\{b_m\}$ of positive numbers, we write $a_m \lesssim b_m$  if $a_m \le C b_m$ for some finite constant $C > 0$. If $a_m \lesssim b_m$ and $b_m \lesssim a_m$, we write $a_m \asymp b_m$. We write $X \overset{\mathcal{D}}{=}Y$ to indicate that random variables $X$ and $Y$ have the same distribution. We write $(x)_+ = \max\{x, 0\}$ for $x\in\R$. By $\bigsqcup$ we denote the disjoint union. By $\# S$ we denote the cardinality of a set $S$.

\section{Methodology}\label{S: 2}

\subsection{Multiscale test on a segment}\label{ss:mtest}
Assume that observations $Z_1,\ldots,Z_n $ are from Model~\ref{QSR} {with $\beta \in (0,1)$}, and let interval $I \in \{[\tau_{k}, \tau_{k+1}) : k = 0, \ldots, \ncp\} $ be an arbitrary segment of the quantile function $f$. It follows that $\{Z_i : x_i \equiv (i-1)/n \in I \}$ are independent and have the common $\beta$-quantile, whereas their distributions $F_i$ may not be identical. Towards a uniform treatment, we transform the data into simple binary samples as 
\begin{equation*}
W_i \coloneqq W_{i}(Z_i,\theta) \coloneqq \mathbbm{1}_{\{Z_i\leq  \theta\}}, \quad i=1,\dots,n.
\end{equation*}
If $\theta$ is equal to the common $\beta$-quantile of $\{Z_i : x_i \in I \}$, the transformed samples $\{W_i : x_i \in I \}$ are i.i.d.~Bernoulli random variables with mean $\beta$. The condition of $\theta$ equal to the common quantile can be verified using the data, when viewed as a hypothesis testing problem. We tackle the testing problem by multiscale procedures, which build on a delicate combination of  multiple simple tests, and are known to be optimal in the minimax sense (see e.g.\ \citealp{dumbgen2001multiscale,frick2014multiscale}). Specifically, we consider the testing problems, for intervals $J \subseteq I$, 
\begin{align*}
    & H_{J}: \mathbb{P}(Z_i\leq \theta) = \beta,\quad  x_i \in J, \\
    \text{versus}\quad & A_{J}: \mathbb{P}(Z_i\leq \theta) \neq \beta, \quad x_i \in J.
\end{align*}
For each of them, we employ the likelihood ratio test statistic
\begin{equation*}
    L_J(Z,\theta) =  \abs{J}\left[\overline{W}_J \log \left(\tfrac{\overline{W}_J}{\beta}\right)+(1-\overline{W}_J)\log \left(\tfrac{1-\overline{W}_J}{1-\beta}\right)\right],
\end{equation*}
where $\abs{J} = \#\{i : x_i \in J\}$ and $\overline{W}_J = \left(\sum_{i\in J}W_i\right)/\abs{J}$. {Here, we use the convention $0\log 0 = 0$.} The null hypothesis $H_J$ is {rejected} if $L_J$ exceeds a certain critical value. As in~\citet{frick2014multiscale}, we incorporate all of such tests with the control of family-wise error, and balance the detection powers over different scales of $\abs{J}$. This leads to the \emph{multiscale statistic} with scale penalty 
\begin{equation}\label{e: T_I}
T_{I}(Z,\theta)\; = \; \underset{J\subseteq {I}}{\max}\sqrt{2L_J (Z,\theta)} - \sqrt{2\log\left(\tfrac{e\abs{I} }{\abs{J}}\right)}.
\end{equation}
To calibrate the family-wise error rate, define the critical value
\begin{equation}\label{e: q_m}
    q_{\alpha}(\abs{I})  = \min\left\{q :  \mathbb{P}\bigl(T_{I}(Z,\theta) > q \bigm\vert H_J, J \subseteq I\bigr)\leq \alpha \right\},
\end{equation}
which is also referred to as a local quantile on the segment $I$. Note that $\{W_{i}(Z_i,\theta)\,:\,x_i \in I\}$ are i.i.d.~Bernoulli distributed with mean $\beta$, under all $H_J$ for $J \subseteq I$. Thus, the distribution of $T_{I}(Z,\theta)$ conditioned on all $H_J$ for $J \subseteq I$ is known and equal to the maximum of dependent and transformed Bernoulli distributions. A simple and explicit form of $q_{\alpha}(\abs{I})$ is not yet available, but we can estimate it by Monte--Carlo simulations. Thus, the decision problem of whether $\{Z_i : x_i \in I \}$ have the same $\beta$-quantile equal to $\theta$ can be addressed by the multiscale test $\mathbbm{1}_{\{T_I(Z, \theta) > q_{\alpha}(\abs{I})\}}$ with the type I error no more than $\alpha$. 

\begin{remark}[Multiscale statistics]\label{r: penality and quantile}
The scale penalty term $\sqrt{2\log (e\abs{I}/\abs{J}})$ in \eqref{e: T_I} is motivated by the observation that $\max_{J: \abs{J} = m} \sqrt{2 L_J (Z, \theta)}$ asymptotically concentrates around $\sqrt{2 \log(\abs{I}/m)}$, see e.g.\ \citet{kabluchko2011extremes}. This penalty not only standardises the influence of scale $\abs{J}$ but also ensures the asymptotic tightness of $T_I$ as $\abs{I} \to \infty$. The form of scale penalty is though not unique, see e.g.\  \citet{Walther2022Calibrating} for alternative forms. Moreover, the distribution of $T_{I}(Z, \theta)$ conditioned on $\cap_{J\subseteq I} H_J$ converges weakly to an almost surely finite random variable as $\abs{I} \to \infty$ under a mild restriction on the smallest scale, namely, $\abs{J}\gg \log^3 \abs{I}$, see e.g.\ \cite{frick2014multiscale} and \citet{dumbgen2001multiscale}. It follows that  $\limsup_{\abs{I}\to \infty}q_{\alpha}(\abs{I}) <\infty$ for any fixed $\alpha \in (0,1)$. Further, we derive an explicit upper bound of the form $\limsup_{\abs{I}\to \infty}q_{\alpha}(\abs{I}) \lesssim \sqrt{-\log \alpha}$, see Lemma~\ref{l: ub q_n} in Appendix~\ref{A1}. This bound may serve as an estimate of $q_{\alpha}(\abs{I})$ particularly when $\abs{I}$ is so large that Monte--Carlo simulations become computationally unaffordable. Alternatively, we might be able to approximate $q_{\alpha}(\abs{I})$ in a deterministic manner, following ideas in \cite{fang2020segmentation}. Their results, however, do not directly apply due to the presence of the scale penalty. The exploration of this possibility is an interesting avenue of future research.
\end{remark}

\subsection{The proposed MUSCLE}
A reasonable candidate $g$ for the recovery of the quantile function $f$ in Model~\ref{QSR} should pass the multiscale tests on individual segments. Thus, we restrict ourselves to functions lying in the multiscale side-constraint\footnote{In multiscale segmentation methods (e.g.\ \citealp{frick2014multiscale,li2016fdr,pein2017heterogeneous,LiGM19,jula2022multiscale}), the side constraint takes the form of $T_I(Z, \theta) - q_{\alpha}(\abs{I}) \le 0$, whereas in~\eqref{e: C_k} we use $\mathring{I}$ instead of $I$. This subtle adjustment ensures the independence of test statistics on adjacent segments and allows improved bounds on the distributional behaviors. For instance, in later \cref{t: over}, with such an adjustment, we can improve the upper bound of probability $\mathbbm{P}(\widehat{\ncp}\geq \ncp +k)$ from $\alpha^{\floor{k/2}}$ to $\alpha^k$ for all $k\in \mathbbm{N}$.\label{f:tech}}
\begin{equation}\label{e: C_k}
    \mathcal{C}_k \coloneqq \biggl\{g= \sum_{j=0}^{k} \theta_j \mathbbm{1}_{{I}_j} \;:\;  \bigsqcup_{j = 0}^k I_j = [0,1), \,
    T_{\mathring{I}_j}(Z,\theta_j) - q_{\alpha}(\abs{\mathring{I}_j}) \leq 0, \, j=0,\dots,k\biggr\},
\end{equation}
where $\mathring{I}$ denotes the interior of interval $I$. Note that in this particular form of $\mathcal{C}_k$ we control the type I error \emph{locally} on every segment $I_i$. Further, to avoid overfitting to the data, we minimize the complexity of candidates, which is measured by the number of change points, and define 
\begin{equation}\label{e: hat_K}
    \widehat{\ncp} \coloneqq \min\bigl\{k \in \mathbbm{N} \cup \{0\}\, : \,\mathcal{C}_k \neq \emptyset\bigr\}.
\end{equation}
We introduce the \underline{m}ultiscale q\underline{u}antile \underline{s}egmentation \underline{c}ontrolling \underline{l}ocal \underline{e}rrors (MUSCLE) as the candidate in $C_{\hat{K}}$ that fits best to the data, namely, 
\begin{equation}\label{e: MUSCLE}
    \hat{f} \; \in \; \underset{g\in \mathcal{C}_{\widehat{\ncp}}}{\arg \min}\,\sum_{i = 1}^{n} \bigl(Z_i-g(x_i)\bigr)\bigl(\beta-\mathbbm{1}_{\{Z_i\le g(x_i)\}}\bigr).
\end{equation}
The fitness to the data is quantified with respect to the {asymmetric absolute deviation loss} (a.k.a.~check loss), which is typically employed in defining quantiles (see e.g.\ \citealp{koenker_2005}). 

Note that $\tilde{f} = \sum_{i=1}^{n} Z_i \mathbbm{1}_{[x_{i},x_{i+1})}$ clearly lies in $\mathcal{C}_{n-1}$, so $\hat{\ncp}$ is finite and takes values in $\{0, 1, \ldots, n-1\}$. Thus $\mathcal{C}_{\hat{K}} \neq \emptyset$ and consequently a solution $\hat{f}$ of \eqref{e: MUSCLE} always exists. Owing to the specific structure of the optimization problem, $\hat{f}$ can be efficiently computed in the dynamic programming framework (see \cref{S: implementation}).

\subsection{Comparison with global error control}
A distinctive feature of MUSCLE is the control of local errors on individual segments, which is built in the construction of the multiscale side-constraint $\mathcal{C}_k$ in \eqref{e: C_k}. To further understand the strength of local error control, we compare MUSCLE with MQS (Multiscale Quantile Segmentation; \citealp{jula2022multiscale}), a multiscale method but with the global error control. The distinction between these two methods lies in the construction of multiscale side-constraint. MQS searches for a candidate with minimum number of change points and minimum asymmetric absolute deviation loss over 
\begin{equation*}
    \mathcal{C}_k^0 \coloneqq \biggl\{g= \sum_{j=0}^{k} \theta_j \mathbbm{1}_{I_j} \;: \; \bigsqcup_{j = 0}^k I_j = [0,1)\text{ and }{T}^0_{I_j}(Z,\theta_j) - q_{\tilde{\alpha}}(n) \leq 0, \, j=0,\dots,k \biggr\},
\end{equation*}
where $T_I^0(Z,\theta)\coloneqq {\max}_{J \subseteq I}\sqrt{2L_J (W,\theta)} - \sqrt{2\log\left({en}/{\abs{J}}\right)}$ and $\tilde{\alpha} \in (0,1)$ is the tuning parameter. Compared to \eqref{e: C_k}, the difference in quantiles and penalties arises from the control of different errors. The global error control of MQS ensures that the true quantile function $f$ lies in $\mathcal{C}_K^0$ with probability no less than $1-\tilde\alpha$. By contrast, the local error control of MUSCLE guarantees that each segment of the true quantile function $f$ satisfies the side-constraint with probability no less than $1 - \alpha$, which further implies that $f$ lies in $\mathcal{C}_K$ with probability no less than~$(1 - \alpha)^{K+1}$. Then, unless $\tilde{\alpha} \ge 1 - (1-\alpha)^{K+1}\approx (K+1)\alpha$, the multiscale side-constraint $\mathcal{C}_K$ is generally smaller than~$\mathcal{C}^0_K$. As a consequence, MUSCLE is likely to report more change points than MQS. Thus, the local error control enables MUSCLE to overcome the conservativeness (i.e., the tendency to miss true change points) that is often observed for multiscale segmentation methods. 

As a simulation study, we consider the recovery of the teeth function that contains $K = 80$ change points under the additive noise model~\eqref{e:noise_signal}, where the random errors are scaled Student's $t_3$-distributed with standard deviation equal to one. The simulation results are given in \cref{f: MQS MUSCLE} and \cref{t: MQSE MUSCLE}. Clearly, MUSCLE ($\alpha = 0.3$) recovers nicely the number and the locations of change points, while MQS ($\tilde{\alpha}= 0.3$) severely underestimates the number of change points.  To enhance the detection power of MQS, we adjust $\tilde{\alpha} = 1- (1-0.3)^{81} \approx 0.999$. This adjustment significantly improves MQS in the estimation of $\ncp$, bringing its performance closer to MUSCLE, but at the expense of nearly no statistical guarantee. Still, MQS with such a large $\tilde{\alpha}$ underperforms MUSCLE in estimating the change point locations as well as the whole signal. See \cref{S: 4.1} (and also \cref{S: local error control}) for further comparisons. 

\begin{figure}[ht]
        \centering
        \includegraphics[width=\linewidth]{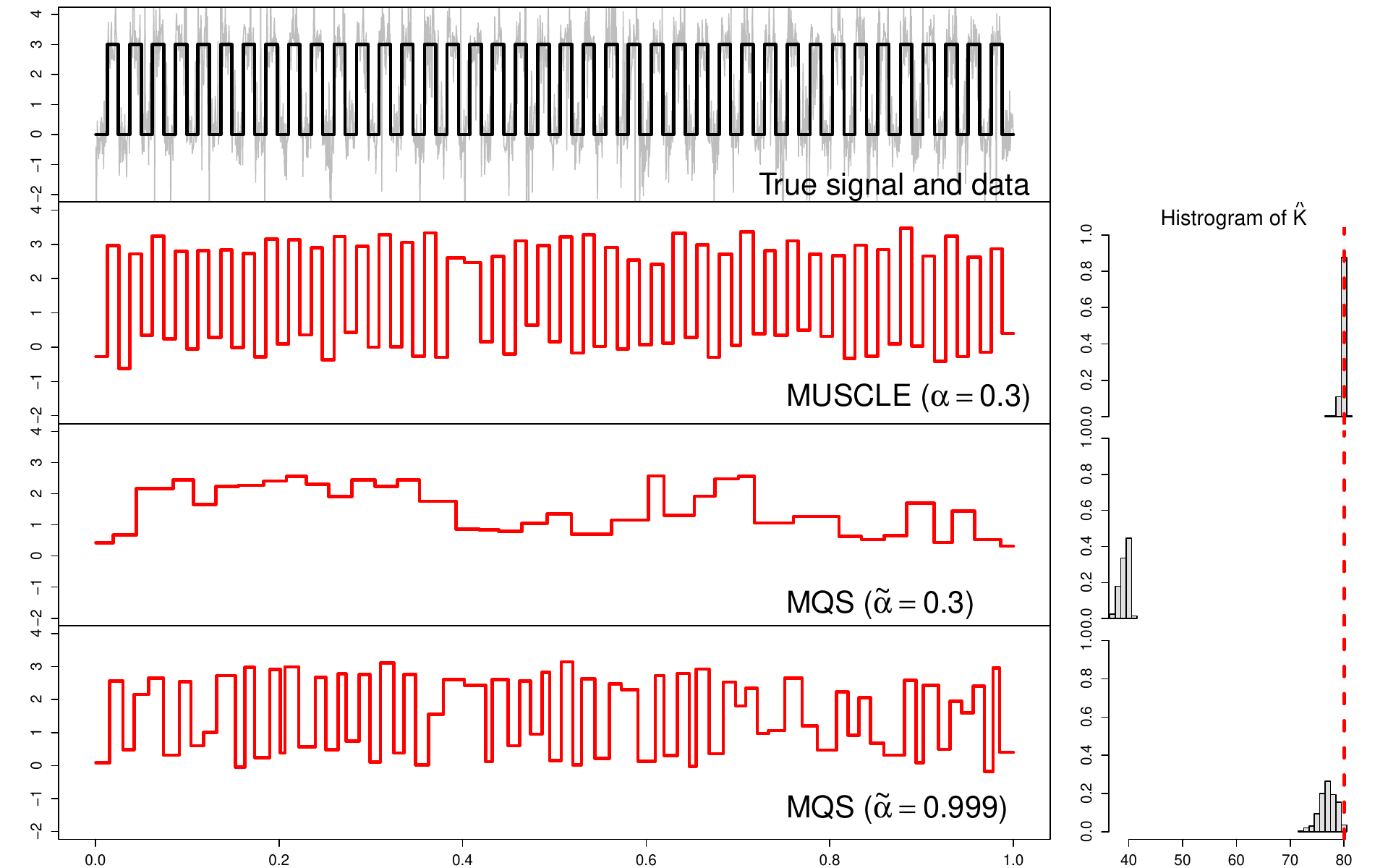}
        \caption{Comparison of the proposed MUSCLE and MQS \citep{jula2022multiscale}.  The true signal (black line), which is a teeth function, and data (gray line; $2000$ samples) are shown in the top panel. The estimates by MUSCLE ($\alpha = 0.3$), MQS ($\tilde{\alpha} = 0.3$) and MQS ($\tilde{\alpha} = 0.999$) are plotted in lower panels. The performance of each method, over $200$ repetitions, in estimating the number of change points is visualised by histograms on the right hand side. The true number of change points $K = 80$ is marked by a vertical dashed line.}
        \label{f: MQS MUSCLE}
\end{figure}

\begin{table}[htbp]
\centering
\caption{Comparison of MQS and MUSCLE in the setup of \cref{f: MQS MUSCLE}. For each criterion, the average value and the standard deviation (in brackets), over 200 repetitions, are presented. Here $d_H$ is the Hausdorff distance between estimated and true change points, MISE the mean integrated squared error of $\hat{f}$ and MIAE the mean integrated absolute error. Their formal definitions are given in \eqref{e: Hausdorff distance}, \eqref{e: MISE}  and \eqref{e: MIAE} in Appendix~\ref{A2}, respectively.}\label{t: MQSE MUSCLE}
\begin{tabular}{c|c|c|c}
\toprule
   &  MUSCLE ($\alpha = 0.3$) & MQS ($\tilde{\alpha} = 0.3$)  &  MQS  ($\tilde{\alpha} = 0.999$)      \\
\midrule
$\widehat{\ncp}$  & $80\;(0.4)$ & $40\;(0.5)$ & $77\;(1.6)$  \\
$d_{H}$ & $0.001\;(0.0030)$ & $0.015\;(0.0037)$& $0.008\;(0.0006)$\\
MISE & $0.1\;(0.041)$ & $2.7\;(0.065)$ & $1.8\;(0.108)$ \\
MIAE &  $0.2\;(0.020)$ &$1.4\;(0.008)$ & $0.9\;(0.057)$\\
\bottomrule
\end{tabular}
\end{table}

\section{Statistical Properties}\label{S: 3}
Here we provide statistical guarantees on the proposed MUSCLE under Model~\ref{QSR}. In all of the theoretical results, we allow the quantile function $f$ as well as distribution functions $F_i$ in Model~\ref{QSR} to change as sample size varies. We use subscript $n$ to emphasise such a dependence whenever needed. 

\subsection{Consistency and localization rates}
\label{S: 3.1}

Following \citet{jula2022multiscale}, we measure the changes in quantiles by the \emph{quantile jump function}. 
\begin{definition}[\!\!\citealp{jula2022multiscale}]
Let $F$ be a distribution function and $\beta \in (0,1)$. The $\beta$-quantile of $F$ is given by $F^{-1}(\beta) \coloneqq \inf\{x : F(x)\geq \beta\}.$
The \emph{quantile jump function} is defined by 
\begin{equation*}
    \xi_{F,\beta}(\delta) \coloneqq \left|F\bigl(F^{-1}(\beta)+\delta\bigr)-\beta\right |,\qquad\text{for}\quad \delta \in \R.
\end{equation*}
\end{definition}

\begin{example}[Distribution with continuous density]\label{ex: quantile jump}
    Let $F$ be a distribution function that is continuously differentiable. Then for $\beta\in (0,1)$ we have
        $\xi_{F, \beta}(\delta) = \abs{F(F^{-1}(\beta)+\delta)-\beta} = \abs{F(F^{-1}(\beta)+\delta)-F(F^{-1}(\beta))} = F'\bigl(F^{-1}(\beta)\bigr)\delta +o(\delta).$
    In particular, in the Gaussian model with changing means, the quantile jump function with $\beta = 0.5$ satisfies 
    \begin{equation*}
        \xi_{\Phi, 0.5}(\delta)\; \approx\;  \varphi(0)\delta\; = \;\frac{1}{\sqrt{2\pi}} \delta\; \approx \; \frac{2}{5}\delta,
    \end{equation*}
    where $\Phi$ and $\varphi$ denote the distribution function and the density of \emph{standard} normal distribution, respectively.
\end{example}
Recall the notation in Model~\ref{QSR}. For $k = 1,\dots,K$, let $\lambda_k = \min\{\tau_k-\tau_{k-1}, \tau_{k+1}-\tau_k\}$ and $I_k = (\tau_k-\lambda_k/2,\tau_k+\lambda_k/2)$, and define the minimum and maximum of $F_i$ on $I_k$ as
$$
F_{\min,k} = \underset{i: x_i \in I_k}{\min} F_i\quad\text{and}\quad
    F_{\max,k} = \underset{i : x_i\in I_k}{\max} F_i.
$$
For $k = 1,\ldots, K$, let $\delta_k = \theta_k - \theta_{k-1}$ and refer to its absolute value as \emph{jump size}. We use
\[ \xi_k = \min\left\{\xi_{F_{\min,k},\beta}(\delta_k/2),\xi_{F_{\max,k},\beta}(-\delta_k/2) \right\}\]
to characterize the quantile jump at $\tau_k$. In case of $F_i$ being continuously differentiable, $\xi_k$ is equivalent to $\abs{\delta_k}$ up to a constant if $\abs{\delta_k}$ is small, see \cref{ex: quantile jump}. Further, we define
\begin{equation}\label{e: LTO}
        \Lambda \coloneqq \underset{k=1,\dots,\ncp}{\min}\abs{\tau_k-\tau_{k-1}},\quad
        \Theta \coloneqq \underset{k=1,\dots,\ncp}{\min} \xi_k\; 
        \text{ and }\; 
        \Omega  \coloneqq \underset{k=1,\dots,\ncp}{\min} \sqrt{\lambda_k} \xi_k.
\end{equation}
Here $\Lambda$ is the minimal segment length, $\Theta$ is the minimal quantile jump, and $\Omega$ can be interpreted as the \emph{signal-to-noise ratio}, as it measures the difficulty in detection of quantile changes. It follows immediately that $\sqrt{\Lambda}\Theta  \leq \Omega$. 

The following exponential bound of underestimation probability ensures the detection power of MUSCLE.
\begin{theorem}[Underestimation bound]\label{t: under}
Assume Model~\ref{QSR} and let $\widehat{\ncp}$ in \eqref{e: hat_K} be the estimated number of change-points by MUSCLE with $\alpha \in (0,1)$. Then it holds
\begin{equation*}
    \mathbb{P}(\widehat{\ncp}\geq \ncp) \;\geq\; \prod_{k=1}^\ncp \gamma_{n,k}^2,
\end{equation*}
where 
\[
\gamma_{n,k} = 1-2\exp{\left(-n\lambda_k\xi_k^2\right)} -2\exp{\left[-\Bigl( \sqrt{n\lambda_k}\xi_k-\tfrac{1}{\sqrt{2}}\max_{m\leq n}q_{\alpha}(m)-\sqrt{\log\tfrac{2e}{\lambda_k}}\Bigr)_+^2 \right]}.
\]
In particular, it implies with \eqref{e: LTO} that 
\begin{equation}\label{e: under}
    \mathbb{P}(\widehat{\ncp}<\ncp)   
    \leq  4\ncp\Bigg\{\exp \left(-n\Omega^2\right) + 
    \exp\left[-\Bigl(\sqrt{n}\Omega -\tfrac{1}{\sqrt{2}}\max_{m\leq n}q_{\alpha}(m)-\sqrt{\log\tfrac{2e}{\Lambda}}\Bigr)_+^2\right]
      \Bigg\}.
\end{equation}
\end{theorem}

Although it controls local errors on individual segments, MUSCLE will not overestimate the number of change points, due to its search for the candidate with a minimum number of change points.
\begin{theorem}[Overestimation bound]\label{t: over}
Assume Model~\ref{QSR} and let $\widehat{\ncp}$ in \eqref{e: hat_K} be the estimated number of change-points by MUSCLE with $\alpha \in (0,1)$. Then for any $k\geq 1$,
\begin{equation*}
    \mathbb{P}(\widehat{\ncp}\geq \ncp+k) \leq \alpha^k.
\end{equation*}
In particular, \(\mathbb{P}(\widehat{\ncp}> \ncp) \leq \alpha.\)
\end{theorem}

By \cref{t: over,t: under}, we obtain that $\mathbb{P}(\hat{\ncp}\neq \ncp)$ is upper bounded by the sum of $\alpha$ and the exponential terms in~\eqref{e: under}. Next theorem gives a precise statement on model selection consistency.

\begin{theorem}[Model selection consistency]\label{t: consistency 1}
Assume Model~\ref{QSR} with the number $K = K_n$ of change points. Let $\alpha_n$ be a sequence in $(0,1)$ that converges to zero as $n\to \infty$, and $\Lambda = \Lambda_n$, $\Theta = \Theta_n$ and $\Omega = \Omega_n$ in \eqref{e: LTO}. Suppose that either of the following two conditions is fulfilled:
\begin{enumerate}[label=\emph{(\roman*)}]
    \item \label{c: consistency i}
    $\liminf \Lambda_n >0$ and $\sqrt{n}\Omega_n - \max_{m\leq n}q_{\alpha_n}(m)/\sqrt{2}\to \infty$. 
    \item \label{c: consistency ii}$\liminf \Lambda_n =0$ and for some constant $c>2$,
    \begin{equation*}
        \sqrt{n}\Omega_n \geq c\sqrt{-\log \Lambda_n}\; \text{ and }\;
        (c-2)\sqrt{-2\log \Lambda_n}-{\max_{m\leq n}q_{\alpha_n}(m)} \to \infty.
    \end{equation*}
\end{enumerate}
If $\hat{K}_n$ in \eqref{e: hat_K} is the estimated number of change points by MUSCLE with $\alpha = \alpha_n$, then
\begin{equation*}
    \lim_{n\to \infty}\mathbb{P}(\widehat{\ncp}_n = \ncp_n) =1.
\end{equation*}
Namely, the MUSCLE estimates the number of change point consistently.
\end{theorem}

\begin{remark}
By Lemma~\ref{l: ub q_n} in Appendix~\ref{A1}, we have $\max_{m\leq n}q_{\alpha_n}(m) \lesssim \sqrt{-\log \alpha_n}.$ This allows us to select $\alpha_n$ approaching zero at certain explicit rates to satisfy the above conditions on $q_{\alpha_n}(m)$. We discuss the implication of \cref{t: consistency 1} under conditions \emph{\ref{c: consistency i}} and \emph{\ref{c: consistency ii}}, separately. 
\begin{unlist}
\item[Case \emph{\ref{c: consistency i}.} ] Note that $\liminf_n \Lambda_n >0$ implies that $\sup \ncp_n<\infty$, that is, the number of change points stays bounded. In particular, if the quantile function $f=f_n$ in Model~\ref{QSR} stays the same for different $n$,  MUSCLE with {any vanishing sequence $\alpha_n \searrow 0$} estimates consistently the number of change points. 
\item[Case \emph{\ref{c: consistency ii}.} ]   
It implies that $\ncp_n$ can be consistently estimated by MUSCLE if the function $f=f_n$ in Model~\ref{QSR} satisfies 
    \begin{equation}
        \label{e: limit condition}
        \sqrt{n}\Omega_n \geq c\sqrt{-\log \Lambda_n}.
    \end{equation}
Consider the example of $f_n = \Delta_n \mathbbm{1}_{[0.5,0.5+\Lambda_n)}$ with additive {standard} Gaussian noise and jump size $\Delta_n>0$, which corresponds to Model~\ref{QSR} with $\beta = 0.5$. By simple computation we have 
\begin{equation*}
\Theta_n = \min\{\xi_{1,n},\xi_{2,n}\} =  \Phi\left(\tfrac{\Delta_n}{2}\right) - \Phi(0) =
 \tfrac{\Delta_n }{2\sqrt{2\pi}}+ o(\Delta_n),
\end{equation*}
where $\Phi$ is the distribution function of standard Gaussian random variable. 
Since $\Omega_n \ge \sqrt{\Lambda_n}\Theta_n$, the two change points of $f_n$ can be consistently detected by MUSCLE provided that $\sqrt{n\Lambda_n}\Delta_n \geq 6 c \sqrt{-\log \Lambda_n}$. Note that no method can consistently detect change points if 
\begin{equation*}
   \limsup_{n\to \infty}\frac{\sqrt{n\Lambda_n}\Delta_n}{\sqrt{-2\log \Lambda_n}} \;<\;1,
\end{equation*}
see e.g.\ \citet{arias2011detection}.
Thus, up to a constant, the condition \eqref{e: limit condition} can not be improved.
\end{unlist}
\end{remark}

Let $T$ and $\widehat{T}$ denote the sets of true change points and estimated change points, respectively. The {localization error} between $T$ and $\widehat{T}$ is defined by 
\begin{equation}
    \label{e: local error}
    d(T;\widehat{T}) \coloneqq \underset{k = 1,\dots,\ncp}{\max}\,\underset{j = 1,\dots,\widehat{\ncp}}{\min}\left|\tau_k - \hat{\tau}_j\right|.
\end{equation}
The {Hausdorff distance} between $T$ and $\hat{T}$ is given by
\begin{equation}
    \label{e: Hausdorff distance}
    d_{H}\left(T,\hat{T}\right) \coloneqq \max\left\{d(T;\widehat{T}), d(\widehat{T};T)\right \}  = \max\left\{\underset{k}{\max} \, \underset{j}{\min}\, \left|\tau_k - \hat{\tau}_j\right|,\underset{j}{\max}\, \underset{k}{\min}\,\left|\tau_k - \hat{\tau}_j\right|\right\}.
\end{equation}

\begin{theorem}[Localization rate in Hausdorff distance] \label{t: Hausdorff}
Assume Model~\ref{QSR} with $T_n$ the set of change points and $K_n = \# T_n$. Let $\alpha_n$ be a sequence in $(0,1)$, and $\Lambda = \Lambda_n$, $\Theta = \Theta_n$ and $\Omega = \Omega_n$ in \eqref{e: LTO}. Let also $\widehat{T}_n$ be the set of estimated change points by MUSCLE with $\alpha = \alpha_n$, then
\begin{multline}
\label{e: Hausdorff}
    \mathbb{P}\left(d_H(T_n,\hat{T}_n)\geq \varepsilon_n\right) \leq \alpha_n +  4\ncp_n\exp \left(-n\varepsilon_n\Theta_n^2\right)   
    +4\ncp_n \exp\left[-\Bigl(\sqrt{n\varepsilon_n}\Theta_n  - \tfrac{\underset{m\leq n}{\max}q_{\alpha_n}(m)}{\sqrt{2}}-\sqrt{\log\tfrac{2e}{\varepsilon_n}}\Bigr)_+^2\right]  
\end{multline}
for any sequence $\varepsilon_n$ such that $\Lambda_n\geq 2\varepsilon_n >0$. 
\end{theorem}

\begin{remark}
Setting $\alpha_n = n^{-\gamma}$ for some $\gamma>0$, we have $\max_{ m\leq n} q_{\alpha_n}(m) \lesssim \sqrt{\log n}$ by Lemma~\ref{l: ub q_n} in Appendix~\ref{A1}. Let 
\begin{equation*}
 \varepsilon_n = \frac{C \log n}{n\Theta_n^2}\quad\text{for some sufficiently large constant C.}   
\end{equation*}
Then the negative of logarithm of the last term in \eqref{e: Hausdorff} is
\begin{equation*}
\Gamma_n \coloneqq\Bigl(\sqrt{n\varepsilon_n}\Theta_n - \tfrac{\underset{m\leq n}{\max}q_{\alpha_n}(m)}{\sqrt{2}}-\sqrt{\log\tfrac{2e}{\varepsilon_n}} \Bigr)_+^2 - \log \ncp_n - \log 4 \gtrsim\log n\to\infty.
\end{equation*}
Note that the last term in \eqref{e: Hausdorff} is always greater than the second one. Thus, $\Gamma_n \to \infty$ implies that the right-hand side in \eqref{e: Hausdorff} asymptotically vanishes. That is, MUSCLE estimates change points in terms of Hausdorff distance with the rate $\varepsilon_n \asymp (\log n) / (n \Theta_n^2)$ with probability tending to $1$. This rate coincides (up to a log factor) with the minimax optimal rates in the Gaussian setup (see \citealp[Proposition~3.3]{verzelen2023optimal}).
\end{remark}

\subsection{Control of false positives}
\label{S: 3.2}
In order to quantify the possible false detection of change points incurred by controlling local errors, we consider \emph{false discovery rate} (FDR, introduced by \citealt{li2016fdr}) and \emph{overestimation rate} (OER, defined below). For the former, we treat the estimated change points $\{\hat{\tau}_1,\dots, \hat{\tau}_{\widehat{\ncp}}\}$ as discoveries.  For  $i \in \{1,\dots, \widehat{\ncp}\}$, the recovery $\hat{\tau}_i$ is classified as a \emph{true {discovery}} if there exists a true change point contained in 
\begin{equation*}
    \left[ \tfrac{\hat{\tau}_{i-1}+\hat{\tau}_{i}}{2}, \;\tfrac{\hat{\tau}_{i}+\hat{\tau}_{i+1}}{2}\right)
\end{equation*}
with $\hat{\tau}_0 = 0$ and $\hat{\tau}_{\widehat{\ncp}+1}=1$. Otherwise, it is a \emph{false discovery}. The FDR is defined as 
\begin{equation}\label{e: FDR}
    \text{FDR} \coloneqq \mathbb{E}\left[\frac{\rm FD}{\widehat{\ncp}+1}\right],
\end{equation}
where FD denotes the number of false discoveries. In a similar spirit, we introduce OER as
\begin{equation}
    \label{e: OER}
    {\rm OER} \coloneqq \mathbb{E}\left[\frac{(\widehat{\ncp}-\ncp)_+}{\max\{\widehat{\ncp},\, 1\}}\right]
\end{equation}
Both FDR and OER of MUSCLE can be upper bounded by the only tuning parameter $\alpha$.

\begin{theorem}[FDR and OER control]\label{t: FDR}
Assume Model~\ref{QSR}, and let  $\{\hat{\tau}_1,\dots, \hat{\tau}_{\widehat{\ncp}}\}$ be the estimated change points by MUSCLE with $\alpha \in (0,1)$. Then its FDR in \eqref{e: FDR} and OER in \eqref{e: OER} satisfy
$$    
{\rm FDR} \leq \alpha,\quad \text{and}\quad
    {\rm OER} \leq 
    \begin{cases}
    \frac{\alpha}{(1-\alpha )\ncp+\alpha}\,\leq\,\alpha ,& \ncp \geq 1,\\
    \alpha,& \ncp=0. 
    \end{cases}
$$
\end{theorem}

\cref{t: FDR} provides a statistical interpretation for the only tuning parameter $\alpha$ of MUSCLE in the sense of false positive control. In real world data analysis, one can choose $\alpha$ according to the tolerance to false positives. Besides, it is recommended to try several choices of $\alpha$. This would reveal which change points are most likely to be genuine and which are not, since the detected change points are often nested as $\alpha$ increases (cf.\ \cref{f: well-log}). The upper bound in \cref{t: FDR} is sharper than the result in \citet[Theorem~2.2]{li2016fdr} for FDRSeg.

\section{Algorithm and Implementation}
\label{S: implementation}
MUSCLE is defined as a solution to the non-convex optimization problem in \eqref{e: C_k}--\eqref{e: MUSCLE} with a large number of constraints. This optimization manifested by its $\ell_0$-pseudo norm objective, remains computationally challenging in general, despite the recent progress in mixed integer optimization \citep{BeVa20}. Nevertheless, thanks to its specific structure, MUSCLE can be efficiently computed within the dynamic programming framework. This is possible as we can decompose the optimization problem into $n$ subproblems that retain a recursive structure. For $j = 1, \ldots, n$, the $j$-th subproblem involves data $Z_1,\dots, Z_{j}$ and takes the form~of 
\begin{equation*}
    \widehat{K}[j] \coloneqq \min \Bigg\{k \,: \underset{0\leq i \leq k}{\max}T_{\mathring{I}_i}(Z,\theta_i)-q_{\alpha}(|\mathring{I}_i|)\leq  0,
     \text{for some } f = \sum_{i=0}^{k}\theta_i \mathbbm{1}_{I_i} \text{ with } \bigsqcup_{i=0}^{k}I_i = \bigl[0,\tfrac{j}{n}\bigr) \Bigg\}.
\end{equation*}
Introduce $\mathcal{I}\left(I\right)$ as an indicator on whether there is a feasible constant solution on interval $I$, i.e.,
\begin{equation*}
    \mathcal{I}  \left(I\right)=
    \begin{cases}
    1 & \text{if there is a $\theta$ such that } T_{\mathring{I}}(Z,\theta)\leq q_{\alpha}(\abs{\mathring{I}}),\\
    0 & \text{otherwise.}
    \end{cases}
\end{equation*}
It is easy to see that the following recursive relation holds
\begin{subequations}\label{e: bellman}
\begin{align*}
    \widehat{K}[0] &\coloneqq -1, \\
    \widehat{K}[j] &= \min\left\{\widehat{K}[i]+1\,: \mathcal{I}\left(\bigl[\tfrac{i}{n},\tfrac{j}{n}\bigr)\right) = 1, \, i = 0,\dots,j-1 \right\}
\end{align*}
\end{subequations}
for $j = 1,\ldots, n$, which is often referred to the \emph{Bellman equation} \citep{cormen2022introduction}. Then the estimated number $\widehat{K}$ of change points by MUSCLE, which is $\widehat{K}[n]$, can be computed via $\widehat{K}[j]$ with $j$ ranging from $1$ to $n$ based on the above recursive relation. 

The crucial part of computation lies in evaluating $\mathcal{I}  \bigl([{i}/{n},{j}/{n})\bigr)$. By the definition of  $T_{I}(Z,\theta)$ in \eqref{e: T_I}, we see that the constraints involved in $\mathcal{I}  \bigl([{i}/{n},{j}/{n})\bigr)$ are formulated in terms of  $g_{\beta}:(0,1)\to \mathbb{R}$, 
\begin{equation*}
    g_{\beta}:\quad x\;\;\mapsto\;\; x\log\left(\tfrac{x}{\beta}\right)+(1-x)\log\left(\tfrac{1-x}{1-\beta}\right).
\end{equation*}
The strict convexity of $g_{\beta}(\cdot)$ allows us to translate the evaluation of $\mathcal{I}  \bigl([{i}/{n},\;{j}/{n})\bigr)$ into the computation of $q$-quantiles of $(Z_i,\dots,Z_j)$ for different $q\in (0,1)$. A straightforward approach is to sort the data segment $(Z_i,\dots,Z_j)$ for every pair of $i$ and $j$, which requires a total computational complexity $O(n^3 \log n)$. Towards a computational speedup, a double heap data structure was employed in \citet{jula2022multiscale}, which unfortunately requires an $O(n^2)$ memory space in the worst scenario. The two approaches will become quickly infeasible as the scale of the problem increases. Instead, we adopt a recent advance in data structures, known as the \emph{wavelet tree} \citep{grossi2003high,Nav14}, which achieves an optimal balancing between computational and memory complexities \citep{brodal2011towards}. Utilising the wavelet tree, we are able to compute an arbitrary $q$-quantile of $\{Z_i,\dots,Z_j\}$ in $O(n \log n)$ runtime, with little expense that involves a  pre-computation  of $O(n \log n)$ and a memory space of $O(n \log n)$. 

Finally, combining the wavelet tree and the dynamic programming with pruning techniques, we obtain \cref{alg: MUSCLE} (in Appendix~\ref{A3}) for computing the global optimal solution to the non-convex optimization problem defining MUSCLE. 

\begin{theorem}\label{th: complexity}
MUSCLE, defined as the global optimal solution to the optimization problem in \eqref{e: C_k}--\eqref{e: MUSCLE}, can be computed by a pruned dynamic program with the wavelet tree data structure, which has a space complexity (i.e., amount of memory) of $O(n \log n)$ and a worst-case computational complexity (i.e., runtime) of $O(n^4\log n)$. Further, if we use the intervals of dyadic lengths in MUSCLE, the worst-case computational complexity is $O(n^3\log^2 n)$.
\end{theorem}

\begin{remark}[Runtime and speedup]
\label{r: compleixity}
We stress that $O(n^4\log n)$ or $O(n^3\log^2 n)$ is an upper bound on the computational complexity, see the proof of \cref{th: complexity} (in Appendix~\ref{A3}). In fact, the actual computational complexity depends on the unknown signal and the noise level. Roughly spoken, a large number of detected change points would lead to a low computational complexity. For instance, in case of a signal with $O(n)$ change points that are almost evenly spaced and a low noise level, the computational cost may even be $O(n\log n)$. 

The use of sparse collection of intervals is motivated by the observation that intervals with large overlaps encode similar information in the data (see \citealp{Walther2022Calibrating} for a thorough discussion). In practice, we recommend the use of intervals that are of dyadic lengths, which provides a balance between computational and statistical efficiencies (cf.\ \citealp{pein2018fully, kovacs2023seeded}). We stress that MUSCLE with this modification enjoys nearly the same statistical guarantees (in \cref{S: 3}) with some natural adjustments (cf.\ \citealp{Walther2022Calibrating}). Empirically, we find through simulations that the computational complexity is no more than $O(n^2)$ in nearly all scenarios.  

In addition, for time series of ultra-large scale, we may split the whole dataset into several pieces and compute MUSCLE on each piece separately. The resulting estimates are afterwards merged by refitting (again via MUSCLE) the two segments that are close to the splitting location of the dataset. Interestingly, this splitting-merging procedure often leads to an estimate that is fairly close to MUSCLE applied to the whole dataset. The reason behind is that MUSCLE controls the local error defined on individual segments, and it is thus robust to the choice of data window (cf.\ \cref{f: E1}). If we split the data into pieces of a fixed size (say $300$ up to rounding effects), this procedure leads to a computational complexity that depends linearly on the sample size~$n$. In addition, we can compute the estimates on pieces of data simultaneously via \emph{parallel computing} (see e.g.~\citealp{almasi1994highly}), which allows a linear speedup in the number of computational units. This splitting and merging variant of MUSCLE is denoted by \emph{MUSCLE-S}.
\end{remark}

\begin{remark}[Data structure]\label{r:dstr}
Besides the wavelet trees, there are alternative data structures that allow efficient computation of quantiles on intervals of data (known as range quantile queries in computer science), such as \emph{merge sort trees} and \emph{persistent segment trees} (see also \citealp{castro2016wavelet}). A detailed comparison of different data structures in terms of computational and memory complexities is in \cref{t: complexity} in Appendix~\ref{A3}. It reveals that the wavelet trees and persistent segment trees achieve the best efficiency in terms of both computation and memory. 
\end{remark}

\section{Simulations and Applications}
\label{S: 4}
\subsection{Simulation studies}
\label{S: 4.1}
We investigate the empirical performance of MUSCLE and its splitting-merging variant  MUSCLE-S (see \cref{r: compleixity}) under various situations. To benchmark the performance, we include PELT \citep{killick2012optimal}, WBS \citep{fryzlewicz2014wild}, SMUCE \citep{frick2014multiscale}, FDRSeg \citep{li2016fdr} and MQS \citep{jula2022multiscale} in Gaussian models, and R-FPOP \citep{fearnhead2019changepoint}, KSD \citep{madrid2021optimal}, MQS \citep{jula2022multiscale}, NOT-HT \citep[the heavy tailed version]{baranowski2019narrowest}, ED-PELT \citep{haynes2017computationally}, RNSP \citep[the CUSUM version]{fryzlewicz2024robust} in non-Gaussian models, as competitors. For implementation, we compute WBS from {R} package \texttt{wbs}, PELT from \texttt{changepoints}, SMUCE and FDRSeg from \texttt{FDRSeg}, MQS from \texttt{mqs}, NOT-HT from \texttt{not}, ED-PELT from \texttt{changepoint.np}, all available from CRAN. The implementations of R-FPOP, KSD and RNSP are available on Github (\url{https://github.com/guillemr/robust-fpop}; \url{https://github.com/hernanmp/NWBS}; \url{https://github.com/pfryz/nsp}). Default tuning parameters are used for PELT, R-FPOP, WBS, KSD and NOT-HT. For SMUCE, MQS, FDRSeg MUSCLE, MUSCLE-S and RNSP, we select $\alpha = 0.3$. The length of pieces in MUSCLE-S is set to $300$. We use the \texttt{max\_sign\_cusums} criterion (see \citealp{sen1975tests} and \citealp{fryzlewicz2024robust}) for estimating change points in RNSP.

As quantitative evaluation criteria of change points, we consider the estimated number $\widehat{K}$ of change points, the {localization error}  $d(\cdot;\cdot)$ in \eqref{e: local error}, the {Hausdorff distance} $d_{H}(\cdot,\cdot)$ in \eqref{e: Hausdorff distance}, and the FDR in \eqref{e: FDR}. We also include the standard criteria of function estimation, consisting of the {mean integrated square error} (MISE) and the {mean integrated absolute error} (MIAE), formally defined in \eqref{e: MISE} and \eqref{e: MIAE} in Appendix~\ref{A2}. 
Additionally, as change point problems can be viewed as clustering problems, we incorporate the V-measure \citep{rosenberg2007v}, which ranges from $0$ to $1$ with larger values indicating higher accuracy. 

\subsubsection{Changes in median}
\label{S: 5.1.1}
We consider the additive noise model in \eqref{e:noise_signal} with $\beta = 0.5$, i.e., noise terms $\varepsilon_i$ are independent random variables with median zero. The test signal in the Gaussian scenario \ref{i:gauss} has $K = 2$ change points at $986$ and $1016$ with sample size $n = 2000$. The values between change points are $-4$, $0$ and $4$. In other scenarios \ref{i:tdist}--\ref{i:hetemix} we use the \texttt{blocks} signal \citep{donoho1994ideal} of length $n = 2048$. Its change points are 205, 267, 308, 472, 512, 820, 902, 1332, 1557, 1598, 1659 and the values between change points are $0$, $14.64$, $-3.66$, $7.32$, $-7.32$, $10.98$, $-4.39$, $3.29$, $19.03$, $7.68$, $15.37$, $0$. The noise distributions are specified as follows.

\begin{enumerate}[wide, label = (E\arabic*)]
    \item\label{i:gauss} Normal distribution $\mathcal{N}(0,0.9)$. 
    \item\label{i:tdist} Scaled $t$-distribution $2^{-0.5}\cdot \bigl(8\cdot \mathbbm{1}_{\{1\leq i \leq 389\}} + 0.5 \cdot\mathbbm{1}_{\{390 \leq i \leq 666\}} +4 \cdot \mathbbm{1}_{\{667 \leq i \leq 1445\}} + \mathbbm{1}_{\{1446 \leq i \leq 2048\}}\bigr) \cdot t_3$.
    \item\label{i:cauchy} Scaled Cauchy distribution $\bigl(0.6\cdot \mathbbm{1}_{\{1\leq i \leq 389\}} + 0.05 \cdot\mathbbm{1}_{\{390 \leq i \leq 666\}} +0.6\cdot \mathbbm{1}_{\{667 \leq i \leq 1445\}} +0.2\cdot \mathbbm{1}_{\{1446 \leq i \leq 2048\}}\bigr) \cdot {\rm Cauchy(0,1)}.$
    \item\label{i:chi2} Scaled and centered $\chi^2$-distribution $6^{-0.5}\cdot\bigl(6\cdot \mathbbm{1}_{\{1\leq i \leq 389\}} + 0.5\cdot \mathbbm{1}_{\{390 \leq i \leq 666\}} +6\cdot \mathbbm{1}_{\{667 \leq i \leq 1445\}} +2\cdot \mathbbm{1}_{\{1446 \leq i \leq 2048\}}\bigr) \cdot \bigl(\chi_3^2 -  \text{median}(\chi_3^2)\bigr).$
   \item\label{i:hetemix} Mixed distribution \(\mathbbm{1}_{\{1\leq i \leq 389\}}\cdot \mathcal{N}(0,64) + 1/({2\sqrt{3}})\cdot \mathbbm{1}_{\{390 \leq i \leq 666\}} \cdot  t_3 + {4}/{\sqrt{6}}\cdot\mathbbm{1}_{\{667 \leq i \leq 1445\}}\cdot \bigl(\chi_3^2 -  \text{median}(\chi_3^2)\bigr) + 0.1 \cdot  \mathbbm{1}_{\{1445 \leq i \leq 2048\}}\cdot \text{Cauchy(0,1)}\).
\end{enumerate}

In all scenarios, the simulations are repeated 200 times. We display simulation results in scenarios \ref{i:gauss}--\ref{i:tdist} and \ref{i:hetemix} in \cref{f: blocks Gaussian,f: blocks t3,f: blocks Mix H}, respectively. Further detailed results as well as scenarios \ref{i:cauchy} and \ref{i:chi2} are given in Appendix~\ref{A2} (\cref{f: blocks Cauchy,f: blocks Chi_sq,t: E1-E3}). In the Gaussian scenario~\ref{i:gauss}, FDRSeg and WBS overestimate the number of change points, leading to a relatively large FDR. In contrast, PELT, SMUCE and MUSCLE provide more accurate estimation of $K$. Despite MUSCLE having a slightly larger MISE compared to SMUCE and FDRSeg, it outperforms PELT and WBS. Overall, the performance of MUSCLE is comparable to the procedures tailored for the Gaussian model. In non-Gaussian scenarios \ref{i:tdist}--\ref{i:hetemix}, there are 11 change points in median, and 14 change points in distribution. We compute localization error, Hausdorff distance, FDR and V-measure with respect to the change points in median, and MISE and MIAE with respect to the median function.  
We discuss the performance of each method in detail as follows. 
\begin{unlist}
    \item R-FPOP, while being competitive in \ref{i:cauchy}, exhibits noticeable underperformance in all heteroscedastic scenarios \ref{i:tdist}, \ref{i:chi2} and \ref{i:hetemix}. A possible reason is that R-FPOP estimates the standard deviation of noise using the entire data, which is not well-suited in case of heteroscedasticity. 
    \item NOT-HT performs closely to the best in \ref{i:tdist}, but shows an undesirable performance in the other scenarios. 
    \item MQS estimates well the number $K$ of change points in scenarios {\ref{i:gauss}}, \ref{i:tdist}, \ref{i:cauchy} and \ref{i:hetemix}, while significantly underestimates $K$ in \ref{i:chi2}. In all scenarios, the estimated locations of change points are relatively imprecise. 
    \item KSD aims at the detection of distributional changes, but often fails to detect all of $14$ change points. The detected change points contain only a subset of change points in median, as revealed by Hausdorff distance and MISE (as well as localization error, MIAE and V-measure in \cref{t: E1-E3}).
    \item ED-PELT is also designed to detect changes in distribution. It seriously overestimates the number of change points in all scenarios, and performs even worse than KSD.
    \item RNSP estimates well the number of change points, and yields competitive median regressions, in scenarios \ref{i:tdist}, \ref{i:cauchy} and \ref{i:hetemix}. However, in \ref{i:chi2}, it steadily underestimates the number of change points and reports relatively worse Hausdorff distance.
    \item MUSCLE is among the best ones in all scenarios, and in particular, in scenarios \ref{i:chi2} and \ref{i:hetemix}. The competitive performance of MUSCLE aligns with its statistical guarantees (\cref{S: 3}). As a speedup variant of MUSCLE, MUSCLE-S has nearly the same performance as MUSCLE, and tends to report slightly more change points. 
\end{unlist}

Regarding computation time (\cref{t: E1-E3}), PELT and R-FPOP are the fastest, followed closely by SMUCE, WBS, NOT-HT, MUSCLE-S, MQS and FDRSeg. MUSCLE is noticeably slower, and RNSP and KSD are the slowest. In all scenarios, MUSCLE-S is roughly ten times faster than MUSCLE, and takes about one~second for 2000 samples. All runtimes are measured on a standard laptop (Intel i5 CPU, 1.40\,GHz, four cores, 16\,GB RAM).

\begin{figure}[t]
        \centering
        \includegraphics[width=\linewidth]{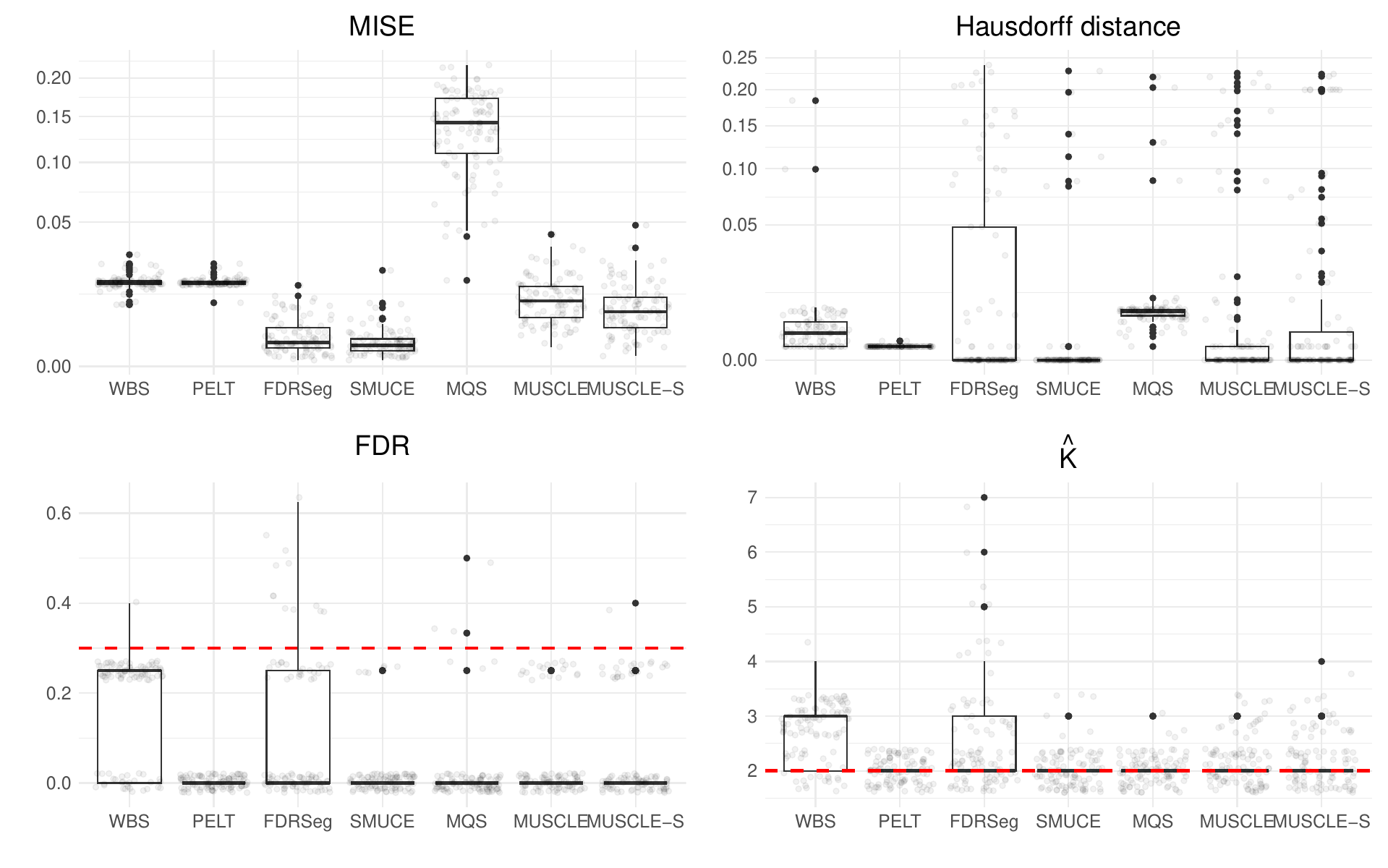}
        \caption{Recovery of the underlying signal in \ref{i:gauss}. In each panel, the overall performance  over 200 repetitions is summarized as a boxplot and individual repetitions are jittered in dots with a low intensity. In the bottom left panel, the theoretical upper bound $\alpha = 0.3 $ on the FDR of MUSCLE is marked by a red dashed line. In the bottom right panel, the true number of change points is $K = 2$ (marked by a red dashed line).}
        \label{f: blocks Gaussian}
\end{figure}

\begin{figure}[t]
        \centering
        \includegraphics[width=\linewidth]{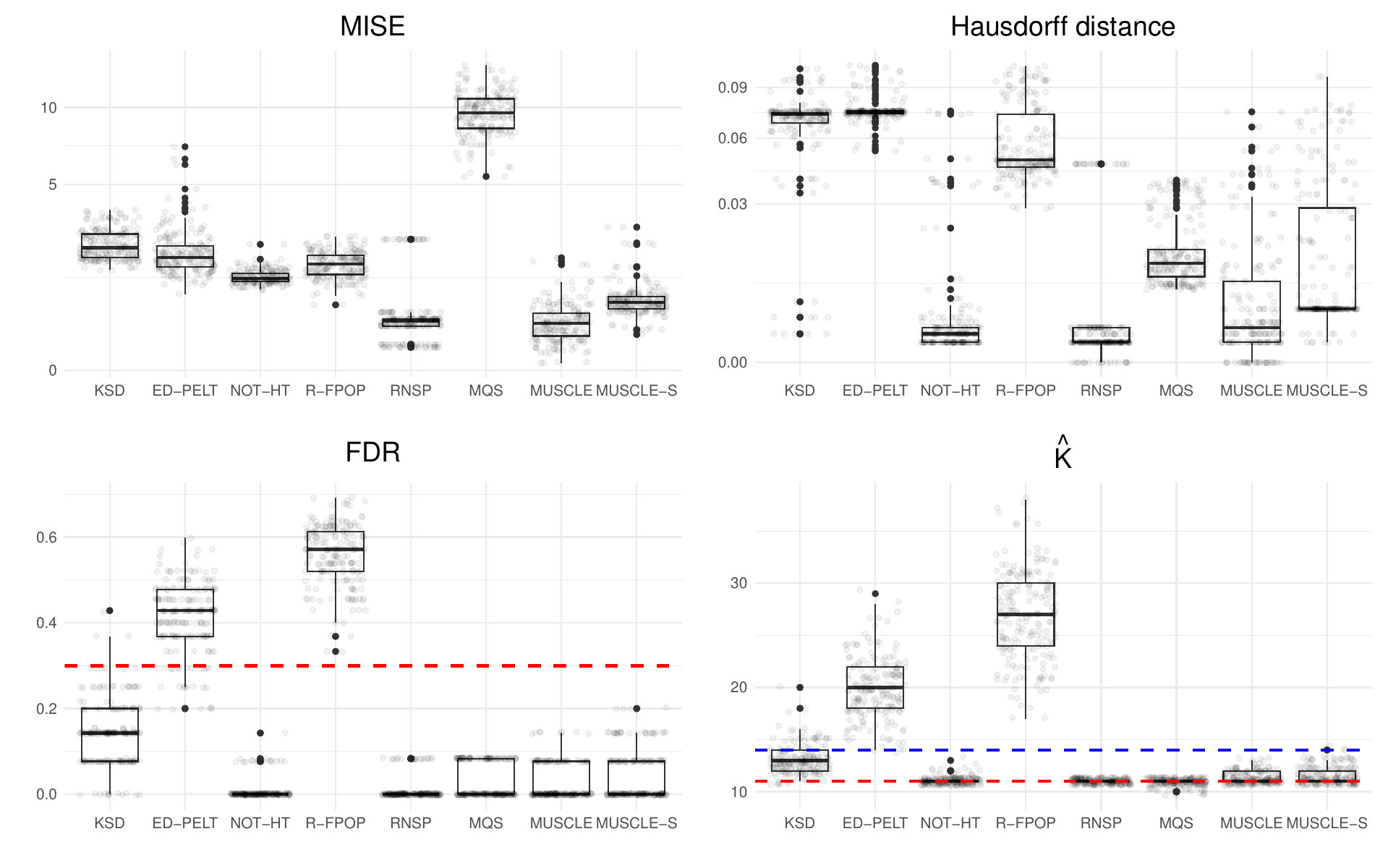}
        \caption{Recovery of the \texttt{blocks} signal in \ref{i:tdist}. In each panel, the overall performance  over 200 repetitions is summarized as a boxplot and individual repetitions are jittered in dots with a low intensity. In the bottom left panel, the theoretical upper bound $\alpha = 0.3 $ on the FDR of MUSCLE is marked by a red dashed line. In the bottom right panel, the true numbers of change points in median and in distribution are $K = 11$ (marked by a red dashed line) and $K = 14$ (marked by a blue dashed line), respectively. }
        \label{f: blocks t3}
\end{figure}

\begin{figure}[t]
    \centering
    \includegraphics[width=\linewidth]{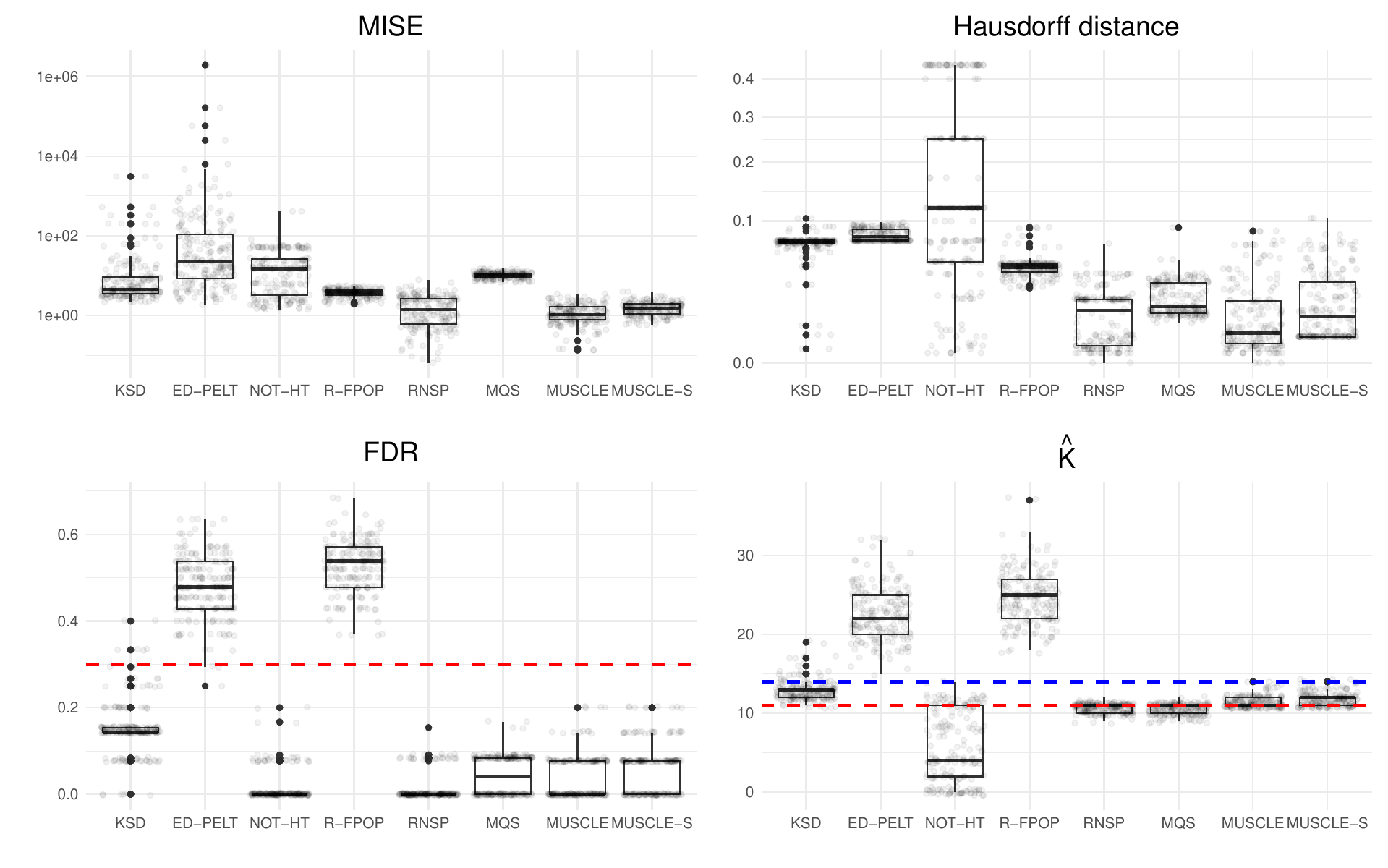}
      \caption{Recovery of the \texttt{blocks} signal in \ref{i:hetemix}. In each panel, the overall performance  over 200 repetitions is summarized as a boxplot and individual repetitions are jittered in dots with a low intensity. In the bottom left panel, the theoretical upper bound $\alpha = 0.3 $ on the FDR of MUSCLE is marked by a red dashed line. In the bottom right panel, the true numbers of change points in median and in distribution are $K = 11$ (marked by a red dashed line) and $K = 14$ (marked by a blue dashed line), respectively.}
    \label{f: blocks Mix H}
\end{figure}

\subsubsection{Changes in distribution}\label{ss:cp_box}
We now consider the detection of change points in distribution. If two random variables do not share the same distribution, at least one quantile of them must be different. Thus, a straightforward approach to detect change points in distribution is applying MUSCLE repeatedly to estimate various quantiles of observations, e.g., median ($\beta = 0.5$), lower $(\beta = 0.25)$ and upper ($\beta = 0.75$) quartiles. However, for a change point in distribution, this straightforward approach may report several slightly different estimates when using different quantiles. To address this issue, we modify the multiscale tests on individual segments (see \cref{ss:mtest}) by considering simultaneously multiple quantiles. Let $\beta = (\beta_1, \ldots, \beta_m) \in (0,1)^m$ be the considered levels of quantiles. For an interval $I \subseteq [0,\, 1)$, we perform a multiscale test on whether $\{Z_i : x_i \equiv (i-1)/n \in I \}$ have the same $\beta_r$-quantile for every $r = 1, \ldots, m$. Imitating the development in \cref{S: 2}, we may arrive at the multiscale side-constraint
\begin{multline*}
    \mathcal{C}_k^m \coloneqq \biggl\{g= \sum_{j=0}^{k} \theta_j \mathbbm{1}_{{I}_j} : \theta_j = (\theta_{j1},\ldots, \theta_{jm}), \bigsqcup_{j = 0}^k I_j = [0,1), \\
    \text{ and } T_{\mathring{I}_j}(Z,\theta_{jr}) \le q_{\alpha/m}(\abs{\mathring{I}_j}), \, j=0,\dots,k,\, r= 1,\ldots, m \biggr\}.
\end{multline*}
Note that we use the level $\alpha/m$ in local quantiles $q_{\alpha/m}$ in \eqref{e: q_m}, which comes from the Bonferroni correction on the multiplicity of considering $m$ quantiles. Based on $\mathcal{C}_k^m$, we introduce a variant of MUSCLE, referred to as M-MUSCLE (Multiple MUSCLE), by 
\begin{equation*}
    \hat{f}^m \; \in \; \underset{f\in \mathcal{C}_{\hat{K}_m}^m}{\arg \min}\, \sum_{r = 1}^{m}\sum_{i = 1}^{n} \bigl(Z_i-f_{\beta_r}(x_i)\bigr) \bigl(\beta_r-\mathbbm{1}_{\{Z_i\le f_{\beta_r}(x_i)\}}\bigr),
\end{equation*}
with $\widehat{K}_m \coloneqq \min\bigl\{k \in \mathbbm{N} \cup \{0\} : \mathcal{C}_k^m \neq \emptyset\bigr\}.$ This loss function is also used in composite quantile regression (see e.g.\ \citealp{zou2008composite}).

We demonstrate the performance of M-MUSCLE in scenario \ref{i:tdist} with $t_3$-noise, see \cref{f:MMUSCLE}. We select $\beta = (\beta_1,\beta_2,\beta_3) = (0.25,0.5,0.75)$ and set $\alpha = 0.3$. M-MUSCLE  estimates quite accurately the number of change points as well as their locations. In comparison with other nonparametric methods that aim at distributional changes, KSD often underestimates the number of change points, while ED-PELT largely overestimates the number of change points (see \cref{f: blocks t3}, the bottom right panel).

\begin{figure}[ht]
        \centering
        \includegraphics[width = \linewidth]{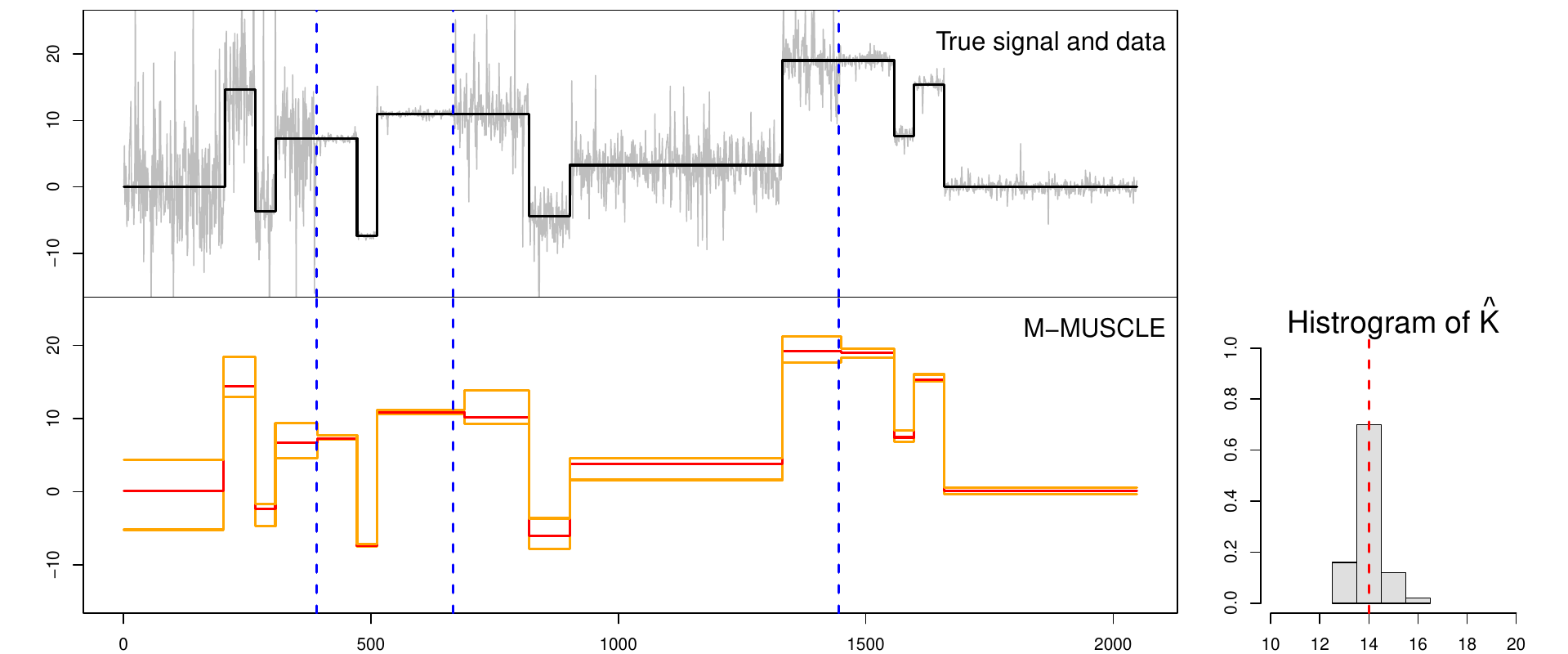}
        \caption{Recovery of the \texttt{blocks} signal in \ref{i:tdist} by M-MUSCLE ($\alpha = 0.3$, bottom). The estimated upper and lower quartiles are in orange lines and the estimated median in red line. The histogram on the right summarizes its performance in estimating the number of change points over 200 repetitions, where the vertical dashed red line indicates the true number of change points. On the top, the true signal (black line) and the data (grey line) are plotted. Three change points of variances are highlighted by vertical dashed blue lines.}
        \label{f:MMUSCLE}
\end{figure}

\subsection{Real data examples}
\label{S: 5.2}
\subsubsection{Well-log dataset}\label{ss:well-log}
The well-log dataset \citep{ruanaidh2012numerical} contains 4050 measurements of the nuclear magnetic response of underground rocks. Collected by inserting a probe into a borehole, the true signal exhibits approximate piecewise constancy, with each segment corresponding to a distinct rock type characterized by constant physical properties. Change points in the signal reveal changes between rock types, holding crucial importance in the context of oil drilling. We refer to \citet{fearnhead2006exact}, \citet{wyse2011approximate} and \citet{fearnhead2019changepoint} for further details.

The raw data exhibits several short but sharp oscillations, with some of these fluctuations significantly smaller in amplitude compared to others. Given these distinct patterns, a Gaussian noise assumption might not be justified, which is supported by Shapiro--Wilk's test, yielding a $p$-value of $2.2\times 10^{-16}$. Consequently, many studies analyzing well-log data rely on a {pre-cleaned} dataset, wherein outliers are manually removed in advance. In general, the detection and removal of outliers are challenging tasks due to the scarcity of labelled data, stemming from the infrequent occurrence of outlier instances. Typical outlier detection methods identify data points that deviate significantly from the rest. However, proposing a universal measure of deviation applicable to all datasets and scenarios is hard or even impossible (cf.\ \citealp{boukerche2020outlier}).

The robustness of MUSCLE eliminates the need for a preliminary {pre-cleaning} step. We employ MUSCLE to estimate the median ($\beta = 0.5$) of well-log signal with different choices of the tuning parameter, namely, $\alpha = 0.1, 0.3$ and $0.5$, and by comparison we also consider other robust or nonparametric segmentation methods, including KSD, MQS, RNSP, NOT-HT, R-FPOP and ED-PELT, see \cref{f: well-log}. As shown in the top panel, the raw data has five pronounced \emph{flickering} events (i.e., events on tiny temporal scales). The estimates provided by KSD, MQS, RNSP and MUSCLE with $\alpha=0.1$ or $0.3$ are very similar, and none of them detect any flickering events. Despite this, they all capture the main shape of the well-log data. We set $\alpha = 0.5$ for MQS and RNSP in order to improve their detection power. By contrast, NOT-HT identifies two and R-FPOP detects three flickering events. However, R-FPOP seems to overfit the data around the last flickering (see \cref{f: Well log RFPOP} in Appendix~\ref{A2}). ED-PELT reports the highest number of change points. However, as it constantly overestimates $K$ in simulation studies, its estimation may not be trustworthy. With $\alpha = 0.5$, MUSCLE successfully recovers all of the five flickering events. In general, it is challenging to verify whether these flickering events (or some of them) are genuine change points or outliers causing by mistakes in measurement. Nevertheless, MUSCLE with various choices of $\alpha$ reveals a nested sets of detected features in the well-log data. As a larger $\alpha$ corresponds a lower statistical confidence, such a nested sets of detected features shed light on how likely these features are genuine.

\begin{figure}[ht]
        \centering
        \includegraphics[width = \linewidth]{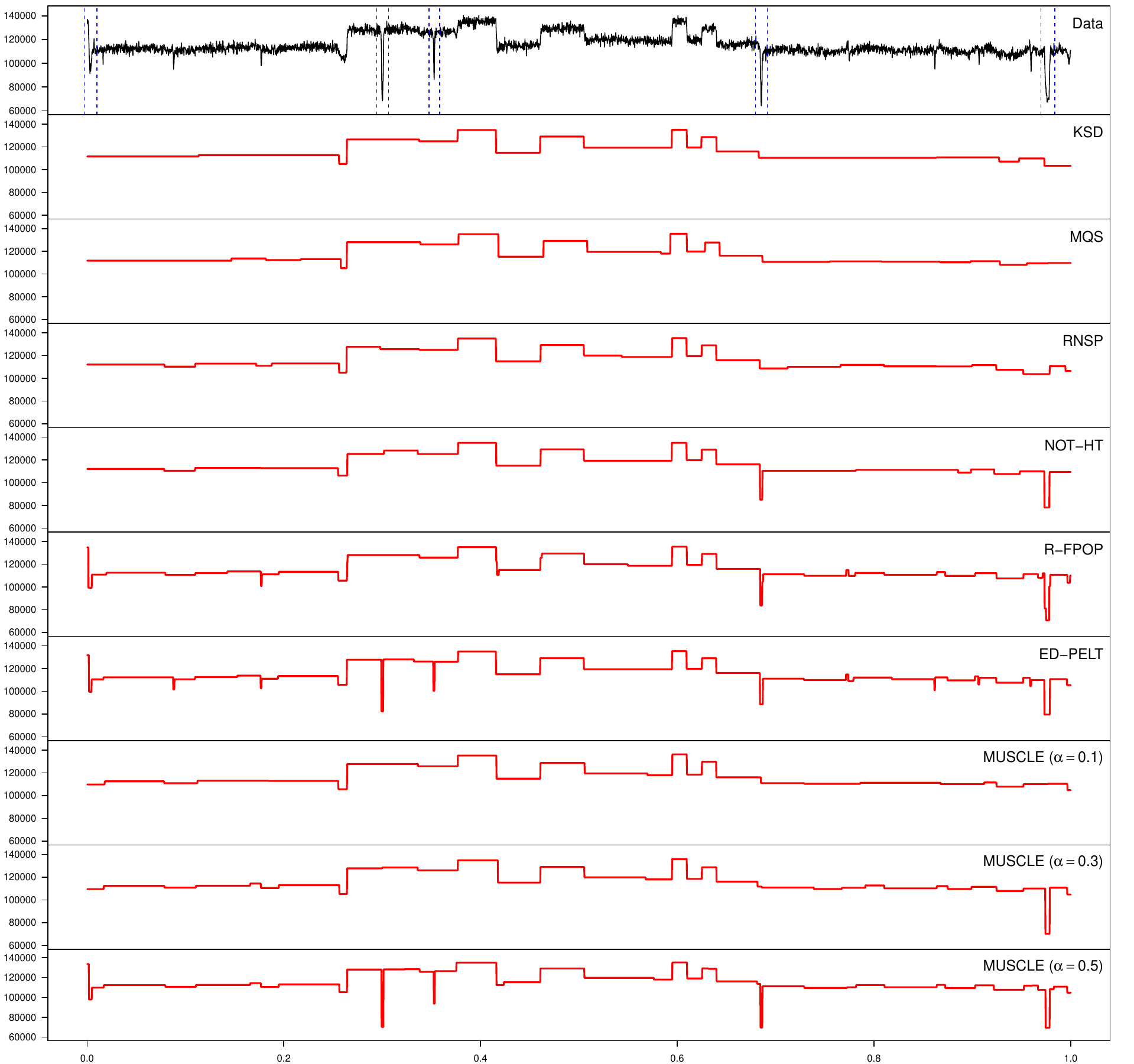}
\caption{Segmentation of the well-log dataset \citep{ruanaidh2012numerical}.  The raw data is in the top panel with five pronounced flickering events marked by vertical blue dashed lines. NOT-HT, KSD, R-FPOP and ED-PELT use default parameters. MQS and RNSP are tuned with $\alpha = 0.5$.}
        \label{f: well-log}
\end{figure}

\subsubsection{Ion channels recordings: Gramicidin D}
\label{S: 5.2.2}
Ion channels are proteins that control ion flows crossing the cell membrane by randomly opening and closing pores and such process is called \emph{gating}. A popular tool for the quantitative analysis of gating dynamics is the \emph{voltage-clamp} technique, which allows the measurement of electrical currents flowing through a single ion channel over time (cf.\ \citealp{sakmann2013single}). Since the measuring process involves an indispensable low-pass filter, the statistical model is 
\begin{equation}
\label{e: dependent noise}
Z_i = (\rho \ast f)\left(\tfrac{i}{\omega}\right) + \varepsilon_i, \quad i = 1,\dots, n.
\end{equation}
with $\omega$ the sampling rate, and $\rho$ the convolution kernel (corresponding to the low-pass filter). The underlying signal $f$ represents the conductance profile of ion channels, which is a piecewise constant function with change points caused by gating events. The random errors $\varepsilon_i$ are centered, heterogeneous, and typically non-Gaussian \citep{ven98ii}, and they are dependent due to the low-pass filter. 

We showcase an extension of MUSCLE to the dependent noise in \eqref{e: dependent noise}. To take into account the dependence structure, we need to adjust the local quantiles $q_{\alpha}$ in \eqref{e: q_m}. It is generally difficult to determine the dependence of $\mathbbm{1}_{\{{\varepsilon}_i\leq 0\}}$, even though the dependence of noise $\varepsilon_i$ arises from the known low-pass filter, unless specific assumptions on the noise distribution function are made. Note however that, as a crucial building block of MUSCLE, the multiscale test statistic converges weakly to an accessible distribution in form of supremum of Gaussian random variables under mild conditions (see e.g.\ \citealp{frick2014multiscale} and \citealp{DeEV20}). Inspired by this, we redefine  the local quantiles $q_{\alpha}$ in \eqref{e: q_m} by treating noise $\varepsilon_i$ as dependent Gaussian distributed, with the dependence structure derived from the convolution kernel $\rho$. We compute such local quantiles by Monte--Carlo simulations. Another challenge posed by model~\eqref{e: dependent noise} is that the convoluted signal $\rho * f$ is no longer piecewise constant. But the influence of convolution is only locally active, as the convolution kernel $\rho$ has a compact support. As a consequence, $\rho * f$ is still constant on $[s + \kappa, t)$, with $\kappa$ the kernel length (i.e.\ the size of the support of $\rho$), if $f$ is constant on $[s,t)$. Thus, we also modify the multiscale statistic $T_{I}$ in~\eqref{e: T_I} by restricting to intervals $J$ such that $J\subseteq I$ and $J - \kappa \subseteq I$. Based on these two adjustments, we obtain an extension of MUSCLE for dependent errors, denoted by {D-MUSCLE}. 

We compare the performance of D-MUSCLE with two state-of-the-art methods JULES \citep{pein2018fully} and HILDE \citep{pein2020heterogeneous}, which are designed especially for ion channel analysis. The implementation of JULES and HILDE is in {R} package \texttt{clampSeg} on CRAN. We use an ion channel recording of gramicidin~D (provided by the Steinem lab, Institute of Organic and Biomolecular Chemistry, University of Göttingen) with the sampling rate $\omega = 20$~kHz and the kernel length of 9.6~ms. The results are given in \cref{f: ion channel}. JULES performs relatively worse than the others as it misses a change point at 0.17~s (marked by the vertical dashed line). HILDE, MUSCLE and D-MUSCLE have a close performance, except that MUSCLE and D-MUSCLE report additionally a flickering event between 0.38~s and 0.40~s. Compared with MUSCLE, the detected change points by D-MUSCLE seem more reasonable by visual inspection. 

\begin{figure}[ht]
        \centering
        \includegraphics[width = \linewidth]{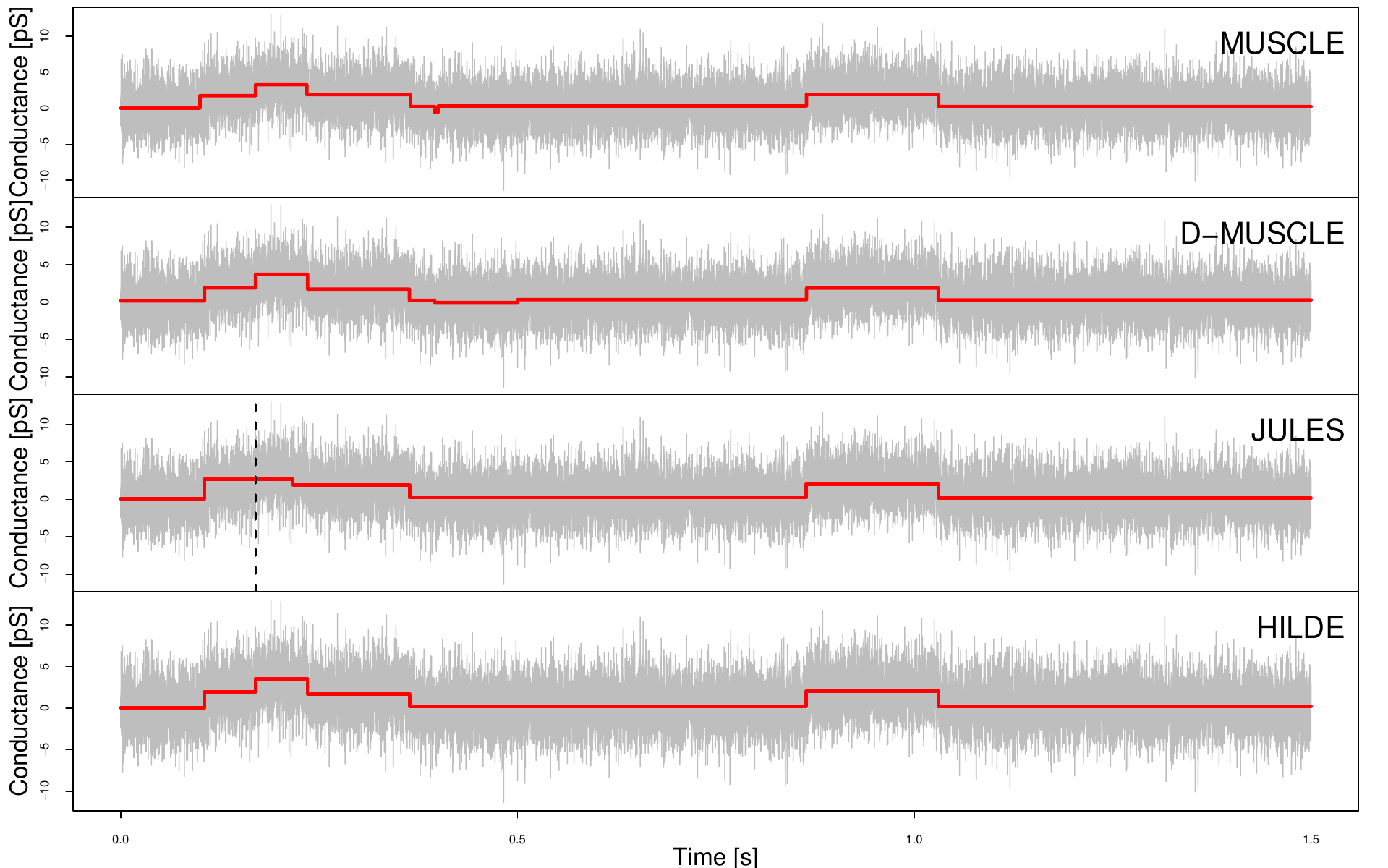}
        \caption{Segmentation (red lines) of gramicidin D recordings (grey lines; from the Steinem lab, University of Göttingen). We use the default parameters for JULES ($\alpha = 0.05$) and HILDE ($\alpha_1 = 0.01,\, \alpha_2 = 0.04$) and tune similarly for the proposed MUSCLE and D-MUSCLE with $\alpha = 0.05$. 
        }
        \label{f: ion channel}
\end{figure}

\section{Discussion}
\label{S: 6}
In this paper, we introduce a novel change point segmentation method MUSCLE that is robust and meanwhile has a high detection power. It combines the ideas of multiscale testing, quantile segmentation and variational estimation. As a distinctive feature, MUSCLE controls the local errors associated with individual segments. This local error control, as opposed to the commonly studied global error control, allows improved detection power, robustness to data windowing, as well as efficient computation. Under the quantile segmentation model with serial independence, we have established the consistency and the localization error rates for MUSCLE. Remarkably, these results match  (up to log factors) the minimax optimality results in Gaussian setups. Further, we obtain finite sample guarantees on its falsely detected change points in terms of FDR and OER. This provides a statistical interpretation on the unique tuning parameter $\alpha$ of MUSCLE in the sense that the expected proportion of falsely detected change points does not exceed $\alpha$. Moreover, we develop an efficient dynamic programming algorithm for MUSCLE, which leverages advanced wavelet tree data structures, modern pruning techniques, and sparse interval system proposals. We also introduce a speedup variant (i.e., MUSCLE-S) of MUSCLE, which scales linearly with sample size, particularly suitable for large scale datasets. The efficacy of MUSCLE(-S) is validated through comprehensive comparison studies using both simulated and real world data.   

We showcase two extensions of MUSCLE and demonstrate their empirical performances. The first extension, M-MUSCLE, focuses on detecting distributional changes by monitoring multiple quantiles simultaneously. The theoretical properties of MUSCLE would probably extend to M-MUSCLE, and there may be a link to segmentation methods using Wasserstein distances (e.g.\ \citealp{HoKW21}). The second extension, D-MUSCLE, addresses dependent data. Its statistical guarantees may be derived under the framework of functional dependence (introduced by \citealp{Wu05}) by additionally requiring the distributions of additive errors to be identical.  Further investigation into these extensions, as discussed above, represents intriguing directions for future research.

Finally, the concept of local error appears to be well-suited for online change point detection, as evidenced by its stability across various data windows. Consequently, developing an online version of MUSCLE represents another interesting direction of future research. 

\section*{Acknowledgments}
This work was supported by DFG (German Research Foundation) CRC~1456 {Mathematics of Experiment}, and by the DFG under Germany's Excellence Strategy, project EXC~2067 {Multiscale Bioimaging: from Molecular Machines to Networks of Excitable Cells} (MBExC).
The authors thank Manuel Fink and Claudia Steinem (Institute of Organic and Biomolecular Chemistry, University of Göttingen) for providing the ion channel data, and also Axel Munk and Robin Requadt for helpful discussions.


\begin{thebibliography}{}

\bibitem[Almasi and Gottlieb, 1994]{almasi1994highly}
Almasi, G.~S. and Gottlieb, A. (1994).
\newblock {\em Highly parallel computing}.
\newblock Benjamin-Cummings Publishing Co., Inc.

\bibitem[Arias-Castro et~al., 2011]{arias2011detection}
Arias-Castro, E., Cand\`es, E.~J., and Durand, A. (2011).
\newblock Detection of an anomalous cluster in a network.
\newblock {\em Ann. Statist.}, 39(1):278--304.

\bibitem[Baranowski et~al., 2019]{baranowski2019narrowest}
Baranowski, R., Chen, Y., and Fryzlewicz, P. (2019).
\newblock Narrowest-over-threshold detection of multiple change points and
  change-point-like features.
\newblock {\em J. R. Stat. Soc. Ser. B. Stat. Methodol.}, 81(3):649--672.

\bibitem[Berdugo et~al., 2020]{Mae20}
Berdugo, M., Delgado-Baquerizo, M., Soliveres, S., Hernández-Clemente, R.,
  Zhao, Y., Gaitán, J.~J., Gross, N., Saiz, H., Maire, V., Lehmann, A.,
  Rillig, M.~C., Solé, R.~V., and Maestre, F.~T. (2020).
\newblock Global ecosystem thresholds driven by aridity.
\newblock {\em Science}, 367(6479):787--790.

\bibitem[Bertsimas and Van~Parys, 2020]{BeVa20}
Bertsimas, D. and Van~Parys, B. (2020).
\newblock Sparse high-dimensional regression: exact scalable algorithms and
  phase transitions.
\newblock {\em Ann. Statist.}, 48(1):300--323.

\bibitem[Boukerche et~al., 2020]{boukerche2020outlier}
Boukerche, A., Zheng, L., and Alfandi, O. (2020).
\newblock Outlier detection: Methods, models, and classification.
\newblock {\em ACM Comput. Surv.}, 53(3):1--37.

\bibitem[Boysen et~al., 2009]{boysen2009consistencies}
Boysen, L., Kempe, A., Liebscher, V., Munk, A., and Wittich, O. (2009).
\newblock Consistencies and rates of convergence of jump-penalized least
  squares estimators.
\newblock {\em Ann. Statist.}, 37(1):157--183.

\bibitem[Brodal et~al., 2011]{brodal2011towards}
Brodal, G., Gfeller, B., J{\o}rgensen, A., and Sanders, P. (2011).
\newblock Towards optimal range medians.
\newblock {\em Theor. Comput. Sci.}, 412(24):2588--2601.

\bibitem[Castro et~al., 2016]{castro2016wavelet}
Castro, R., Lehmann, N., P{\'e}rez, J., and Subercaseaux, B. (2016).
\newblock Wavelet trees for competitive programming.
\newblock {\em Olymp. Inform.}, 10:19--37.

\bibitem[Cormen et~al., 2009]{cormen2022introduction}
Cormen, T.~H., Leiserson, C.~E., Rivest, R.~L., and Stein, C. (2009).
\newblock {\em Introduction to algorithms}.
\newblock MIT Press, Cambridge, MA, third edition.

\bibitem[Dehning et~al., 2020]{Pries20}
Dehning, J., Zierenberg, J., Spitzner, F.~P., Wibral, M., Neto, J.~P., Wilczek,
  M., and Priesemann, V. (2020).
\newblock Inferring change points in the spread of covid-19 reveals the
  effectiveness of interventions.
\newblock {\em Science}, 369(6500):eabb9789.

\bibitem[Dette et~al., 2020]{DeEV20}
Dette, H., Eckle, T., and Vetter, M. (2020).
\newblock Multiscale change point detection for dependent data.
\newblock {\em Scand. J. Stat.}, 47(4):1243--1274.

\bibitem[Donoho and Johnstone, 1994]{donoho1994ideal}
Donoho, D.~L. and Johnstone, I.~M. (1994).
\newblock Ideal spatial adaptation by wavelet shrinkage.
\newblock {\em Biometrika}, 81(3):425--455.

\bibitem[Du et~al., 2016]{du2016stepwise}
Du, C., Kao, C.-L.~M., and Kou, S.~C. (2016).
\newblock Stepwise signal extraction via marginal likelihood.
\newblock {\em J. Amer. Statist. Assoc.}, 111(513):314--330.

\bibitem[D\"{u}mbgen, 1991]{dumbgen91}
D\"{u}mbgen, L. (1991).
\newblock The asymptotic behavior of some nonparametric change-point
  estimators.
\newblock {\em Ann. Statist.}, 19(3):1471--1495.

\bibitem[D\"{u}mbgen et~al., 2006]{dumbgen2006limit}
D\"{u}mbgen, L., Piterbarg, V.~I., and Zholud, D. (2006).
\newblock On the limit distribution of multiscale test statistics for
  nonparametric curve estimation.
\newblock {\em Math. Methods Statist.}, 15(1):20--25.

\bibitem[D\"{u}mbgen and Spokoiny, 2001]{dumbgen2001multiscale}
D\"{u}mbgen, L. and Spokoiny, V.~G. (2001).
\newblock Multiscale testing of qualitative hypotheses.
\newblock {\em Ann. Statist.}, 29(1):124--152.

\bibitem[Eichinger and Kirch, 2018]{mosum18}
Eichinger, B. and Kirch, C. (2018).
\newblock A {MOSUM} procedure for the estimation of multiple random change
  points.
\newblock {\em Bernoulli}, 24(1):526--564.

\bibitem[Ermshaus et~al., 2023]{ESL23}
Ermshaus, A., Sch{\"a}fer, P., and Leser, U. (2023).
\newblock Window size selection in unsupervised time series analytics: A review
  and benchmark.
\newblock In Guyet, T. and et~al., editors, {\em Adv. Anal. Learn. Temporal
  Data}, pages 83--101.

\bibitem[Fan et~al., 2015]{fan2015identifying}
Fan, Z., Dror, R.~O., Mildorf, T.~J., Piana, S., and Shaw, D.~E. (2015).
\newblock Identifying localized changes in large systems: Change-point
  detection for biomolecular simulations.
\newblock {\em Proc. Natl. Acad. Sci. U.S.A.}, 112(24):7454--7459.

\bibitem[Fang et~al., 2020]{fang2020segmentation}
Fang, X., Li, J., and Siegmund, D. (2020).
\newblock Segmentation and estimation of change-point models: false positive
  control and confidence regions.
\newblock {\em Ann. Statist.}, 48(3):1615--1647.

\bibitem[Fearnhead, 2006]{fearnhead2006exact}
Fearnhead, P. (2006).
\newblock Exact and efficient {B}ayesian inference for multiple changepoint
  problems.
\newblock {\em Stat. Comput.}, 16(2):203--213.

\bibitem[Fearnhead and Rigaill, 2019]{fearnhead2019changepoint}
Fearnhead, P. and Rigaill, G. (2019).
\newblock Changepoint detection in the presence of outliers.
\newblock {\em J. Amer. Statist. Assoc.}, 114(525):169--183.

\bibitem[Frick et~al., 2014]{frick2014multiscale}
Frick, K., Munk, A., and Sieling, H. (2014).
\newblock Multiscale change point inference.
\newblock {\em J. R. Stat. Soc. Ser. B. Stat. Methodol.}, 76(3):495--580.

\bibitem[Fryzlewicz, 2014]{fryzlewicz2014wild}
Fryzlewicz, P. (2014).
\newblock Wild binary segmentation for multiple change-point detection.
\newblock {\em Ann. Statist.}, 42(6):2243--2281.

\bibitem[Fryzlewicz, 2024]{fryzlewicz2024robust}
Fryzlewicz, P. (2024).
\newblock Robust narrowest significance pursuit: inference for multiple
  change-points in the median.
\newblock {\em J. Bus. Econom. Statist.}, 42(4).

\bibitem[Futschik et~al., 2014]{Futschik14}
Futschik, A., Hotz, T., Munk, A., and Sieling, H. (2014).
\newblock {Multiscale DNA partitioning: statistical evidence for segments}.
\newblock {\em Bioinformatics}, 30(16):2255--2262.

\bibitem[Gin\'{e} and Nickl, 2016]{gine2021mathematical}
Gin\'{e}, E. and Nickl, R. (2016).
\newblock {\em Mathematical foundations of infinite-dimensional statistical
  models}.
\newblock Cambridge University Press, New York.

\bibitem[Grossi et~al., 2003]{grossi2003high}
Grossi, R., Gupta, A., and Vitter, J.~S. (2003).
\newblock High-order entropy-compressed text indexes.
\newblock In {\em Proc. 14th ACM-SIAM Symp. Discrete Algorithms (SODA)}, pages
  841--850. Society for Industrial and Applied Mathematics.

\bibitem[Haltmeier et~al., 2022]{HaLiMu22}
Haltmeier, M., Li, H., and Munk, A. (2022).
\newblock A variational view on statistical multiscale estimation.
\newblock {\em Annu. Rev. Stat. Appl.}, 9:343--372.

\bibitem[Harchaoui and L\'{e}vy-Leduc, 2010]{HaLe10}
Harchaoui, Z. and L\'{e}vy-Leduc, C. (2010).
\newblock Multiple change-point estimation with a total variation penalty.
\newblock {\em J. Amer. Statist. Assoc.}, 105(492):1480--1493.

\bibitem[Haynes et~al., 2017]{haynes2017computationally}
Haynes, K., Fearnhead, P., and Eckley, I.~A. (2017).
\newblock A computationally efficient nonparametric approach for changepoint
  detection.
\newblock {\em Stat. Comput.}, 27(5):1293--1305.

\bibitem[Horv\'{a}th et~al., 2021]{HoKW21}
Horv\'{a}th, L., Kokoszka, P., and Wang, S. (2021).
\newblock Monitoring for a change point in a sequence of distributions.
\newblock {\em Ann. Statist.}, 49(4):2271--2291.

\bibitem[Kabluchko, 2011]{kabluchko2011extremes}
Kabluchko, Z. (2011).
\newblock Extremes of the standardized {G}aussian noise.
\newblock {\em Stochastic Process. Appl.}, 121(3):515--533.

\bibitem[Killick et~al., 2012]{killick2012optimal}
Killick, R., Fearnhead, P., and Eckley, I.~A. (2012).
\newblock Optimal detection of changepoints with a linear computational cost.
\newblock {\em J. Amer. Statist. Assoc.}, 107(500):1590--1598.

\bibitem[Koenker, 2005]{koenker_2005}
Koenker, R. (2005).
\newblock {\em Quantile Regression}.
\newblock Cambridge University Press.

\bibitem[Kov\'{a}cs et~al., 2023]{kovacs2023seeded}
Kov\'{a}cs, S., B\"{u}hlmann, P., Li, H., and Munk, A. (2023).
\newblock Seeded binary segmentation: a general methodology for fast and
  optimal changepoint detection.
\newblock {\em Biometrika}, 110(1):249--256.

\bibitem[Li et~al., 2019]{LiGM19}
Li, H., Guo, Q., and Munk, A. (2019).
\newblock Multiscale change-point segmentation: beyond step functions.
\newblock {\em Electron. J. Stat.}, 13(2):3254--3296.

\bibitem[Li et~al., 2016]{li2016fdr}
Li, H., Munk, A., and Sieling, H. (2016).
\newblock F{DR}-control in multiscale change-point segmentation.
\newblock {\em Electron. J. Stat.}, 10(1):918--959.

\bibitem[Liu et~al., 2013]{liu2013change}
Liu, S., Yamada, M., Collier, N., and Sugiyama, M. (2013).
\newblock Change-point detection in time-series data by relative density-ratio
  estimation.
\newblock {\em Neural Netw.}, 43:72--83.

\bibitem[Madrid~Padilla et~al., 2021]{madrid2021optimal}
Madrid~Padilla, O.~H., Yu, Y., Wang, D., and Rinaldo, A. (2021).
\newblock Optimal nonparametric change point analysis.
\newblock {\em Electron. J. Stat.}, 15(1):1154--1201.

\bibitem[Maidstone et~al., 2017]{maidstone2017optimal}
Maidstone, R., Hocking, T., Rigaill, G., and Fearnhead, P. (2017).
\newblock On optimal multiple changepoint algorithms for large data.
\newblock {\em Stat. Comput.}, 27(2):519--533.

\bibitem[Matteson and James, 2014]{MaJa14}
Matteson, D.~S. and James, N.~A. (2014).
\newblock A nonparametric approach for multiple change point analysis of
  multivariate data.
\newblock {\em J. Amer. Statist. Assoc.}, 109(505):334--345.

\bibitem[Navarro, 2014]{Nav14}
Navarro, G. (2014).
\newblock Wavelet trees for all.
\newblock {\em J. Discrete Algorithms}, 25:2--20.

\bibitem[Niu et~al., 2016]{NHZ16}
Niu, Y.~S., Hao, N., and Zhang, H. (2016).
\newblock Multiple change-point detection: a selective overview.
\newblock {\em Statist. Sci.}, 31(4):611--623.

\bibitem[Page, 1955]{page1955test}
Page, E.~S. (1955).
\newblock A test for a change in a parameter occurring at an unknown point.
\newblock {\em Biometrika}, 42:523--527.

\bibitem[Pein et~al., 2020]{pein2020heterogeneous}
Pein, F., Bartsch, A., Steinem, C., and Munk, A. (2020).
\newblock Heterogeneous idealization of ion channel recordings--open channel
  noise.
\newblock {\em IEEE Trans. Nanobioscience}, 20(1):57--78.

\bibitem[Pein et~al., 2017]{pein2017heterogeneous}
Pein, F., Sieling, H., and Munk, A. (2017).
\newblock Heterogeneous change point inference.
\newblock {\em J. R. Stat. Soc. Ser. B. Stat. Methodol.}, 79(4):1207--1227.

\bibitem[Pein et~al., 2018]{pein2018fully}
Pein, F., Tecuapetla-G{\'o}mez, I., Sch{\"u}tte, O.~M., Steinem, C., and Munk,
  A. (2018).
\newblock Fully automatic multiresolution idealization for filtered ion channel
  recordings: flickering event detection.
\newblock {\em IEEE Trans. Nanobioscience}, 17(3):300--320.

\bibitem[Requadt et~al., 2025]{IDC25}
Requadt, R., Fink, M., Kubica, P., Steinem, C., Munk, A., and Li, H. (2025).
\newblock Robust inference of cooperative behavior of multiple ion channels in
  voltage-clamp recordings.
\newblock {\em IEEE Trans. Nanobioscience}, 24(3):305--317.

\bibitem[Rosenberg and Hirschberg, 2007]{rosenberg2007v}
Rosenberg, A. and Hirschberg, J. (2007).
\newblock V-measure: A conditional entropy-based external cluster evaluation
  measure.
\newblock In {\em Proc. EMNLP-CoNLL 2007}, pages 410--420.

\bibitem[Ruanaidh and Fitzgerald, 2012]{ruanaidh2012numerical}
Ruanaidh, J. J.~O. and Fitzgerald, W.~J. (2012).
\newblock {\em Numerical Bayesian methods applied to signal processing}.
\newblock Springer Science \& Business Media.

\bibitem[Russell and Rambaccussing, 2019]{russell2019breaks}
Russell, B. and Rambaccussing, D. (2019).
\newblock Breaks and the statistical process of inflation: the case of
  estimating the ‘modern’long-run phillips curve.
\newblock {\em Empir. Econ.}, 56:1455--1475.

\bibitem[Sakmann, 2013]{sakmann2013single}
Sakmann, B. (2013).
\newblock {\em Single-channel recording}.
\newblock Springer Science \& Business Media.

\bibitem[Sen and Srivastava, 1975]{sen1975tests}
Sen, A. and Srivastava, M.~S. (1975).
\newblock On tests for detecting change in mean when variance is unknown.
\newblock {\em Ann. Inst. Statist. Math.}, 27(3):479--486.

\bibitem[Truong et~al., 2020]{TOV20}
Truong, C., Oudre, L., and Vayatis, N. (2020).
\newblock Selective review of offline change point detection methods.
\newblock {\em Signal Process.}, 167:107299.

\bibitem[Tsybakov, 2009]{tsybakov2009introduction}
Tsybakov, A.~B. (2009).
\newblock {\em Introduction to nonparametric estimation}.
\newblock Springer, New York.

\bibitem[Vanegas et~al., 2022]{jula2022multiscale}
Vanegas, L.~J., Behr, M., and Munk, A. (2022).
\newblock Multiscale quantile segmentation.
\newblock {\em J. Amer. Statist. Assoc.}, 117(539):1384--1397.

\bibitem[Venkataramanan et~al., 1998]{ven98ii}
Venkataramanan, L., Kuc, R., and Sigworth, F.~J. (1998).
\newblock Identification of hidden markov models for ion channel currents. {II.
  S}tate-dependent excess noise.
\newblock {\em IEEE Trans. Signal Process.}, 46(7):1916--1929.

\bibitem[Verzelen et~al., 2023]{verzelen2023optimal}
Verzelen, N., Verzelen, N., Fromont, M., Fromont, M., Lerasle, M., Lerasle, M.,
  Reynaud-Bouret, P., and Reynaud-Bouret, P. (2023).
\newblock Optimal change-point detection and localization.
\newblock {\em Ann. Statist.}, 51(4):1586--1610.

\bibitem[Vostrikova, 1981]{vostrikova1981detecting}
Vostrikova, L.~Y. (1981).
\newblock Detecting “disorder” in multidimensional random processes.
\newblock In {\em Doklady akademii nauk}, volume 259, pages 270--274. Russian
  Academy of Sciences.

\bibitem[Wald, 1945]{wald1945sequential}
Wald, A. (1945).
\newblock Sequential tests of statistical hypotheses.
\newblock {\em Ann. Math. Statist.}, 16:117--186.

\bibitem[Walther and Perry, 2022]{Walther2022Calibrating}
Walther, G. and Perry, A. (2022).
\newblock Calibrating the scan statistic: finite sample performance versus
  asymptotics.
\newblock {\em J. R. Stat. Soc. Ser. B. Stat. Methodol.}, 84(5):1608--1639.

\bibitem[Wu, 2005]{Wu05}
Wu, W.~B. (2005).
\newblock Nonlinear system theory: another look at dependence.
\newblock {\em Proc. Natl. Acad. Sci. USA}, 102(40):14150--14154.

\bibitem[Wyse et~al., 2011]{wyse2011approximate}
Wyse, J., Friel, N., and Rue, H.~v. (2011).
\newblock Approximate simulation-free {B}ayesian inference for multiple
  changepoint models with dependence within segments.
\newblock {\em Bayesian Anal.}, 6(4):501--527.

\bibitem[Zou et~al., 2014]{zou2014nonparametric}
Zou, C., Yin, G., Feng, L., and Wang, Z. (2014).
\newblock Nonparametric maximum likelihood approach to multiple change-point
  problems.
\newblock {\em Ann. Statist.}, 42(3):970--1002.

\bibitem[Zou and Yuan, 2008]{zou2008composite}
Zou, H. and Yuan, M. (2008).
\newblock Composite quantile regression and the oracle model selection theory.
\newblock {\em Ann. Statist.}, 36(3):1108--1126.

\end{thebibliography}


\begin{appendix}

\section{Proofs}\label{A1}
We note that the statistical guarantees of MUSCLE rely solely on the fact that 
\[
W_i\bigl(Z_i, f(x_i)\bigr) = \mathbbm{1}_{\{Z_i \leq f(x_i)\}}= \mathbbm{1}_{\{\varepsilon_i \le 0\}}, \quad i = 1, \ldots, n,
\]
are i.i.d.~Bernoulli random variables. 
Therefore, the assumption of statistical independence for the observations $Z_i$ can be relaxed to the independence of the indicators 
$\mathbbm{1}_{\{Z_i \leq f(x_i)\}} = \mathbbm{1}_{\{\varepsilon_i \leq 0\}}$.

\subsection{Technical results}
Let $W_1,W_2,\dots$ are i.i.d.\ Bernoulli random variables with success probability $\beta$. For any $k,m,i,j\in \mathbb{N}$, $x\in [0,k]$, $1\leq j-i\leq m$, we define
\begin{align*}
    A_k(x) &\coloneqq \frac{x}{k}\log\left(\frac{x}{k}\cdot\frac{1}{\beta}\right) +\left(1-\frac{x}{k}\right) \log\left(\left(1-\frac{x}{k}\right)\cdot\frac{1}{1-\beta}\right) \\
    B_{m,i,j} &\coloneqq \sqrt{2(j-i)A_{j-i}\left(\sum_{k=i+1}^{j}W_k\right)}-\sqrt{2\log\left(\frac{em}{j-i}\right)}.
\end{align*}
 Define a series of random times: $T_ 0 \coloneqq 0$ and
\[
    T_s \coloneqq \max \left\{l \,  :  \,\underset{T_{s-1}<i< j\leq l}{\max} B_{l-T_{s-1},i,j}-q_{\alpha}(l-T_{s-1})\leq 0\right\},
\]
for $s\geq 1$, where the quantiles $q_{\alpha}(\cdot)$ are given by \eqref{e: q_m}. Furthermore, let 
\begin{equation}\label{e: S}
    S \coloneqq \max \{s : T_s < n\}.
\end{equation}

\begin{proposition}\label[proposition]{p: tail bound}
Let $S$ be in \eqref{e: S} and $\alpha\in (0,1)$. Then for any $s\in \mathbb{N}$,
\(    \mathbb{P}(S\geq s) \leq \alpha^s.\)
Moreover, we have
\(    \mathbb{E}(S)\leq \frac{\alpha}{1-\alpha}.\)
\end{proposition}

\begin{proof}
Let $s\geq 1$, $n \geq 1$, $1 \leq t_1<t_2< \cdots < t_{s-1}$. It follows from the definitions of $q_{\alpha}(m)$ in \eqref{e: q_m} and $T_{s}$ for $s\geq 1$ that 
\begin{align}
    \label{ieq:cond tail}
    &\mathbb{P}[T_s-T_{s-1}\geq n \mid T_{s-1} = t_{s-1},\dots, T_1 = t_1]\nonumber\\
    = \, &\mathbb{P}\left[\max_{t_{s-1}<i<j\leq l}B_{l-k,i,j}\le q_{\alpha}(l-t_{s-1}),\; \text{ for some } l \geq n+t_{s-1}\right] \geq 1-\alpha.
\end{align}
Since
\(
    \mathbb{P}(S\geq s) = \mathbb{P}(T_s<n) = \mathbb{P}\left(\sum_{i=1}^s T_s<n\right) 
    \leq \mathbb{P}(T_{i}-T_{i-1} <n \text{ for all } i= 1,\dots, s),
\)
we show by introduction over $s$ that for any $n\geq 1$, $\mathbb{P}(T_{i}-T_{i-1} <n \text{ for all } i= 1,\dots, s)\leq \alpha^s$.
For $s=1$, by definitions of $q_{\alpha}(m)$ in \eqref{e: q_m} we have $\mathbb{P}(T_1-T_0 <n) = \mathbb{P}(T_1<n) = 1- \mathbb{P}(T_1\geq n) \leq 1-(1-\alpha)= \alpha$.
Now assume that the statement holds for $s-1$, i.e., $\mathbb{P}(T_{i}-T_{i-1} <n \text{ for all } i= 1,\dots, s-1)\leq \alpha^{s-1}$. Then
\begin{align*}
    &\mathbb{P}(T_{i}-T_{i-1} <n \text{ for all } i= 1,\dots, s)\\
    =\,&\sum_{\substack{t_{i}-t_{i-1}<n\\ i = 1,\dots,s-1}}\mathbb{P}\Bigl\{T_s-T_{s-1}<n,T_{s-1}-T_{s-2}=t_{s-1}-t_{s-2},\dots,T_{2}-T_{1}= t_{2}-t_{1},T_{1}=t_1\Bigr\}\\
    =\,&\sum_{\substack{t_{i}-t_{i-1}<n\\ i = 1,\dots,s-1}}\mathbb{P}\Bigl\{T_s-T_{s-1}<n\, \Big\vert \, T_{s-1}-T_{s-2}=t_{s-1}-t_{s-2},\dots, T_{1}=t_1\Bigr\}\\
    &\mbox{} \qquad\qquad\qquad\qquad\times \mathbb{P}\Bigl\{T_{s-1}-T_{s-2}=t_{s-1}-t_{s-2},\dots,\, T_{1}=t_1\Bigr\}\\
    \leq\, & \sum_{\substack{t_{i}-t_{i-1}<n\\ i = 1,\dots,s-1}} \alpha \cdot \mathbb{P}\Bigl\{T_{s-1}-T_{s-2}=t_{s-1}-t_{s-2}, \dots, T_{1}=t_1\Bigr\}\\
    =\, & \alpha \cdot \mathbb{P}(T_{i}-T_{i-1} <n \text{ for all } i= 1,\dots, s-1) \leq \alpha^{s},
\end{align*}
where the first inequality follows from \eqref{ieq:cond tail} and the last inequality is implied by induction assumption. Therefore, we obtain $\mathbb{P}(T_{i}-T_{i-1} <n \text{ for all } i= 1,\dots, s)\leq \alpha^s$ for all $s\geq 1$, and hence $P(S\geq s)\leq \alpha^s$. In particular,
\begin{equation*}
    \mathbb{E}(S)= \sum_{s=1}^{\infty}\mathbb{P}(S\geq s) \leq \sum_{s=1}^{\infty} \alpha^s \leq \frac{\alpha}{1-\alpha},
\end{equation*}
which concludes the proof.
\end{proof}
\begin{proposition}
\label[proposition]{prop: over estimation of zero}    
Let $\beta \in (0,1)$ and suppose that $Z_1,\dots,Z_n$ have the common $\beta$-quantile, i.e., ${\mathbb{P}(Z_i \leq f(x_i))=\beta}$ with ${f} \equiv \theta$ for some $\theta \in \mathbb{R}$. Then, for any $k\geq 1$, the estimated number of change point $\widehat{K}$ by MUSCLE satisfies
\[
    \mathbb{P}\left(\widehat{K}\geq k\right) \leq \alpha^k\quad
 \text{and} \quad  \mathbb{E}\left(\widehat{K}\right) \leq \frac{\alpha}{1-\alpha}.
\]
\end{proposition}

\begin{proof}
With loss of generality, we assume that ${f} = 0$. Suppose that there exists an estimator $\widetilde{f} = \sum_{k=0}^{\widetilde{K}}\theta_k I_k$ with $\theta_{k-1}\neq \theta_{k}$ for $k = 1,\dots,K$, which has $\widetilde{K}$ change points and satisfies the multiscale side-constrains of $C_{\tilde{K}}$ in \eqref{e: C_k}. Then $\widetilde{f}$ has exactly $\widetilde{K}$ false discoveries. Since the MUSCLE minimizes the number of estimated change points fulfilling the multiscale side-constrains, we have $\widehat{K} \leq \widetilde{K}$. We now construct such an estimator $\widetilde{f}$ as follows.\\
Set $T_0 = 0$. We fit the observations $\{Z_i : i\geq 1 \}$ using the constant value $0$ up to $T_1$, defined as the largest index $j$ such that $0$ satisfies the multiscale constraint on the interval $[0, j/n)$. We then fit the subsequent observations $\{Z_{T_1 + i} : i \geq 1\}$ using $0$ until $T_2$ , the largest index $j$ such that $0$ satisfies the multiscale constraint on $[T_1/n, j/n)$. Repeating this procedure until all observations are fitted yields a series of segments. For ensuring the resulting estimator has exactly $\widetilde{K}$ change points, we slightly perturb the segment values so that the multiscale constraints remain satisfied on each segment. More precisely,
Let $W_i = \mathbbm{1}_{\{Z_i\leq 0\}}$ for $i=1,\dots,n$ and 
\begin{equation*}
    \tilde{\theta} =
    \begin{cases}
    \min\{Z_i  :   Z_i>0, \, i=1,\dots,n\}/2,&\quad \text{if there is a } \,Z_i >0,\\
    1, &\quad \text{otherwise.}
    \end{cases}
\end{equation*}
Then $\tilde{\theta}>0$. Moreover, by the definition of $\tilde{\theta}$, for any $l\in \mathbb{N}$ the events $\{Z_i\leq 0 \}$ and $\{Z_i \leq \tilde{\theta}/l\}$ are equivalent. 
Define the stopping time 
\(
    T_1 = \max \left\{j : T_{[0/n,j/n)}(W,0) \leq q_{\alpha}(j) \right\}. 
\)
Set $I_0 = [0/n, T_1/n)$. Then by definition of $T_1$ we have $T_{\mathring{I_0}}(W,0) \leq q_{\alpha}(T_1)$. We defined the modification $\widetilde{W}_i = \mathbbm{1}_{\{Z_i\leq \tilde{\theta}\}}$ for all $i/n\in I_0$, then $W_i = \widetilde{W}_i$.
Thus,
\(
    T_{\mathring{I_0}}(\widetilde{W},\tilde{\theta}) = T_{\mathring{I_0}}(W,0) \leq q_{\alpha}(T_1),
\)
which indicates that $\tilde{\theta}$ fulfils the multiscale side-constrain on $\mathring{I_0}$ and therefore, we can estimate ${f}|_{I_0}$ with $\tilde{\theta}$.
We continue this procedure until obtain an estimator for whole ${f}$. More precisely, for any $k\in \mathbb{N}$, let 
\begin{equation*}
    T_k = \max \left\{j : T_{[T_{k-1}/n,j/n)}(W,0) \leq q_{\alpha}(j-T_{k-1}) \right\}, 
\end{equation*}
and $I_{k-1} =[T_{k-1}/n,T_{k}/n)$. The restriction ${f}|_{I_{k-1}}$ can be estimated by $\tilde{\theta}/k$. Finally, let 
\(
    \widetilde{K} = \max\{k : T_k < n\},
\)
and estimation is given by
\(
    \widetilde{f} = \sum_{k=0}^{\tilde{K}}\frac{\tilde{\theta}}{k+1} \mathbbm{1}_{I_k}.
\)
Since $\tilde{f}$ is piecewise different, it has exactly $\widetilde{K}$ change points. Thus, applying \Cref{p: tail bound}, we have
\(
\mathbb{P}\left(\widehat{K}\geq k\right) \leq \mathbb{P}\left(\widetilde{K}\geq k\right) \leq \alpha^k, 
\)
and 
\(
\mathbb{E}\left(\widehat{K}\right) \leq \mathbb{E}\left(\widetilde{K}\right)\leq \frac{\alpha}{1-\alpha}.
\)
\end{proof}

\subsection{Proofs of Theorems \ref{t: under}--\ref{t: Hausdorff}}
\begin{lemma}\label[lemma]{l: ub of entropy}
For any $p,q\in(0,1)$,
\begin{equation*}
    2(p-q)^2 \leq p\log\frac{p}{q}+(1-p)\log\frac{1-p}{1-q}\leq \frac{(p-q)^2}{q(1-q)}.
\end{equation*}
\begin{proof}
The first inequality follows from Pinsker's inequality \citep[Lemma 2.5]{tsybakov2009introduction} and \citet[Theorem A.3]{jula2022multiscale}. For the second one, by  $x(x-1)\geq x\log(x)$ for $x>0$, we have 
\begin{equation*}
            \frac{p}{q}\log\frac{p}{q} \leq \frac{p}{q}\left(\frac{p}{q}-1\right) \quad\text{and}\quad
        \frac{1-p}{1-q}\log\frac{1-p}{1-q} \leq \frac{1-p}{1-q}\left(\frac{1-p}{1-q}-1\right).
\end{equation*}
Summing them up gives us the second inequality.
\end{proof}
\end{lemma}

\begin{lemma}\label{l: ub q_n}
Let $\beta \in (0,1)$ be fixed. There exists a constant $C$ such that for all $\alpha \in (0,1)$,
\begin{equation*}
     \underset{\abs{I}\geq 1}{\sup} \,q_{\alpha}(\abs{I}) \leq C+ 2\sqrt{2\log \left(\frac{2}{\alpha}\right)},
\end{equation*}
where $q_{\alpha}(\abs{I})$ is given in \eqref{e: q_m}.
\end{lemma}

\begin{proof}
Recall from \eqref{e: T_I} that 
\begin{equation*}
T_{I}(Z,\theta)\; = \; \underset{J\subseteq {I}}{\max}\sqrt{2L_J (Z,\theta)} - \sqrt{2\log\frac{e\abs{I} }{\abs{J}}},
\end{equation*}
with 
\(
    L_J(Z,\theta) =  \abs{J}\bigl[\overline{W}_J \log (\frac{\overline{W}_J}{\beta})+(1-\overline{W}_J)\log (\frac{1-\overline{W}_J}{1-\beta})\bigr],
\)
$\overline{W}_J = \left(\sum_{i\in J}W_i\right)/\abs{J}$ and $W_i = \mathbbm{1}_{\{Z_i\leq \theta \}}\overset{i.i.d.}{\sim} {\rm Ber}(\beta)$ under $H_I$. Since $T_I(Z,\theta)$ is determined by $W$, in this proof we use $T_I(W)$ instead.
For any $J\subseteq I$, define $a_J\in \mathbb{R}^{\abs{I}}$ by
\begin{equation*}
    a_{J,k} = 
    \begin{cases}
    \frac{1}{\sqrt{\abs{J}}}, &\quad \text{if } k\in J,\\
    0,&\quad \text{otherwise},
    \end{cases}
\end{equation*}
where $a_{J,k}$ denotes the $k$-th component of $a_{J}$. Set $A=\{a_{J}:J\subseteq I\}$ and $\lambda_a = \sqrt{2\log en \norm{a}_{\infty}^2}$. We further define $R_I(W)$ and  $S_I(W)$ by
\begin{equation*}
    R_I(W) = \underset{a\in A}{\max}\sqrt{2}\abs{a^{\top}(W-\beta)} - \lambda_a \quad\text{and}\quad
    S_I(W) = \underset{a\in A}{\max}\sqrt{\frac{2}{\beta(1-\beta)}}\abs{a^{\top}(W-\beta)} - \lambda_a.
\end{equation*}
Then, by Lemma \ref{l: ub of entropy}, we have 
\(
    R_I\leq T_I\leq S_I.
\)
For any $x,y\in \{0,1\}^{\abs{I}}$ and \emph{weight} $a\in \mathbbm{R}_{\geq 0}^{\abs{I}}$, the weighted Hamming distance $d_a$ is given by
\(
    d_a(x,y) = \sum_{i=1}^{\abs{I}} a_i \mathbbm{1}_{\{x_i \neq y_i\}}.
\)
We show that $T_I$ has the following Lipschitz property for the distance $d_a$ with Lipschitz constant $L = \sqrt{2}$: For any $x\in \{0,1\}^{\abs{I}}$, there exists a $a=a(x)\in \mathbbm{R}_{\geq 0}^{\abs{I}}$ with Euclidean norm $\norm{a}=1$ such that 
\(
    T_I(x)-T_I(y)\leq L d_a(x,y), \text{ for all } y\in\{0,1\}^{\abs{I}}.
\)
For any $x,y\in \{0,1\}^{\abs{I}}$, let $\tilde{a} = \tilde{a}(x)$ denote the maximizer of $S_I(x)$. Then
\begin{align*}
    &T_I(x)-T_I(y) 
    \leq S_I(x)-R_I(y) \\
    =\;&  \underset{a\in A}{\max}\left(\sqrt{\frac{2}{\beta(1-\beta)}}\,\abs{a^{\top}(x-\beta)}-\lambda_a\right) - \underset{a\in A}{\max}\left(\sqrt{2}\,\abs{a^{\top}(y-\beta)}-\lambda_a\right) \\
    =\;&  \sqrt{\frac{2}{\beta(1-\beta)}}\,\abs{\tilde{a}^{\top}(x-\beta)}-\lambda_{\tilde{a}}-\underset{a\in A}{\max}\left(\sqrt{2}\,\abs{\tilde{a}^{\top}(y-\beta)}-\lambda_a\right) \\
    \leq\;&  \sqrt{2}\,\abs{\tilde{a}^{\top}(x-\beta)}-\lambda_{\tilde{a}}-\sqrt{2}\,\abs{\tilde{a}^{\top}(y-\beta)}+\lambda_{\tilde{a}} \\
    =\;&  \sqrt{2}\,\abs{\tilde{a}^{\top}(x-\beta)}-\sqrt{2}\,\abs{\tilde{a}^{\top}(y-\beta)} \\
    \leq \; &  \sqrt{2}\, \abs{\tilde{a}^{\top}(x-y)} = \sqrt{2}\,\sum_{i=1}^{\abs{I}} \tilde{a}_i \mathbbm{1}_{\{x_i \neq y_i\}}= \sqrt{2}\,d_{\tilde{a}}(x,y),
\end{align*}
where the second last equality follows since $x,y\in \{0,1\}^{\abs{I}}$. Let $M_I$ denote the median of $T_I$. By \citet{frick2014multiscale}, \citet{dumbgen2001multiscale} and \citet{dumbgen2006limit}, the random sequence $T_I$ converges in distribution to an almost surely finite random variable $M$ as $\abs{I}\to \infty$. Thus, median of $M$ is finite and $C\coloneqq \sup_{\abs{I}\in \mathbb{N}}M_I < \infty$. Applying \citet[Corollary 3.2.5]{gine2021mathematical}, we have 
\(
    \mathbb{P}(T_I-M_I \geq t) \leq 2 \exp{\left(-t^2/8\right)}    
\)
for any $t>0$. By choosing $t = 2\sqrt{2\log (2/\alpha)}$, we have 
\begin{equation*}
    q_{\alpha}(\abs{I}) \leq M_I + 2\sqrt{2 \log\frac{2}{\alpha}} \leq C + 2\sqrt{2 \log\frac{2}{\alpha}}
\end{equation*}
for any $\abs{I}\in \mathbb{N}$.
\end{proof}

\begin{proof}[Proof of Theorem \ref{t: under}]
Recall from \cref{S: 3.1} that 
\begin{equation*}
    I_k = \left(\tau_k-\tfrac{\lambda_k}{2},\tau_k +\tfrac{\lambda_k}{2}\right),
\end{equation*}
where $\lambda_k = \min\{\tau_k-\tau_{k-1},\tau_{k+1}-\tau_k\}$. By definition, these intervals are disjoint. Let $\theta_k^+=\max\{\theta_{k-1},\theta_k\}$ and $\theta_k^-=\min\{\theta_{k-1},\theta_k\}$ and spilt $I_k$ accordingly, i.e.,
\begin{equation*}
    I_k^+ = \{t\in I_k : f(t) = \theta_k^+\},\text{ and } I_k^- = \{t\in I_k : f(t) = \theta_k^-\}.
\end{equation*}
Assume that $\widehat{K}<K$ and let $\widehat{I}_0,\dots,\widehat{I}_{\widehat{K}}$ denote the estimated segments. Since $\widehat{K}<K$, there exists a interval $I_k$ such that $I_k$ is completely contained in some estimated segment, say $\widehat{I}_*$. This can be verified by contradiction: Suppose that there is no such interval. Then any interval $I_k$ contains at least one estimated change point. Since $I_k$'s are disjoint, then we have $\widehat{K}\geq K$, which contradicts to the underestimation assumption. Thus, there exists a $I_k$ is contained in some estimated segment $\widehat{I}_*$, i.e., no change point is detected on $I_k$.

We show that the probability that no change point can be detected on $I_k$ is sufficiently small. To this end, we consider the event that there exists a constant on $\hat{I}_*$ fulfilling the multiscale side-constrains on both $I_k^+$ and $I_k^-$. Define
\begin{equation*}
    C_k = \Bigg\{ \exists c\in \mathbb{R} : \sqrt{2L_{I_k^+}(Z,c)}-\sqrt{2\log\frac{e\abs{\hat{I}_*}}{\abs{I_k^+}}} \leq q_{\alpha}(\abs{\hat{I}_*}) \text{ and }
    \sqrt{2L_{I_k^-}(Z,c)}-\sqrt{2\log\frac{e\abs{\hat{I}_*}}{\abs{I_k^-}}} \leq q_{\alpha}(\abs{\hat{I}_*}) \Bigg\} 
\end{equation*}
Note that either $c\geq \theta_k^+ -\delta_k/2$ or $c\geq \theta_k^-+\delta_k/2$, we define 
\begin{align*}
    &C_k^+ = \Bigl\{ c : c \geq \theta_k^+ -{\delta_k}/{2} \;\text{ and }\;  \sqrt{2L_{I_k^+}(Z,c)}-\sqrt{2\log\frac{e\abs{\hat{I}_*}}{\abs{I_k^+}}} \leq q_{\alpha}(\abs{\hat{I}_*}) \Bigr\} \\
    \text{and }\;  &C_k^- = \Bigl\{ c : c \leq \theta_k^- +{\delta_k}/{2} \;\text{ and }\;  \sqrt{2L_{I_k^-}(Z,c)}-\sqrt{2\log\frac{e\abs{\hat{I}_*}}{\abs{I_k^-}}} \leq q_{\alpha}(\abs{\hat{I}_*}) \Bigr\}.
\end{align*}
Obviously, we have $C_k \subseteq C_k^+ \cup C_k^-$. Note that $C_k^+$ and $C_k^-$ are independent. We obtain 
\(
    \mathbb{P}(C_k) \leq 1-(1-\mathbb{P}(C_k^+))(1-\mathbb{P}(C_k^-)).
\)
We provide an upper bound of $\mathbb{P}(C_k^-)$, then the probability $\mathbb{P}(C_k^+)$ can be upper bounded in a similar way. For any $i\in I_k^-$, we have $F_i^{-1}(\beta)=\theta_k^-$. Further set $F^-_{\min,k} = \min_{i\in I_k^-} F_i$ and $\bar{F}_k^- = \left(\sum_{i\in I_k^-}F_i \right)/\abs{I_k^-}$. Let 
\(
    \zeta_{I_k^-} = \inf \bigl\{x\in \mathbb{R} : \sum_{i\in I_k^-}\mathbbm{1}_{\{Z_i \leq x\}}\geq \abs{I_s^-}\beta\bigr\}
\)
be the \emph{empirical quantile} of the observations on $I_k^-$. For all $c\geq \theta_k^- + \delta_k/2$, if $\zeta_{I_k^-}\leq \theta_k^- + \delta_k/2$, then 
\begin{equation}\label{e: beta}
    \beta \leq \overline{W}(Z,\zeta_{I_k^-})\leq \overline{W}\left(Z,\theta_k^-+\frac{\delta}{2}\right)\leq \overline{W}(Z,c),
\end{equation}
where $\overline{W}(Z,x) =\sum_{i\in I_k^-}\mathbbm{1}_{\{Z_i \leq x\}}/\abs{I_k^-}$.
The function 
\begin{equation*}
    g:x\mapsto \abs{I_k^-}\left(x\log\frac{x}{\beta}+(1-x)\log\frac{1-x}{1-\beta}\right)
\end{equation*}
is strictly convex with minimum $g(\beta) = 0$, hence $g$ is strictly increasing on $(\beta,1)$. Thus, by \eqref{e: beta} 
\(
    L_{I_k^-}(Z,c) = g\left(\overline{W}(Z,c)\right) \geq g\left(\overline{W}\left(Z,\theta_k^-+\frac{\delta_k}{2}\right)\right) = L_{I_k^-}\left(Z,\theta_k^-+\frac{\delta_k}{2}\right).
\)
Therefore, 
\begin{equation}\label{e:splitC}
\mathbb{P}(C_k^-)
    \leq  \mathbb{P}\bigl(C_k^- \cap \bigl\{\zeta_{I_k^-}\leq \theta_k^- + \frac{\delta_k}{2}\bigr\}\bigr) \\
    + \mathbb{P}\bigl(\zeta_{I_k^-}> \theta_k^- + \frac{\delta_k}{2}\bigr). 
    \end{equation}
By \citet[Theorem A.2]{jula2022multiscale}, the first term is bounded as
\begin{align*}
 \mathbb{P}\left(C_k^- \cap \left\{\zeta_{I_k^-}\leq \theta_k^- + \frac{\delta_k}{2}\right\}\right) 
    & \leq \mathbb{P}\left(\sqrt{2L_{I_k^-}\left(Z,\theta_k^-+\frac{\delta_k}{2}\right)}-\sqrt{2\log\frac{e\abs{\hat{I}_*}}{\abs{I_k^-}}} \leq q_{\alpha}(\abs{\hat{I}_*})\right) \\
   & \leq \mathbb{P}\left(L_{I_k^-}\left(Z,\theta_k^-+\frac{\delta_k}{2}\right)\leq \frac{1}{2}\left(q_{\alpha}(\abs{\hat{I}_*})+\sqrt{2\log\frac{e\abs{\hat{I}_*}}{\abs{I_k^-}}}\right)^2\right) \\
   & \leq  2\exp \left(-\frac{1}{2}\left(\sqrt{2n\lambda_k}\xi_{\bar{F}^-_{\beta}}(\frac{\delta_k}{2})-q_{\alpha}(\abs{\hat{I}_*})-\sqrt{2\log\frac{e\abs{\hat{I}_*}}{\abs{I_k^-}}}\right)^2_+\right).
\end{align*}
The second term in \eqref{e:splitC} is bounded as (via \citealp[Theorem A.3]{jula2022multiscale})
\begin{align*}
  \mathbb{P}\left(\zeta_{I_k^-}> \theta_k^- + \frac{\delta_k}{2}\right)
    \leq 2\exp \left(-n\lambda_k \xi_{F^-_{\min,\beta}}\left(\frac{\delta_k}{2}\right)^2\right).
\end{align*}
Note that the facts $\xi_k \leq \xi_{F^-_{\min,k},\beta} (\delta_k/2)\leq \xi_{\bar{F}_k^-,{\beta}}(\delta_k/2)$, $q_{\alpha}(\abs{\hat{I}_*})\leq \max_{m\leq n}q_{\alpha}(m)$ and $\abs{\hat{I}_*}\leq n$. Thus, we have
\(
    \mathbb{P}(C_k) \leq 1-(1-\mathbb{P}(C_k^+))(1-\mathbb{P}(C_k^-)) \leq 1 -\gamma_{n,k}^2,
\)
where 
\begin{equation*}
    \gamma_{n,k} = 1-2\exp{\left(-n\lambda_k\xi_k^2\right)} 
    -2\exp{\left(-\left( \sqrt{n\lambda_k}\xi_k-\frac{1}{\sqrt{2}}\max_{m\leq n}q_{\alpha}(m)-\sqrt{\log\frac{2e}{\lambda_k}}\right)_+^2 \right) }. 
\end{equation*}
If all of these $K$ change-points are detected, then $\widehat{K}\geq K$,
\begin{equation*}
    \mathbb{P}(\widehat{K}\geq K) \geq \mathbb{P}\left(\cap_{k=1}^K C_k^c\right) = \prod_{k=1}^{K}(1-\mathbb{P}(C_k))\geq \prod_{k=1}^{K} \gamma_{n,k}^2,
\end{equation*}
which proves the first claim. 
By the definitions of $\Lambda$ and $\Omega$ in \eqref{e: LTO}, we have for all $k=1,\dots,K$,
\begin{equation*}
    \gamma_{n,k}^2 \geq \Bigg( 1 -2\exp\left(-\left(\sqrt{n}\Omega -\frac{1}{\sqrt{2}}\underset{m\leq n}{\max}q_{\alpha}(m)-\sqrt{\log\frac{2e}{\Lambda}}\right)_+^2\right)
     -2\exp \left(-n\Omega^2\right) \bigg)^2 \eqqcolon \gamma_n.
\end{equation*}
Thus, applying Bernoulli's inequality, we have
\begin{align*}
    \mathbb{P}(\widehat{K}<K) & \leq 1-\gamma_n^K \\
    & \leq 4K \exp \left(-n\Omega^2\right)  
    +4K\exp\left(-\left(\sqrt{n}\Omega -\frac{1}{\sqrt{2}}\underset{m\leq n}{\max}q_{\alpha}(m)-\sqrt{\log\frac{2e}{\Lambda}}\right)_+^2\right).\qedhere
\end{align*}
\end{proof}
\begin{proof}[Proof of Theorem \ref{t: over}]
Let $T(f)$ denote the set of all change points of $f$. Note that
\begin{align*}
        & \mathbb{P}(\widehat{K}\geq K + k) \\
    \leq \, &  \mathbb{P}\Bigl(
    \max \, T_{I_i}(Z,\hat{\theta}_i)-q_{\alpha}(\abs{I_i})\geq 0, \forall \hat{f}\in \Sigma \text{ with }\abs{T(\hat{f})}\leq K+k-1\Bigr) \\
    \leq\,& \mathbb{P}\Bigl(
    \max \, T_{I_i}(Z,\hat{\theta}_i)-q_{\alpha}(\abs{I_i})\geq 0, \forall \hat{f}\in \Sigma \text{ with } T(f)\subseteq T(\hat{f}) \text{ and }\abs{T(\hat{f})}\leq K+k-1\Bigr)\\
    \leq\,& \, \mathbb{P}\Bigl(
    \max T_{I_i}(Z-f,\hat{\theta}_i-f|_{I_i})-q_{\alpha}(\abs{I_i})\geq 0, \forall \hat{f}\in \Sigma \text{ with }\abs{T(\hat{f}-f)}\leq k-1\Bigr). 
\end{align*}
Let $\widetilde{Z} = Z -f$, and each $\widetilde{Z}_i$ has $\beta$-quantile zero. The event in last lines implies that the estimator for zero function by MUSCLE has at least $k$ change points. Thus, it follows from \Cref{prop: over estimation of zero} that 
\(
    \mathbb{P}(\widehat{K}\geq K + k)\leq \mathbb{P}(T_k < n) \leq \alpha^k. 
\)
\end{proof}

\begin{proof}[Proof of Theorem \ref{t: consistency 1}]
It follows from Theorems \ref{t: under} and \ref{t: over} that
\begin{align*}
    & \mathbb{P}(\widehat{K}_n \neq K_n) \\
    \leq \, & 4K_n\Bigg(\exp\left(-\left(\sqrt{n}\Omega_n -\frac{1}{\sqrt{2}}\underset{m\leq n}{\max}q_{\alpha}(n)-\sqrt{\log\frac{2e}{\Lambda_n}}\right)_+^2\right) + \exp \left(-n\Omega_n^2\right) \Bigg) +\alpha_n  \\
    = \, &  4\exp{(-\Gamma_{1,n})} + 4\exp{(-\Gamma_{2,n})} +\alpha_n,
\end{align*}
where $\Gamma_{2,n} \coloneqq n\Omega_n^2 -\log K_n$ and 
\begin{align*}
    \Gamma_{1,n} \coloneqq \left(\sqrt{n}\Omega_n -\frac{1}{\sqrt{2}}\underset{m\leq n}{\max}q_{\alpha}(n)-\sqrt{\log\frac{2e}{\Lambda_n}}\right)_+^2 - \log K_n.
\end{align*}
Note that $\Gamma_{1,n}\leq \Gamma_{2,n}$ for all large $n$ and $\alpha_n \to 0$. It is sufficient to show that $\Gamma_{1,n}\to \infty$ as $n \to \infty$. 
In the first case: $\liminf\Lambda_n >0$ implies that both $\sqrt{\log (2e/\Lambda_n)}$ and $K_n$ are finite, and thus the condition 
\(
    \sqrt{n}\Omega_n - \max_{m\leq n}q_{\alpha}(n)/{\sqrt{2}} \to \infty,
\)
ensures that $\Gamma_{1,n} \to \infty$.

For the second case: Note that $K_n \leq \Lambda_n^{-1}$ and $\sqrt{n}\Omega_n\geq c \sqrt{-\log \Lambda_n}$ for some $c>2$. Applying the basic inequality $\sqrt{x+y}-\sqrt{x}\geq y /(2\sqrt{x})$ for all $x,y\geq 0$, we then have 
\begin{align*}
    &\Gamma_{1,n} \geq \left(c\sqrt{-\log \Lambda_n} -\frac{\max_{m\leq n}q_{\alpha}(n)}{\sqrt{2}}-\sqrt{\log\frac{2e}{\Lambda_n}}\right)_+^2 - \left(\sqrt{-\log \Lambda_n}\right)^2\\
    \geq\, & \left((c-1)\sqrt{-\log \Lambda_n} -\frac{\max_{m\leq n}q_{\alpha}(n)}{\sqrt{2}}-\frac{\log 2e}{2\sqrt{-\log \Lambda_n}} \right)_+^2 - \left(\sqrt{-\log \Lambda_n}\right)^2 \\
    \geq \, &\left((c-2)\sqrt{-\log \Lambda_n} -\frac{\max_{m\leq n}q_{\alpha}(n)}{\sqrt{2}}-\frac{\log 2e}{2\sqrt{-\log \Lambda_n}} \right)_+^2
\end{align*}
which goes to infinity. The last inequality above comes from the basic inequality $(x+y)^2-x^2\geq y^2$ for all $x,y\geq 0$.
\end{proof}

\begin{proof}[Proof of Theorem \ref{t: Hausdorff}]
Note that
\begin{multline}
    \label{ieq: Hausdorff}
    \mathbb{P}\left(d_H(T_n,\widehat{T}_n) \geq \varepsilon_n \right)
    \leq \mathbb{P}\left(\max_{k}\min_{j}\abs{\tau_k -\hat{\tau}_j} \geq \varepsilon_n\right)  \\
   + \mathbb{P}\left(\max_{k}\min_{j}\abs{\tau_k -\hat{\tau}_j} < \varepsilon_n, \max_{j}\min_{k}\abs{\tau_k -\hat{\tau}_j}\geq\varepsilon_n\right).
\end{multline}
For the first probability, if $\max_k\min_j \abs{\tau_k -\hat{\tau}_j}\geq \varepsilon_n$, then there exists a $k\in\{1,\dots,K_n\}$ such that on $I_k = (\tau_k-\varepsilon_n,\tau_k+\varepsilon_n)$ no change point is detected. Since $2\varepsilon_n \leq \Lambda_n$, we have the interval $I_k\subseteq (\tau_{k}-\Lambda_n/2,\tau_k+\Lambda_n/2)$. Applying the same technique as in the proof of Theorem~\ref{t: under} and replacing $\Lambda_n$ by $\varepsilon_n$, we obtain the first upper bound.

For the second probability, we decompose $\widehat{T}_n$ into two disjoint sets $ A_n = \{\hat{\tau}\in \widehat{T}_n: \min_k \abs{\hat{\tau}-\tau_k}<\varepsilon_n\}$ and $B_n = \{\hat{\tau}\in \widehat{T}_n: \min_k\abs{\hat{\tau}-\tau_k}\geq\varepsilon_n\}.$ Since $\max_k\min_j \abs{\tau_k -\hat{\tau}_j}< \varepsilon_n$, we have for any true change point $\tau_k$, there exists an estimated change point $\hat{\tau}_{j(k)}$ such that $\abs{\tau_k-\hat{\tau}_j}< \varepsilon_n$, that is, $\hat{\tau}_{j(k)}\in I_k$ for all $k = 1, \dots,K_n$. Note that $I_k$'s are pairwise disjoint. Thus $\hat{\tau}_{j(k)}$'s are pairwise distinct. By definition, all $\hat{\tau}_{j(k)}$'s belong to $A_n$ and therefore $\abs{A_n}\geq K$. Furthermore, the event $\max_{j}\min_{k}\abs{\tau_k -\hat{\tau}_j}\geq\varepsilon_n$ implies that there exists a $j^*$ such that $\abs{\tau_k -\hat{\tau}_{j^*}}\geq \varepsilon_n$ for all $k = 1,\dots,K_n$, i.e., $\hat{\tau}_{j^*}\in B_n$ and $\abs{B_n}\geq 1$. Since $\widehat{T}_n = A_n\sqcup B_n$, we have $\abs{\widehat{T}_n}\geq K_n+1$. Therefore, by Theorem~\ref{t: over}, the second probability is upper bounded by $\mathbb{P}(\hat{K}_n \geq K_n+1)\leq \alpha_n$. 
\end{proof}

\subsection{Proof of Theorem \ref{t: FDR}}
Let $\alpha \in (0,1)$. We first construct an upper bound for expectation of the number of false discoveries $\rm{FD}(\alpha)$ given there is no true discovery. 
\begin{lemma}\label[lemma]{l: FD 0} Under above notations we have
\begin{equation*}
    \mathbb{E}(\rm{FD}(\alpha) \mid TD(\alpha) = 0)\leq \frac{\alpha}{1-\alpha} \eqqcolon G(\alpha)
\end{equation*}
\end{lemma}

\begin{proof}
Note that it is sufficient to prove the claim for constant signal and we may assume that it equals to zero. We study the estimator $\widetilde{f}$ with $\widetilde{K}$ change points, which is defined in the proof of \cref{prop: over estimation of zero}. Again, since MUSCLE minimizes the number of estimated change points, the FD of MUSCLE is upper bounded by $\widetilde{K}$. Thus, applying Proposition~\ref{prop: over estimation of zero}, we have
\(
    \mathbb{E}({\rm FD(\alpha)}\mid{\rm TD}(\alpha) = 0) \leq \mathbb{E}( \widetilde{K}) \leq \frac{\alpha}{1-\alpha}\eqqcolon G(\alpha).
\)
\end{proof}

\begin{proof}[Proof of Theorem \ref{t: FDR}]
Applying Lemma \ref{l: FD 0}, \citet[Lemma~A.2]{li2016fdr} and the proof technique in \citet[Theorem~2.2]{li2016fdr}, we have
\begin{equation*}
    {\rm FDR} \leq \frac{G(\alpha)}{1+G(\alpha)}= \alpha.
\end{equation*}
Next we prove the statement on OER. It follows from \Cref{t: over} that $\mathbb{P}(\widehat{K}-K\geq k) \leq \alpha^k$. Thus,
\(
    \mathbb{E}((\widehat{K}-K)_{+}) = \sum_{k=1}^{\infty}\mathbb{P}(\widehat{K}-K \geq k) \leq {\alpha}/({1-\alpha}).
\)
If $K = 0$, then 
\begin{equation*}
    \mathbb{E}\left(\frac{(\widehat{K}-K)_+}{\widehat{K}\vee 1}\right) =
    \mathbb{E}\left(\frac{\widehat{K}}{\widehat{K}}\mathbbm{1}_{\{\widehat{K}\geq 1\}}\right)=\mathbb{P}(\widehat{K}\geq 1) \leq \alpha. 
\end{equation*}
For $K\geq 1$, 
\(
    \frac{(\widehat{K}-K)_+}{\widehat{K}\vee 1} =  \frac{(\widehat{K}-K)_+}{(\widehat{K}-K)_+ + K}.
\)
The function $x\mapsto x/(x+K)$ is strictly concave and monotone increasing on $[1,\infty)$. By Jensen's inequality 
\begin{equation*}
    \mathbb{E}\left(\frac{(\widehat{K}-K)_+}{\widehat{K}+1}\right)  \leq \frac{\mathbb{E}((\widehat{K}-K)_+)}{\mathbb{E}((\widehat{K}-K)_+) +K} 
    \leq \frac{\frac{\alpha}{1-\alpha}}{\frac{\alpha}{1-\alpha}+K}=\frac{\alpha}{(1-\alpha)K+\alpha}.\qedhere
\end{equation*}
\end{proof}

\section{Additional Materials for \Cref{S: 4}}
\label{A2}
Suppose that $Z_1,\dots,Z_n$ are observations from Model \ref{QSR} with underlying signal ${f}$. Let $\hat{f}$ denote an estimate of ${f}$. Then the \emph{mean integrated squared error} (MISE) of $\hat{f}$ is given by
\begin{equation}
    \label{e: MISE}
    {\rm MISE}\left(\hat{f}\right) = \frac{1}{n}\sum_{i=1}^n \left(\hat{f}(x_i)-{f}(x_i)\right)^2,
\end{equation}
and the \emph{mean integrated absolute error} (MIAE) of $\hat{f}$ by
\begin{equation}
    \label{e: MIAE}
    {\rm MIAE}\left(\hat{f}\right) = \frac{1}{n}\sum_{i=1}^n \left|\hat{f}(x_i) -{f}(x_i)\right|.
\end{equation}
For RNSP, MQS and MUSCLE we apply empirical quantiles to fit data on each segments and for other procedures we use sample means instead.

\Cref{f: blocks Cauchy,f: blocks Chi_sq,t: E1-E3} present additional results for \Cref{S: 5.1.1}, while \Cref{f: Well log RFPOP} corresponds to \Cref{ss:well-log}.

\begin{figure}[t]
        \centering
        \includegraphics[width=\linewidth]{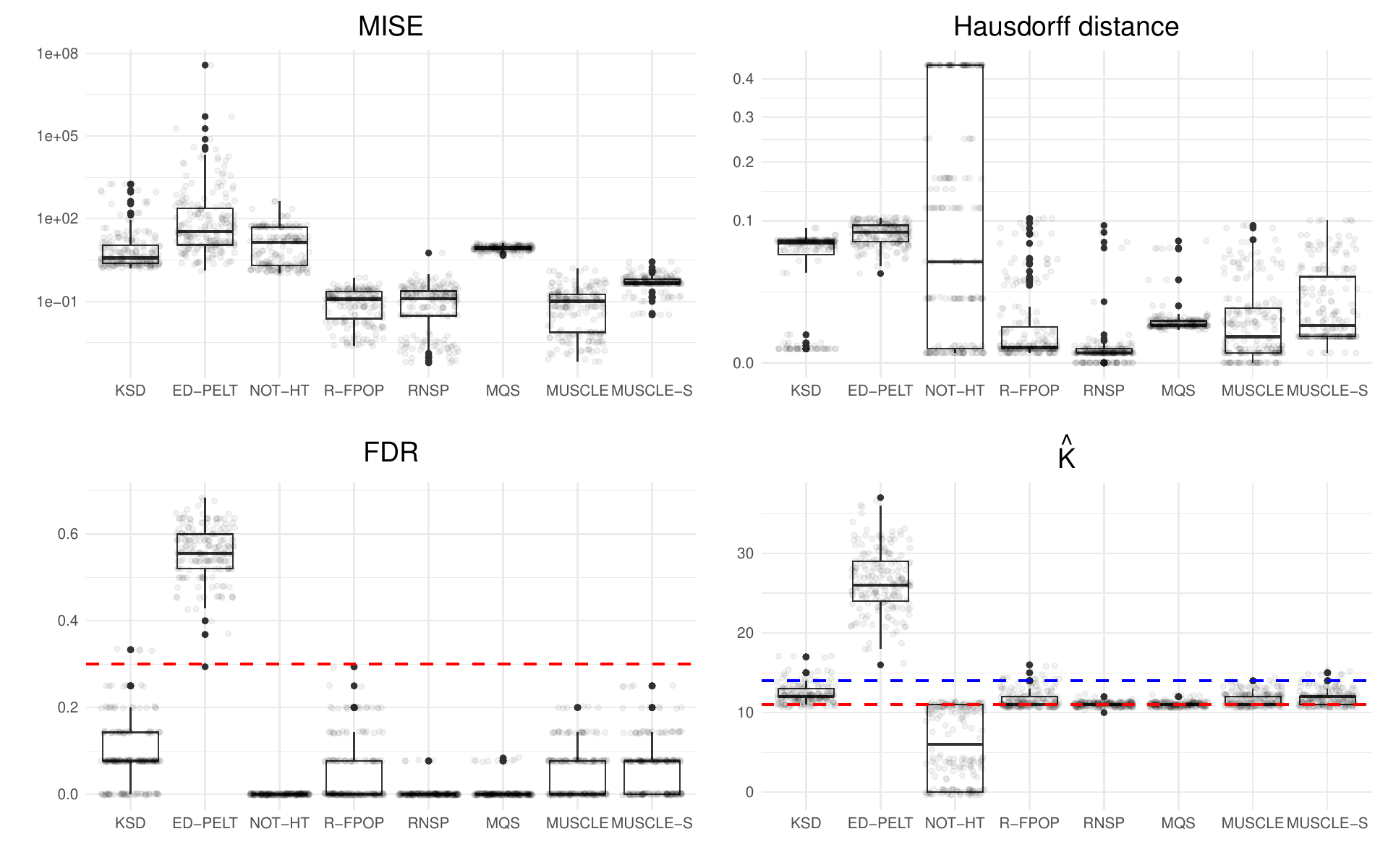}
        \caption{Recovery of the \texttt{blocks} signal in \ref{i:cauchy}. In each panel, the overall performance  over 200 repetitions is summarized as a boxplot and individual repetitions are jittered in dots with a low intensity. In the bottom left panel, the theoretical upper bound $\alpha = 0.3 $ on the FDR of MUSCLE is marked by a red dashed line. In the bottom right panel, the true numbers of change points in median and in distribution are $K = 11$ (marked by a red dashed line) and $K = 14$ (marked by a blue dashed line), respectively. }
        \label{f: blocks Cauchy}
\end{figure}

\begin{figure}[t]
        \centering
        \includegraphics[width=\linewidth]{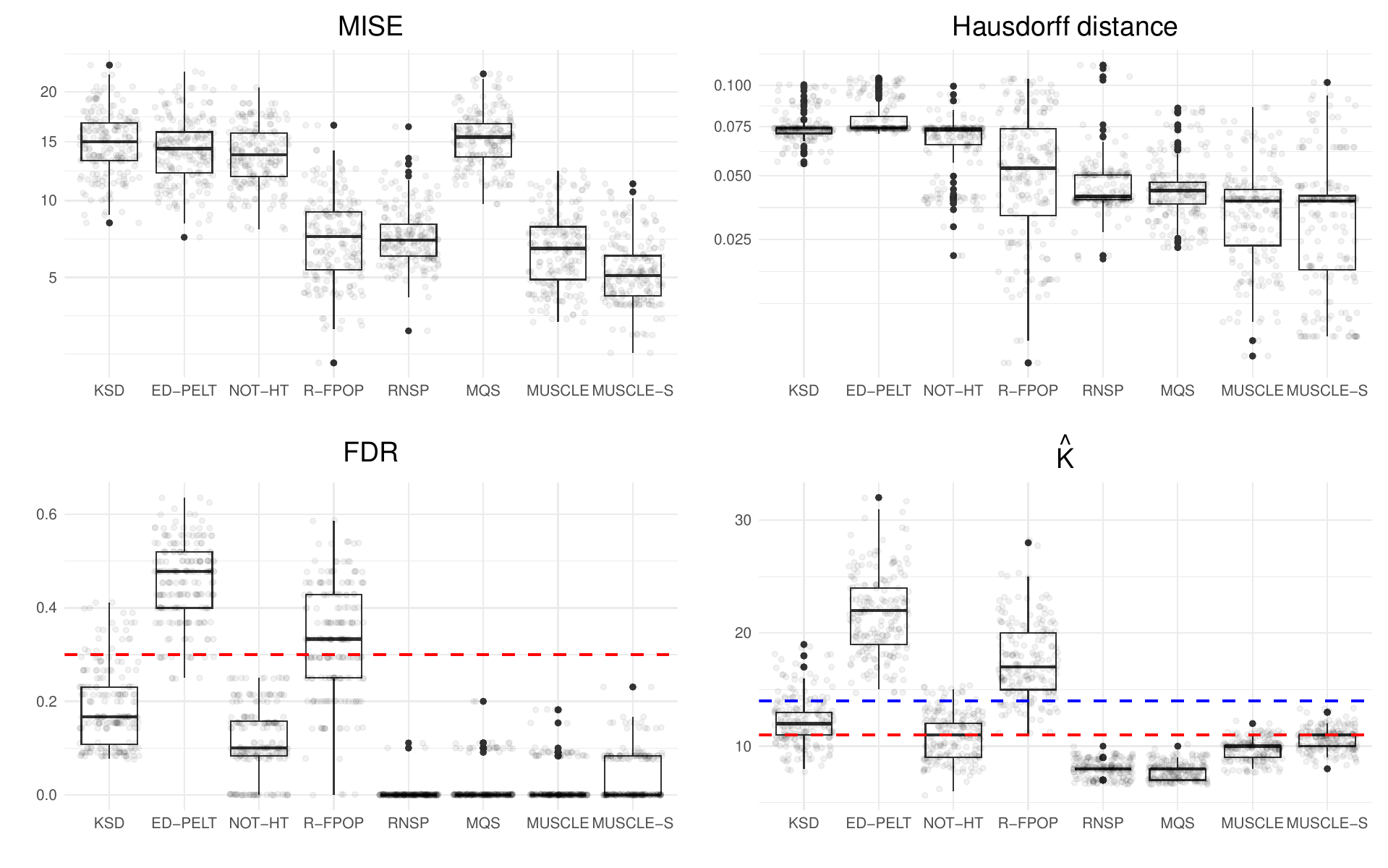}
        \caption{Recovery of the \texttt{blocks} signal in \ref{i:chi2}. In each panel, the overall performance  over 200 repetitions is summarized as a boxplot and individual repetitions are jittered in dots with a low intensity. In the bottom left panel, the theoretical upper bound $\alpha = 0.3 $ on the FDR of MUSCLE is marked by a red dashed line. In the bottom right panel, the true numbers of change points in median and in distribution are $K = 11$ (marked by a red dashed line) and $K = 14$ (marked by a blue dashed line), respectively. }
        \label{f: blocks Chi_sq}
\end{figure}

\begin{table*}[!ht]
\centering
\caption{Statistical performance of procedures in Section \ref{S: 5.1.1} in Gaussian \ref{i:gauss}, scaled $t_3$ \ref{i:tdist}, scaled Cauchy \ref{i:cauchy}, scaled and centered $\chi_3^2$-distribution \ref{i:chi2} and mixed distribution \ref{i:hetemix}. Medians of each criterion and their mean absolute deviations (in brackets) are presented, with the best performances highlighted in bold.}
\tabcolsep=0pt
\begin{tabular*}{1\textwidth}{@{\extracolsep{\fill}}cccccccccc@{\extracolsep{\fill}}}
\toprule%
Scenario & Method & MISE & MIAE & $L(T;\widehat{T})$ & $d_{H}(T,\widehat{T})$ & $\widehat{K}$ &FDR & V-measure & Time [s] \\
\midrule
\multirow{14}{*}{\ref{i:gauss}}  & \multirow{2}{*}{MUSCLE} & $0.010$   & $0.074$   & $\mathbf{0}$        & $\mathbf{0}$       & $\mathbf{2}$        & $\mathbf{0}$        &                                 $\mathbf{1}$       &  $15.871$ \\
                                &                         & ($0.006$) & ($0.001$) & ($0.0002$) & ($0.0350$)& ($0.29$) & ($0.07$) & ($0.036$)&($3.344$)\\
                                \cline{2-10}%
                                & \multirow{2}{*}{MUSCLE-S} & $0.007$  & $0.051$  & $\mathbf{0}$        & $\mathbf{0}$       & $\mathbf{2}$        & $\mathbf{0}$        & $\mathbf{1}$       & $0.697$\\
                                &                         & ($0.005$)  & ($0.024$)& ($0.0002$) & ($0.0442$)& ($0.38$) & ($0.09$) & ($0.047$)& ($0.309$)\\
                                \cline{2-10}%
                                & \multirow{2}{*}{{MQS}} & $0.143$  & $0.070$  & $0.006$        & $0.0065$       & $\mathbf{2}$        & $\mathbf{0}$        & $0.953$       & $2.879$\\
                                &                         & ($0.045$)  & ($0.021$)& ($0.0007$) & ($0.0007$)& ($0$) & ($0$) & ($0.014$)& ($0.082$)\\
                                \cline{2-10}%
                                & \multirow{2}{*}{SMUCE}  & $\mathbf{0.001}$   & $\mathbf{0.025}$   & $\mathbf{0}$        & $\mathbf{0}$       & $\mathbf{2}$        & $\mathbf{0}$        & $\mathbf{1}$       &          $0.026$\\
                                &                         & $(0.001)$ & $(0.010)$ & $(4.8\cdot 10^{-5})$  & $(0.0159)$& $(0.11)$ & $(0.03)$ & $(0.017)$& $(0.002)$ \\
                                \cline{2-10}%
                                & \multirow{2}{*}{FDRSeg} & $\mathbf{0.001}$  & $0.027$  & $\mathbf{0}$         & $\mathbf{0}$       & $\mathbf{2}$        & $\mathbf{0}$        & $\mathbf{1}$       & $1.068$\\
                                &                         & $(0.002)$& $(0.010)$& $(5.6\cdot 10^{-5})$  & $(0.0542)$& $(0.71)$ & $(0.15)$ & $(0.065)$& $(1.689)$\\
                                \cline{2-10}%
                                & \multirow{2}{*}{WBS}  & $0.017$  & $0.032$  & $5\cdot 10^{-4}$ & $0.0002$ & $3$     & $0.25$     & $0.988$ & $0.076$\\
                                &                       & $(0.001)$& $(0.009)$& $(7.4\cdot 10^{-4})$    & $(0.0058)$& $(0.41)$ & $(0.10)$ & $(0.006)$& $(0.004)$\\
                                \cline{2-10}%
                                & \multirow{2}{*}{PELT} & $0.017$  & $0.029$  & $5\cdot 10^{-4}$ & $5\cdot 10^{-4}$ & $\mathbf{2}$ & $\mathbf{0}$   & $0.999$ & $\mathbf{7\cdot 10^{-4}}$\\
                                &                       & $(0.001)$& $(0.009)$& $(1.9\cdot10^{-4})$& $(1.9\cdot10^{-4})$& $(0)$&$(0)$ & $(0.001)$ & $(9.9\cdot 10^{-5})$\\
\midrule
\multirow{12}{*}{\ref{i:tdist}} &\multirow{2}{*}{MUSCLE} & $\mathbf{0.315}$ & $0.188$  & $0.0009$ & $0.0014$ & $\mathbf{11}$& $\mathbf{0}$ & $\mathbf{0.996}$ & $6.278$\\
                                &                        & $(0.225)$& $(0.042)$& $(0.0010)$& $(0.0100)$ & $(0.46)$& $(0.03)$& $(0.005)$ & $(0.583)$\\
                                \cline{2-10}%
                                &\multirow{2}{*}{MUSCLE-S} & $0.665$ & $0.232$ & $0.0034$ & $0.0034$ & $\mathbf{11}$ & $\mathbf{0}$ & $0.990$ & $0.715$\\
                                &                          & $(0.221)$& $(0.043)$ & $(0.0004)$ & $(0.0197)$ & $(0.67)$ & $(0.04)$& $(0.010)$ & $(0.142)$\\
                                \cline{2-10}%
                                &\multirow{2}{*}{MQS} & $9.611$ & $0.900$ & $0.0112$ & $0.0117$ & $\mathbf{11}$ & $\mathbf{0}$& $0.923$ & $0.874$\\
                                &                      & $(1.277)$& $(0.105)$ & $(0.0021)$ & $(0.0077)$ & $(0.20)$ & $(0.04)$ & $(0.008)$& $(0.036)$\\
                                \cline{2-10}%
                                &\multirow{2}{*}{NOT-HT} & $1.219$ & $0.306$ & $0.0009$ & $0.0009$ & $\mathbf{11}$ & $\mathbf{0}$ & $0.985$ & $0.180$\\
                                &                        & $(0.148)$& $(0.048)$& $(0.0010)$& $(0.0080)$& $(0.16)$ & $(0.01)$& $(0.003)$& $(0.009)$\\
                                \cline{2-10}%
                                &\multirow{2}{*}{KSD} & $2.167$ & $0.438$ & $0.0029$& $0.0732$ & $13$ & $0.14$ & $0.942$ & $70.411$\\
                                &                      & $(0.469)$& $(0.072)$& $(0.0031)$ & $(0.0115)$ & $(0.97)$& $(0.05)$ & $(0.015)$& $(5.274)$\\
                                \cline{2-10}%
                                &\multirow{2}{*}{R-FPOP} & $1.632$& $0.495$& $0.0009$ & $0.0488$& $27$ & $0.57$ & $0.902$ & $\mathbf{0.002}$\\
                                &                        & $(0.330)$& $(0.066)$& $(0.0007)$ & $(0.0168)$& $(3.44)$& $(0.06)$& $(0.020)$& $(2 \cdot 10^{-4})$\\
                                \cline{2-10}%
                                &\multirow{2}{*}{ED-PELT} & $1.838$& $0.406$& $0.0009$ & $0.0747$& $20$ & $0.43$ & $0.905$ & $0.118$\\
                                &                        & $(0.544)$& $(0.066)$& $(0.0007)$ & $(0.0037)$& $(2.20)$& $(0.06)$& $(0.014)$& $(0.012)$\\
                                \cline{2-10}%
                                &\multirow{2}{*}{RNSP} & $0.348$& $\mathbf{0.169}$& $\mathbf{0.0004}$ & $\mathbf{0.0004}$& $\mathbf{11}$ & $\mathbf{0}$ & $\mathbf{0.996}$ & $32.587$\\
                                &                        & $(0.392)$& $(0.044)$& $(0.0034)$ & $(0.0083)$& $(0)$& $(0.02)$& $(0.004)$& $(0.708)$\\
\midrule
\multirow{12}{*}{\ref{i:cauchy}} &\multirow{2}{*}{MUSCLE} & $\mathbf{0.101}$& $0.049$& $\mathbf{0.0004}$& $0.0034$ &$\mathbf{11}$ & $\mathbf{0}$& $0.995$& $6.164$\\
                                & & $(0.158)$& $(0.015)$ & $(6.9\cdot 10^{-4})$& $(0.0162)$& $(0.07)$ & $(0.05)$& $(0.009)$& $(1.028)$\\
                                \cline{2-10}%
                                &\multirow{2}{*}{MUSCLE-S} & $0.463$& $0.090$ & $0.0034$& $0.0068$ & $12$ & $0.076$ & $0.987$ & $0.790$\\
                                & & $(0.204)$& $(0.017)$& $(5.7\cdot 10^{-4})$& $(0.0217)$& $(0.77)$& $(0.05)$& $(0.012)$& $(0.206)$\\
                                \cline{2-10}%
                                &\multirow{2}{*}{MQS} & $8.838$& $0.653$& $0.0068$& $0.0068$& $\mathbf{11}$ & $\mathbf{0}$ & $0.930$& $0.849$\\
                                & & $(1.291)$ & $(0.087)$ & $(0.0011)$ & $(0.0042)$& $(0.07)$& $(0.01)$ & $(0.008)$ & $(0.041)$\\
                                \cline{2-10}%
                                &\multirow{2}{*}{NOT-HT} & $14.023$& $2.153$ & $0.0507$& $0.0507$ & $6$ & $\mathbf{0}$ & $0.913$ & $0.173$\\
                                & & $(28.151)$& $(2.325)$ & $(0.1586)$ & $(0.1588)$ & $(4.39)$ & $(0)$ & $(0.366)$ & $(0.008)$\\
                                \cline{2-10}%
                                &\multirow{2}{*}{KSD} & $3.751$ & $1.023$ & $0.0009$ & $0.0717$ & $12$ & $0.08$ & $0.956$ & $69.789$\\
                                & & $(121.466)$ & $(1.421)$& $(7.1\cdot 10^{-4})$ & $(0.0252)$ & $(0.88)$& $(0.06)$ & $(0.014)$ & $(4.310)$ \\
                                \cline{2-10}%
                                &\multirow{2}{*}{R-FPOP} & $0.122$& $0.0565$ & $0.0009$ & $0.0012$ & $\mathbf{11}$ & $\mathbf{0}$ & $0.9839$ & $\mathbf{0.002}$\\
                                & & $(0.094)$& $(0.012)$ & $(2.6\cdot 10^{-4})$ & $(0.0170)$ & $(0.88)$ & $(0.06)$ & $(0.008)$ & $(1.3\cdot 10^{-4})$\\
                                \cline{2-10}%
                                &\multirow{2}{*}{ED-PELT} & $34.029$& $1.289$& $\mathbf{0.0004}$ & $0.0849$& $26$ & $0.55$ & $0.861$ & $0.094$\\
                                &                        & $(3.7\cdot10^5)$& $(3.967)$& $(3.1\cdot 10^{-4})$ & $(0.0114)$& $(2.91)$& $(0.05)$& $(0.019)$& $(0.014)$\\
                                \cline{2-10}%
                                &\multirow{2}{*}{RNSP} & $0.124$& $\mathbf{0.041}$& $\mathbf{0.0004}$ & $\mathbf{0.0004}$& $\mathbf{11}$ & $\mathbf{0}$ & $\mathbf{0.997}$ & $29.164$\\
                                &                        & $(0.164)$& $(0.014)$& $(4.9\cdot 10^{-4})$ & $(0.0032)$& $(0.04)$& $(3\cdot 10^{-3})$& $(0.002)$& $(0.477)$\\
\bottomrule
\end{tabular*}
\label{t: E1-E3}
\end{table*}

\begin{table*}[!ht]
\centering
\caption{Continuation of \Cref{t: E1-E3}.}
\tabcolsep=0pt
\begin{tabular*}{1\textwidth}{@{\extracolsep{\fill}}cccccccccc@{\extracolsep{\fill}}}
\toprule%
Scenario & Method & MISE & MIAE & $L(T;\widehat{T})$ & $d_{H}(T,\widehat{T})$ & $\widehat{K}$ &FDR & V-measure & Time [s] \\
\midrule
\multirow{12}{*}{\ref{i:chi2}} &\multirow{2}{*}{MUSCLE} & $6.694$ & $1.364$ & $0.0378$ & $\mathbf{0.0390}$ & $10$ & $\mathbf{0}$ & $0.942$ & $11.812$\\
                                &                       & $(1.849)$& $(0.222)$& $(0.0111)$& $(0.0131)$ & $(0.70)$ & $(0.03)$ & $(0.017)$ & $(3.139)$\\
                                \cline{2-10}%
                                &\multirow{2}{*}{MUSCLE-S} & $\mathbf{5.111}$ & $\mathbf{1.212}$& $0.0302$& $\mathbf{0.0390}$& $\mathbf{11}$& $\mathbf{0}$& $\mathbf{0.957}$& $1.121$\\
                                &                        & $(1.453)$& $(0.219)$& $(0.0145)$& $(0.0178)$& $(0.79)$& $(0.04)$& $(0.017)$& $(0.218)$\\
                                \cline{2-10}%
                                &\multirow{2}{*}{MQS} & $15.431$ & $1.830$& $0.0434$& $0.0434$& $8$& $\mathbf{0}$& $0.879$& $1.051$\\
                                &                      & $(1.942)$& $(0.215)$& $(0.0088)$& $(0.0094)$& $(0.62)$& $(0.03)$& $(0.013)$& $(0.042)$\\
                                \cline{2-10}%
                                &\multirow{2}{*}{NOT-HT} & $13.817$& $2.757$& $0.0398$& $0.0732$& $\mathbf{11}$ & $0.10$ & $0.919$& $0.185$\\
                                &                        & $(2.142)$& $(0.232)$& $(0.0136)$& $(0.0110)$& $(1.71)$& $(0.06)$& $(0.014)$& $(0.008)$\\
                                \cline{2-10}%
                                &\multirow{2}{*}{KSD} & $15.010$& $2.800$& $0.0306$& $0.0742$ & $12$& $0.17$& $0.896$& $76.534$\\
                                &                        & $(2.335)$& $(0.227)$& $(0.0145)$& $(0.0051)$& $(1.71)$& $(0.07)$& $(0.020)$ & $(5.460)$\\
                                \cline{2-10}%
                                &\multirow{2}{*}{R-FPOP} & $7.463$& $1.389$& $0.0048$& $0.0534$& $17$& $0.33$& $0.924$ & $\mathbf{0.002}$\\
                                &                        & $(2.246)$& $(0.233)$& $(0.0037)$& $(0.0226)$& $(2.65)$& $(0.09)$& $(0.024)$& $(1.5\cdot 10^{-4})$\\
                                \cline{2-10}%
                                &\multirow{2}{*}{ED-PELT} & $14.350$& $2.673$& $\mathbf{0.0026}$ & $0.0742$& $22$ & $0.47$ & $0.887$ & $0.106$\\
                                &                        & $(2.232)$& $(0.212)$& $(0.0019)$ & $(0.0082)$& $(2.83)$& $(0.07)$& $(0.018)$& $(0.009)$\\
                                \cline{2-10}%
                                &\multirow{2}{*}{RNSP} & $7.205$& $1.313$& $0.0410$ & $0.0410$& $8$ & $\mathbf{0}$ & $0.936$ & $31.684$\\
                                &                        & $(1.388)$& $(0.194)$& $(0.0087)$ & $(0.0089)$& $(0.48)$& $(3 \cdot 10^{-3})$ & $(0.013)$& $(1.205)$ \\
\midrule
\multirow{12}{*}{\ref{i:hetemix}} &\multirow{2}{*}{MUSCLE} & $\mathbf{1.064}$ & $0.461$ & $\mathbf{0.0019}$ & $\mathbf{0.0043}$ & $\mathbf{11}$ & $\mathbf{0}$ & $\mathbf{0.987}$ & $6.967$\\
                                & & $(0.531)$ & $(0.087)$ & $(0.0021)$ & $(0.0159)$ & $(0.57)$ & $(0.04)$& $(0.009)$ & $(0.6441)$\\
                                \cline{2-10}%
                                &\multirow{2}{*}{MUSCLE-S} & $1.536$& $0.457$ & $0.0034$ & $0.0107$ & $12$ & $0.076$ & $0.982$ & $0.790$\\
                                & & $(0.467)$ & $(0.091)$ & $(6.3\cdot 10^{-4})$ & $(0.0187)$ & $(0.67)$& $(0.05)$ & $(0.011)$ & $(0.151)$\\
                                \cline{2-10}%
                                &\multirow{2}{*}{MQS} & $0.520$& $1.102$& $0.0132$ & $0.0156$ & $\mathbf{11}$ & $0.042$& $0.917$ & $0.790$\\
                                & & $(1.446)$ & $(0.122)$ & $(0.0018)$ & $(0.0098)$ & $(0.59)$ & $(0.04)$& $(0.009)$& $(0.034)$\\
                                \cline{2-10}%
                                &\multirow{2}{*}{NOT-HT} & $15.060$ & $2.383$& $0.1191$& $0.1191$ & $4$ & $\mathbf{0}$ & $0.837$ & $0.170$\\
                                & & $(21.120)$ &$(1.964)$ & $(0.1443)$& $(0.1381)$& $(4.03)$ & $(0.02)$ & $(0.319)$ & $(0.007)$\\
                                \cline{2-10}%
                                &\multirow{2}{*}{KSD} & $4.549$& $1.260$ & $0.0051$& $0.0732$& $13$ & $0.143$ & $0.937$ & $68.366$ \\
                                & & $(82.669)$ & $(0.819)$ & $(0.0080)$ & $(0.0068)$ & $(0.92)$ & $(0.04)$& $(0.011)$ & $(4.208)$\\
                                \cline{2-10}%
                                &\multirow{2}{*}{R-FPOP} & $3.728$ & $0.825$ & $0.0020$ & $0.0454$ & $25$ & $0.538$ & $0.919 $ & $\mathbf{0.002} $\\
                                & & $(0.637)$ & $(0.090)$ & $(0.0013)$ & $(0.0059)$& $(2.46)$ & $(0.05)$ & $(0.010)$& $(7.3\cdot 10^{-5})$\\
                                \cline{2-10}%
                                &\multirow{2}{*}{ED-PELT} & $22.273$& $1.424$& $0.0014$ & $0.0793$& $22$ & $0.478$ & $0.884$ & $0.106$\\
                                &                        & $(2.1\cdot10^4)$& $(4.895)$& $(8.5\cdot10^{-4})$ & $(0.0070)$& $(2.63)$& $(0.06)$& $(0.014)$& $(0.009)$\\
                                \cline{2-10}%
                                &\multirow{2}{*}{RNSP} & $1.417$& $\mathbf{0.359}$& $0.0124$ & $0.0136$& $\mathbf{11}$ & $\mathbf{0}$ & $0.986$ & $29.996$\\
                                &                        & $(1.121)$& $(0.129)$& $(0.0091)$ & $(0.0115)$& $(0.50)$& $(0.02)$& $(0.007)$& $(0.686)$\\
\bottomrule
\end{tabular*}
\end{table*}

\begin{figure}[t]
        \centering
        \includegraphics[width =\linewidth]{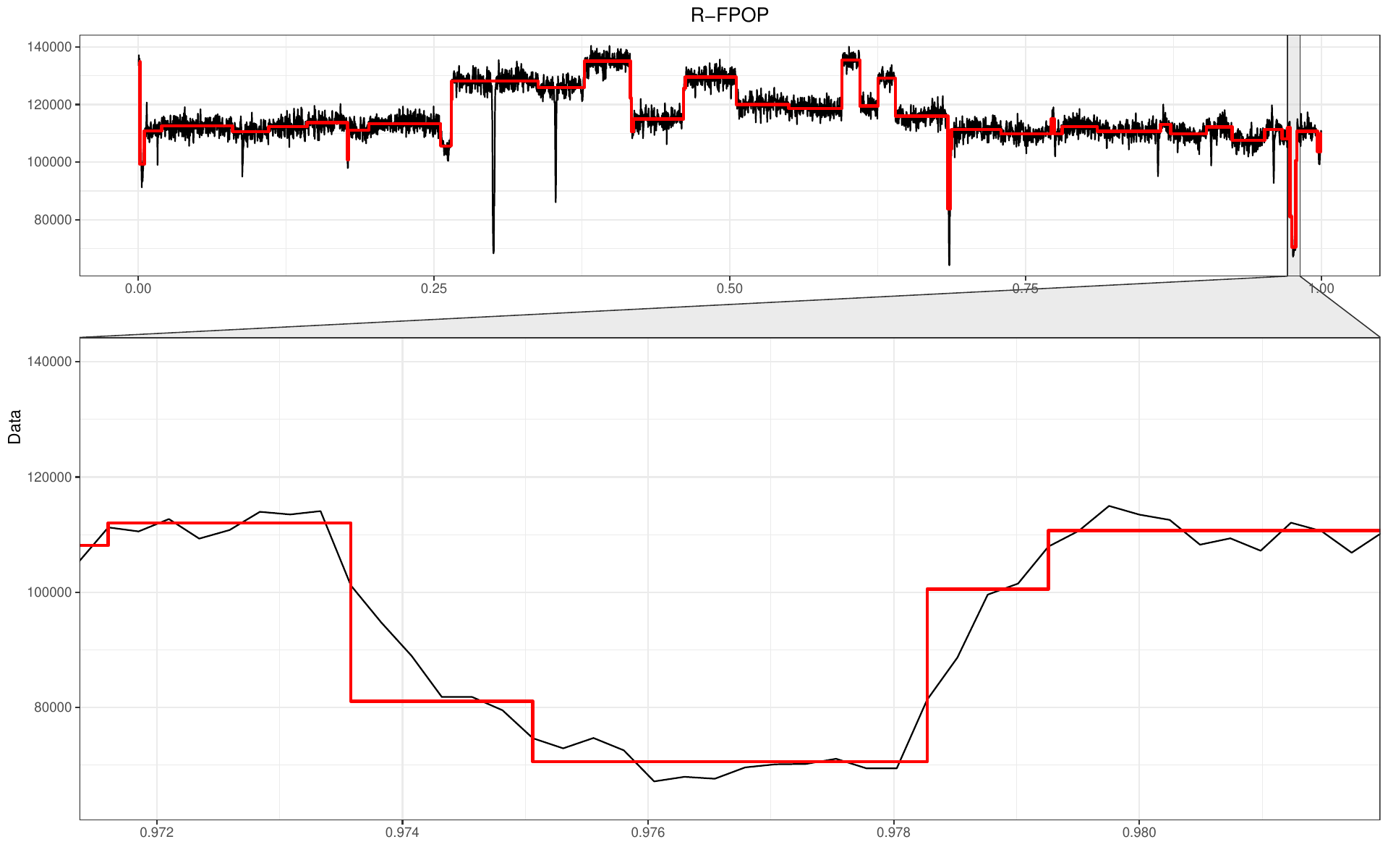}
        \caption{Segmentation of well-log data by R-FPOP. The zoom-in part (around the last flickering event) suggests that R-FPOP overfits the data.}
        \label{f: Well log RFPOP}
\end{figure}

\section{Implementation Details}
\label{A3}

The pseudocode of the proposed MUSCLE algorithm is provided in Algorithm~\ref{alg: MUSCLE}.

\begin{algorithm*}
\caption{MUSCLE}\label{alg: MUSCLE}
\begin{algorithmic}[1]
\Require Observation vector $Z = (Z_1,\dots,Z_n)$, quantile vector $q_{\alpha} = (q_{\alpha}(1),\dots,q_{\alpha}(n))$ and $\beta \in (0,1)$
\Ensure MUSCLE estimator $\hat{f}$
\State Construct a \emph{wavelet tree} of $Z$
\State $j_{\min} \gets 1$, $J \gets 0$, $C_{1,n} = \infty$, $\hat{f}_{1,1} \gets Z_1$, $C_{1,1} = 0$ and $J_1 \gets 0$
\While{ $J<n$ and $C_{1,n}=\infty$}
    \State $i_{\min} \gets n$ and  $m \gets n$
    \For{$j = j_{\min},\dots,n$}
        \If{$J_j = J$}
            \State $i_{\min} \gets \min \{i_{\min},j\}$
        \Else
            \State $S \gets 0$
            \For{$i = j-1,\dots, i_{\min}$}
                \If{$J_i = J$}
                    \State $L = -\infty$ and $U = \infty$ 
                    \For{$l = j-i,\dots, 1$} \Comment{For dyadic intervals, $l \in\{ 2^p,2^{p-1},\dots,1\}$, $p = \floor{{\log_2(j-i)}}$}
                        \For{$ k = i,\dots, j-i$}
                            \State Compute $Z_{k,u_{j-i+1,l}}^{k+l-1}$ and $Z_{k,v_{j-i+1,l}}^{k+l-1}$
                            \State $L = \max \left\{L, Z_{k,u_{j-i+1,l}}^{k+l-1}\right\}$ and $U = \min \left\{U, Z_{k,v_{j-i+1,l}}^{k+l-1}\right\}$ 
                        \EndFor
                    \EndFor
                    \If{$L\leq U$}
                         \State Compute $\hat{\theta}_{i,j}$ 
                         \State Compute $C_{i,j}$, the cost of fitting $Z_i,\dots,Z_j$ with $\hat{\theta}_{i,j}$
                         \State $m = \min\{m,j\}$
                         \If{$C_{1,i-1} + C_{i,j}<C_{1,j}$}
                            \State $\hat{f}_{1,j} \gets \hat{f}_{1,i-1}\mathbbm{1}_{\{1,\dots,i-1\}} + \hat{\theta}_{i,j}\mathbbm{1}_{\{i,\dots,j\}}$
                            \State $C_{1,j} \gets C_{1,i-1} + C_{i,j}$
                            \State $J_j \gets J + 1$
                         \EndIf
                    \Else
                        \If{$T_{[i/n,j/n]}^0(Z,\theta)> \max_{s\leq n-i_{\min}}q_{\alpha}(s)$ for all $\theta\in \mathbb{R}$} 
                            \State $S \gets 1$ \Comment{Pruning 1: drop all indices $j'>j$, as no $j'>j$ with $\mathcal{I}([i/n,j'/n])= 1$}
                        \EndIf
                        \If{$T_{[i/n,j/n]}^0(Z,\theta)> \max_{s\leq j-i_{\min}}q_{\alpha}(s)$ for all $\theta\in \mathbb{R}$}
                            \State \textbf{break} \Comment{Pruning 2: drop all indices $i'<i$, as no $ i'<i$ with $\mathcal{I}([i'/n,j/n])= 1$}
                        \EndIf
                    \EndIf
                \EndIf
            \EndFor
        \EndIf
        \If{$S=1$}
            \State \textbf{break}
      \EndIf
    \EndFor
    \State $J \gets J + 1$ and  $j_{\min} \gets m$
\EndWhile
\State \Return $\hat{f}_{1,n}$
\end{algorithmic}
\end{algorithm*}

\subsection{Computation of local tests}
We give an efficient way to the multiscale test statistic on $I\subseteq [0,1)$. Let $\theta\in \mathbb{R}$. The involved test statistic is 
\begin{equation*}
    T_{I}(Z,\theta) = \underset{J \subseteq I}{\max}\sqrt{2L_J(Z,\theta)} - \sqrt{2\log\frac{e\abs{I} }{\abs{J}}},
\end{equation*}
with $L_J(Z,\theta) = \abs{J}\bigl[\overline{W}_J \log \frac{\overline{W}_J}{\beta}+(1-\overline{W}_J)\log \frac{1-\overline{W}_J}{1-\beta}\bigr]$, $ \overline{W}_J = \frac{1}{\abs{J}}\sum_{i \in J}W_i$ and $W_i = W_{i}(Z_i,\theta) = \mathbbm{1}_{\{Z_i\leq \theta\}}.$
The multiscale constraint $T_I(Z,\theta)\leq q_{\alpha}(\abs{I})$ is equivalent to for all intervals $J\subseteq I$,
\begin{equation}\label{e: tilde q}
    \overline{W}_J \log \frac{\overline{W}_J}{\beta}+(1-\overline{W}_J)\log \frac{1-\overline{W}_J}{1-\beta} \leq \frac{1}{2\abs{J}}\,\left(q_{\alpha}(\abs{I})+\sqrt{2 \log \frac{e\abs{I}}{\abs{J}}}\right)^2\;\eqqcolon\; \tilde{q}_{\abs{I},\abs{J}}.
\end{equation}
Here $\tilde{q}_{\abs{I},\abs{J}} = \tilde{q}_{k,l}$ depends only on $k = \abs{I}$ and $l = \abs{J}$. The function $g_{\beta}:(0,1)\to \mathbb{R}$, 
\(
    x\mapsto x\log\frac{x}{\beta}+(1-x)\log\frac{1-x}{1-\beta}
\)
is strictly convex on $(0,1)$ and has a unique minimizer at $x=\beta$. Let $u_{k,l}$ and $v_{k,l}$ with $u_{k,l} \le v_{k,l}$ be the solutions of $g_{\beta}(x) = \tilde{q}_{k,l}$. Then the inequality in \eqref{e: tilde q} becomes
\(
    \overline{W}_J \in [u_{k,l},v_{k,l}], \text{ for all intervals }  J\subseteq I, \,k = \abs{I}, l = \abs{J}.
\)
By the definition of $W_i$'s, the above displayed relation holds if and only if 
\(
    \theta \in \bigcap_{J\subseteq I}\left[Z_{J,u_{k,l}},Z_{J,v_{k,l}}\right]
\)
where $Z_{J,q}$ is the $q$-quantile of $Z_J = (Z_i)_{i\in J}$ for any $q\in (0,1)$. This intersection  is empty if and only if there is no $\theta\in \mathbb{R}$ fulfils the multiscale constraints on $I$, in which case at least one additional change point is required. The subinterval $J$ is varying in $I$, and the corresponding quantities $u_{k,l}$ and $v_{k,l}$ are changing as well. It amounts to computing the so-called \emph{range quantiles}, which is formulated as the \emph{range quantile problem} below.

\begin{enumerate}[align=left, leftmargin=*,labelsep=.1cm]
    \item[{Input}:] 
    A vector $Y = (Y_1,\dots, Y_n) \in \mathbb{R}^n$ of possibly unsorted elements and a sequence of $m$ queries $Q_1,\dots,Q_m$, where $Q_i = (L_i, R_i, q_i)$ with $1\leq  L_i \leq R_i \leq n$ and $q_i \in [0,1]$ for $i=1,\dots,m$.
    \item[{Output}:]A sequence of $x_1,\dots,x_m$, where $x_i$ is the $q_i$-quantile of the elements in $Y[L_i,R_i] = (Y_{L_i},\dots,Y_{R_i})$.
\end{enumerate}

In computer science, the range quantile problem is solved in two steps: a preprocessing step and a query step. The preprocessing step transforms the original input $Y$ into a certain data structure, and in the query step each query is answered efficiently based on the data structure created in the preprocessing step. We focus on the data structure of \emph{wavelet tree} and refer to \cref{r:dstr} and Table~\ref{t: complexity} for alternative data structures.

\begin{table}[!t]
\centering
\caption{Comparison of different solutions to range quantile problem: By Naive, DH, MST, PST, and WT, we denote the solutions provided by naive implementation, double-heaps, merge sort tree, persistent segment tree, and wavelet tree, respectively. The table includes their space requirements, computational complexities of the preprocessing step, each query, and the total complexity for $m$ queries, where $m \leq n^3$. All results are from \citep{castro2016wavelet}.}%
\begin{tabular*}{0.9\columnwidth}{@{\extracolsep\fill}l|l|l|l|l@{\extracolsep\fill}}
\toprule
Method   & Space & Preprocessing & Query  & Total  \\
\midrule
Naive       &  $O(n^3)$ & $O(n^3\log n)$ & $O(1)$ &  $ O(n^3\log n +m)$\\
DH      & $O(n^3)$ & $O(n^3\log n)$ & $O(\log n)$ & $O(n^3\log n + m\log n)$  \\
MST      & $O(n\log n)$ & $O(n\log n)$ & $O(\log^2 n)$ & $O(n\log n + m\log^2 n)$ \\
PST      & $O(n\log n)$ & $O(n\log n)$ & $O(\log n)$ & $O(n\log n + m\log n)$ \\
WT       & $O(n\log n)$ & $O(n\log n)$ & $O(\log n)$ & $O(n\log n + m\log n)$\\
\bottomrule
\end{tabular*}
\label{t: complexity}
\end{table}

A wavelet tree is a binary tree whose nodes consist of a \emph{value vector} and a \emph{rank vector}. The value vector of root is exactly the whole input and it is spitted into two children by comparing with the median of the value vector. More precisely, its left child contains all elements smaller than median while its right child consists of all elements larger or equal than median (for simplicity we assume that all $n$ data are distinct). For $i \in\{1, \dots, n\}$, the $i$-th entry of rank vector in root indicates the number of elements belonging to the left child between the first and $i$-th value. A wavelet tree can be constructed recursively until its leaves (i.e., nodes with only one element). Note that it is possible to compute all splitting required medians in $O(n\log n)$ steps and each layer can be constructed in $O(n)$ time. Thus, a wavelet tree is created in $O(n\log n)$ steps. Each single query $Q = (L,R,q)$ is answered in $O(\log n)$ steps as follows: Let $k = \ceil{(R-L+1)\cdot q}$ denote order of the $q$-th quantile of $(Y_L,\dots,Y_R)$. At each node we compute the difference between the rank of $R$ and rank of $L-1$, which equals to the number of elements in $(Y_L,\dots,Y_R)$ belonging to its left child. If it is greater or equal than $k$, then the target value is in the left child; otherwise, the target value is in the right child. We continue this process until the last layer and return value in the last leaf. The computation  cost of at each node is $O(1)$, since all ranks are precomputed. Note the depth of a wavelet tree is $O(\log n)$, so the computational complexity of each query is $O(\log n)$. 

This leads to the following proposition.
\begin{proposition}\label[proposition]{prop: I}
   For any subinterval $I\subseteq [0,1)$, the quantity $\mathcal{I}(I)$ can be computed in $O(\abs{I}^2\log n)$ steps provided that the \emph{wavelet tree} of $Z_1,\dots,Z_n$ is constructed.
\end{proposition}

\subsection{Proof of \Cref{th: complexity}}

We proceed in the following steps. 

\emph{Recursion.} Consider the recursive relation in \eqref{e: bellman}. The subproblem with $j=1$ is solved immediately (cf. line 1 in \cref{alg: MUSCLE}). For any $j\geq 2$, the $j$-th subproblem can be solved recursively by \eqref{e: bellman} via checking whether there is a constant function fulfilling the multiscale on $[i/n,j/n)$, i.e., $\mathcal{I}([i/n,j/n)) = 1$ (lines 12--27 in \cref{alg: MUSCLE}). One should keep in mind that even if we knew $\mathcal{I}\left(\left[i/n,j/n\right)\right) = 0$ for some $i$, we can not conclude $\mathcal{I}\left(\left[i'/n,j/n\right)\right)=0$ for $i'< i$, due to the varying penalties and quantiles. Since $\widehat{K}[j]$ requires the minimization among \emph{all} indices $i\leq j-1$ with $\mathcal{I}([i/n,j/n))= 1$, all indices should be taken into account. 

\emph{Pruning.}
Nevertheless, the checking of some indices can be avoided if we employ some pruning techniques similarly as in \citet{li2016fdr}. Towards this, we consider the test statistic $T_I^0(Z,\theta)= {\max}_{J \subseteq I}\sqrt{2L_J (W,\theta)} - \sqrt{2\log\left({en}/{\abs{J}}\right)}$, where  the largest penalty with size $n$ is subtracted. Then it holds for all $Z$, $\theta$ and $I$ that $T_I(Z,\theta) \geq T_I^0(Z,\theta)$. During the minimization over indices ranging from $i = j-1$ to $i = 0$, if we find an index $i$ with $\mathcal{I}\left(\left[i/n,j/n\right)\right) = 0$, then we check whether there exists a $\theta \in \mathbb{R}$ such that 
$
T_{[i/n,j/n]}^0(Z,\theta) \leq \max_{m\leq n}q_{\alpha}(m),
$ 
which is the most relaxed multiscale constraint. If no constant $\theta$ fulfills this constraint displayed above, then $\mathcal{I}(I)=0$ for all $I$ containing $[i/n,j/n)$ and thus the indices $1,\dots,i-1$ can be excluded from further consideration (cf. lines 29--33 in \cref{alg: MUSCLE}). 

\emph{Final dynamic program.}
The final procedure is as follows. Let $\mathcal{A}_0 =\{0\}$ and $\mathcal{B}_0 = \{1,\dots,n\}$ and for $k = 1\dots n$,
\begin{align*}
    r_k &\coloneqq \left\{j: \exists i \in \mathcal{A}_{k-1}, \theta\in \mathbb{R} \text{ s.t. }  T_{[\frac{i}{n},\frac{j}{n}]}^0(Z,\theta)\leq \underset{m}{\max}q_{\alpha}(m) 
    \right\},\\
    \mathcal{A}_k & \coloneqq \{j\in \mathcal{B}_{k-1} \cap [1,r_k]:\,\mathcal{I}([\tfrac{i}{n},\tfrac{j}{n}))=1 \text{ for some } i \in \mathcal{A}_{k-1}\},\\
    \mathcal{B}_k & \coloneqq \mathcal{B}_{k-1}\setminus \mathcal{A}_k.
\end{align*}
That is, $r_k$ serves as an upper bound of the indices of subproblems that can be solved with $k$ change points, and $\mathcal{A}_k$ contains all possible candidates of $\hat{\tau}_k$. If $k^*$ is the index such that $n\in \mathcal{A}_{k^*}$, then the  number of change points estimated by MUSCLE is $\widehat{K} = k^*-1$. 

\emph{Complexity analysis.} 
We first preprocess data $(Z_1,\dots,Z_n)$ by constructing a wavelet tree, which requires $O(n\log n)$ steps and $O(n\log n)$ memory (line 1 in \cref{alg: MUSCLE}). Then, by \cref{prop: I}, $\mathcal{I}([i/n,j/n))$ can be computed in $O((j-i)^2\log n)$ steps (lines 13--18 in \cref{alg: MUSCLE}). Thus, the overall complexity for computing $\widehat{K}$ is, up to a constant factor,  
\begin{equation}\label{e: complexity}
   \sum_{k=0}^{\widehat{K}} \sum_{i\in\mathcal{A}_k } \sum_{j = i+1}^{r_{k+1}}(j-i)^2 \log n
    \lesssim \log n\sum_{k=0}^{\widehat{K}} \abs{\mathcal{A}_k}\left(r_{k+1}-\min\mathcal{A}_k
    \right)^3 
    \lesssim n \log n\cdot \underset{0 \leq k \leq \widehat{K}}{\max}\left(r_{k+1}-\min\mathcal{A}_k
    \right)^3.
\end{equation}
The MUSCLE $\hat{f}$ can be computed long with $\widehat{K}$ in \eqref{e: MUSCLE}. When checking each index in the dynamic program, we need only $O(1)$ additional steps for $\hat{f}$. Thus, MUSCLE has a computation complexity as in \eqref{e: complexity}. The term $\max_{0 \leq k \leq \widehat{K}}\left(r_{k+1}-\min\mathcal{A}_k \right)^3$ in \eqref{e: complexity} depends on the underlying signal and the noise. For instance, if the signal consists of segments of bounded lengths and the noise is mild, then it happens most likely that $\widehat{K} = O(n)$ and also $\max_{0 \leq k \leq \widehat{K}}\left(r_{k+1}-\min\mathcal{A}_k \right)^3 = O(1)$. In this scenario, MUSCLE attains its best computational complexity of  $O(n\log n)$. In general, as $r_k\le n$, it holds always that $\max_{0 \leq k \leq \widehat{K}}\left(r_{k+1}-\min\mathcal{A}_k \right)^3 \le n^3$. By \eqref{e: complexity} we obtain the asserted upper bound $O(n^4\log n)$ on the computational complexity, which is reduced to $O(n^3\log n)$ if only intervals of dyadic lengths are considered.  The overall memory complexity is dominated by the storage of the wavelet tree, which is  $O(n\log n)$.

\end{appendix}
\end{document}